%% file: main.tex
\documentclass{article}
\providecommand{\main}{.}

\usepackage{\main/preamble}

\ifpdf
\hypersetup{
  pdftitle={A Friendly Smoothed Analysis of the Simplex Method},
  pdfauthor={D. Dadush and S. Huiberts}
}
\fi

\title{A Friendly Smoothed Analysis of the Simplex Method\thanks{Submitted to the editors June 2018}}

\author{Daniel Dadush \thanks{Centrum Wiskunde \& Informatica, The Netherlands. Supported by NWO Veni grant 639.071.510. Email: \href{mailto:dadush@cwi.nl}{dadush@cwi.nl}.}
\and Sophie Huiberts \thanks{Centrum Wiskunde \& Informatica, The Netherlands. Email: \href{mailto:sophie.huiberts@gmail.com}{sophie.huiberts@gmail.com}.}}

\begin{document}

\noteswarning

\maketitle

\begin{abstract}
Explaining the excellent practical performance of the simplex method for linear
programming has been a major topic of research for over 50 years. One of the
most successful frameworks for understanding the simplex method was given by
Spielman and Teng (JACM `04), who developed the notion of smoothed analysis.
Starting from an arbitrary linear program (LP) with $d$ variables and $n$
constraints, Spielman and Teng analyzed the expected runtime over random
perturbations of the LP, known as the smoothed LP, where variance $\sigma^2$
Gaussian noise is added to the LP data. In particular, they gave a two-stage
shadow vertex simplex algorithm which uses an expected $\widetilde{O}(d^{55}
n^{86} \sigma^{-30} + d^{70}n^{86})$ number of simplex pivots to solve the
smoothed LP. Their analysis and runtime was substantially improved by Deshpande
and Spielman (FOCS `05) and later Vershynin (SICOMP `09). The fastest current
algorithm, due to Vershynin, solves the smoothed LP using an expected
$O\big(\log^2 n \cdot \log\log n \cdot (d^3\sigma^{-4} + d^5\log^2 n + d^9\log^4
d)\big)$ number of pivots, improving the dependence on $n$ from polynomial to
poly-logarithmic.

While the original proof of Spielman and Teng has now been substantially
simplified, the resulting analyses are still quite long and complex and the
parameter dependencies far from optimal. In this work, we make substantial
progress on this front, providing an improved and simpler analysis of shadow
simplex methods, where our algorithm requires an expected
\[
O(d^2 \sqrt{\log n} ~ \sigma^{-2} + d^3 \log^{3/2} n)
\]
number of simplex pivots. We obtain our results via an improved \emph{shadow
bound}, key to earlier analyses as well, combined with improvements on
algorithmic techniques of Vershynin. As an added bonus, our analysis is
completely \emph{modular} and applies to a range of perturbations, which, aside
from Gaussians, also includes Laplace perturbations.  
\end{abstract}



\input{introduction}

\input{preliminaries}

\input{shadowbounds}
\input{algorithms}

\section{Conclusions and Open Problems}
We have given a substantially simplified and improved shadow bound and used it
to derive a faster simplex method. We are hopeful that our modular approach to
the shadow bound will help spur the development of a more robust smoothed
analysis of the simplex method, in particular, one that can deal with a much wider class
of perturbations such as those coming from bounded distributions.

There remains a gap between upper and lower bounds on the smoothed shadow size.
The lower bound of $\Omega(d^{3/2}\sqrt{\log n})$ by Borgwardt~\cite{Borgwardt87}
does not depend on $\sigma$ and is only proven for $n \to \infty$,
and the lower bound of just over $\Omega(\min(n,\frac{1}{\sqrt \sigma}))$ by Devillers, Glisse, Goaoc, and Thomasse~\cite{devillers2016smoothed}
is only proven to hold for $d=2$. Nonetheless, both of these lower bounds
are significantly lower than our upper bound of $O(d^2\sqrt{\log n}~\sigma^{-2}
+ d^{2.5}(\sigma\sqrt{\log n} + \log n))$.

Lastly, there are no known diameter bounds for smoothed polytopes other than
the general results mentioned above.

\subsection{Acknowledgment} The authors would like to thank the Simons Institute
program on ``Bridging Continuous and Discrete Optimization'' where some of this
work was done.

\bibliographystyle{./halpha}
\bibliography{references}

\end{document}

%% file: introduction.tex
\section{Introduction}

The simplex method for linear programming (LP) is one of the most important
algorithms of the 20th century. Invented by Dantzig in
1947~\cite{report/Dantzig48,col/aapa/Dantzig51}, it remains to this day one of
the fastest methods for solving LPs in practice. The simplex method is not one
algorithm however, but a class of LP algorithms, each differing in the choice of
\emph{pivot rule}. At a high level, the simplex method moves from vertex to
vertex along edges of the feasible polyhedron, where the pivot rule decides which
edges to cross, until an optimal vertex or unbounded ray is found.  Important
examples include Dantzig's most negative reduced cost~\cite{col/aapa/Dantzig51},
the Gass and Saaty parametric objective~\cite{jour/nrlq/GS55} and Goldfarb's
steepest edge~\cite{jour/sparsemc/Goldfarb76} method. We note that for solving
LPs in the context of branch \& bound and cutting plane methods for integer
programming, where the successive LPs are ``close together'', the dual steepest
edge method~\cite{jour/mapr/FG92} is the \emph{dominant} algorithm in
practice~\cite{col/smo/BFGRW00,jour/docmath/Bixby12}, due its observed ability
to quickly re-optimize.

The continued success of the simplex method in practice is remarkable for two
reasons. Firstly, there is no known polynomial time simplex method for LP.
Indeed, there are exponential examples for almost every major pivot rule
starting with constructions based on \emph{deformed products}~\cite{report/KM70,
jour/dm/Jeroslow73,col/pc/AC78,jour/dam/GS79,jour/mapr/Murty80,report/Goldfarb83,jour/cm/AZ98},
such as the Klee-Minty cube~\cite{report/KM70}, which defeat most classical
pivot rules, and more recently based on Markov decision
processes (MDP)~\cite{conf/stoc/FHZ11,conf/ipco/Friedmann11}, which notably defeat
randomized and history dependent pivot rules. Furthermore, for an LP with $d$
variables and $n$ constraints, the fastest provable (randomized) simplex method requires
$2^{O(\sqrt{d\log(2 +(n-d)/d)})}$
pivots~\cite{conf/stoc/Kalai92,jour/algorithmica/MSW96,conf/stoc/HZ15}, while
the observed practical behavior is linear $O(d+n)$~\cite{jour/ms/Shamir87}.
Secondly, it remains
the most popular way to solve LPs despite the tremendous progress for polynomial
time methods~\cite{jour/dansssr/Khachiyan79}, most notably, interior point
methods~\cite{jour/comb/Karmarkar84,jour/mapr/Renegar88,jour/sjopt/Mehrotra92,conf/focs/LS14}.
How can we explain the simplex method's excellent practical performance?

This question has fascinated researchers for decades. An immediate question is
how does one model instances in ``practice'', or at least instances where
simplex should perform well? The research on this subject has broadly speaking
followed three different lines: the analysis of average case LP models, where
natural distributions of LPs are studied, the smoothed analysis of
arbitrary LPs, where small random perturbations are added to the LP data, and
work on structured LPs, such as totally unimodular systems and MDPs (Markov Decision Processes). We review
the major results for the first two lines in the next section, as they are the
most relevant to the present work, and defer additional discussion to the
related work section. To formalize the model, we consider LPs in $d$ variables and $n$
constraints of the following form:
\begin{equation} \begin{split}
\max &~ \vec{c}^\T \vec{x} \hspace{3em} \\
     &~ \vec A \vec x \leq \vec b .
\end{split}\end{equation}
We denote the feasible polyhedron $\vec A \vec x \leq \vec b$ by $P$. We now introduce
relevant details for the simplex methods of interest to this work.

\paragraph{Parametric Simplex Algorithms}

While a variety of pivot rules have been studied, the most successfully analyzed in
theory are the so-called parametric simplex methods, due to the useful geometric
characterization of the paths they follow. The first such method, and the main
one used in the context of smoothed analysis, is the parametric objective method
of Gass and Saaty~\cite{jour/nrlq/GS55}, dubbed the \emph{shadow (vertex)
simplex} method by Borgwardt~\cite{thesis/Borgwardt77}.  Starting at a
\emph{known} vertex $\vec v$ of $P$ maximizing an objective $\vec c'$, the
parametric objective method computes the path corresponding to the sequence of
maximizers for the objectives obtained by interpolating $\vec c' \rightarrow
\vec c$~\footnote{This path is well-defined under mild non-degeneracy
assumptions}. The name shadow vertex method is derived from the fact that the
visited vertices are in correspondence with those on the projection of $P$ onto
$W := \linsp(\vec c, \vec c')$, the 2D convex polygon known as the shadow of
$P$ on $W$ (see figure~\ref{fig:shadow} for an illustration). In particular, the number
of vertices traversed by the method is bounded by the number of vertices of the
projection, known as the \emph{size of the shadow}.

An obvious problem, as with most simplex methods, is how to initialize the
method at a feasible vertex if one exists. This is generally referred to as the
Phase I problem, where Phase II then corresponds to finding an optimal solution.
A common Phase I adds artificial variable(s) to make feasibility trivial and
applies the simplex method to drive them to zero.

A more general method, popular in the context of average case analysis, is the
self-dual parametric simplex method of Dantzig~\cite{Dantzig59}. In this method,
one \emph{simultaneously} interpolates the objectives $\vec{c}' \rightarrow
\vec{c}$ and right hand sides $\vec{b}' \rightarrow \vec{b}$ which has the
effect of combining Phase I and II. Here $\vec{c}'$ and $\vec{b}'$ are chosen to
induce a \emph{known} initial maximizer. While the polyhedron is no longer
fixed, the breakpoints in the path of maximizers (now a piecewise linear curve)
can be computed via certain primal and dual pivots. This procedure was in fact
generalized by Lemke~\cite{jour/ms/Lemke65} to solve linear complementarity
problems. We note that the self dual method can roughly speaking be simulated in a
higher dimensional space by adding an interpolation variable $\lambda$,
i.e.~$\vec A \vec x \leq \lambda \vec b + (1-\lambda) \vec b', 0 \leq \lambda \leq
1$, which has been the principal approach in smoothed analysis.

\subsection{Prior Work}

Here we present the main works in both average case and smoothed analysis which
inform our main results, presented in the next section. A common theme in these
works, which all study parametric simplex methods, is to first obtain a bound on
the expected parametric path length, with respect to some distribution on
interpolations and LPs, and then find a way to use the bounds algorithmically.
This second step can be non-obvious, as it is often the case that one cannot
directly find a starting vertex on the path in question.  We now present the
main random LP models that have been studied, presenting path length bounds and
algorithms. Lastly, as our results are in the smoothed analysis setting, we
explain the high level strategies used to prove smoothed (shadow) path bounds.

\paragraph{Average case Models} The first model, introduced in the seminal work of
Borgwardt \cite{thesis/Borgwardt77,jour/zor/Borgwardt82,Borgwardt87,jour/mor/Borgwardt99},
examined LPs of the form $\max \vec c^\T \vec x, \vec{A} \vec x \leq \vec 1$,
possibly with $\vec x \geq \vec 0$ constraints (note that this model is always
feasible at $\vec 0$), where the rows of $\vec A \in \R^{n \times d}$ are drawn
i.i.d.~from a rotationally symmetric distribution (RSD) and $\vec c \in \R^d
\setminus \{\vec 0\}$ is fixed and non-zero. Borgwardt proved tight bounds on the
expected shadow size of the feasible polyhedron when projected onto any fixed
plane. For general RSD, he proved a sharp $\Theta(d^2
n^{1/(d-1)})$~\cite{Borgwardt87,jour/mor/Borgwardt99} bound, tight for rows drawn
uniformly from the sphere, and for Gaussians a sharp $\Theta(d^{1.5}\sqrt{\log
n})$ bound~\cite{Borgwardt87}, though this last bound is only known
to hold asymptotically as $n \rightarrow \infty$ (i.e.,~very large compared to $d$).
On the algorithmic side, Borgwardt~\cite{jour/zor/Borgwardt82} gave a
\emph{dimension by dimension} (DD)
algorithm which optimizes over such polytopes
by traversing $d-2$ different shadow vertex paths. The DD algorithm proceeds by
iteratively solving the restrictions $\max \vec c^\T \vec x, \vec A \vec x \leq
1, x_i = 0, i \in \set{k+1,\dots,d}$, for $k \geq 2$, which are all of RSD type.

For the next class, Smale~\cite{jour/mapr/Smale83} analyzed the standard self dual method
for LPs where $\vec A$ and $(\vec c, \vec b)$ are chosen from independent RSM
distributions, where Megiddo \cite{jour/mapr/Megiddo86} gave the best known
bound of $f(\min \set{d,n})$ iterations, for some exponentially large function
$f$. Adler~\cite{report/Adler83} and Haimovich~\cite{report/Haimovich83}
examined a much weaker model where the data is fixed, but where the signs of all
the inequalities, including non-negativity constraints, are flipped uniformly at
random. Using the combinatorics of hyperplane arrangements, they achieved a
remarkable bound of $O(\min \set{d,n})$ for the expected length of parametric
paths. These results were made algorithmic shortly thereafter
~\cite{jour/mapr/Todd86,jour/jacm/AM85,jour/jc/AKS87}, where it was shown
that a lexicographic version of the parametric self dual simplex method\footnote{These
works use seemingly different algorithms, though they were shown to be
equivalent to a lexicographic self-dual simplex method by
Meggiddo~\cite{jour/mmor/Megiddo85}.} requires $\Theta(\min \set{d,n}^2)$
iterations, where tightness was established in \cite{jour/jacm/AM85}. While
these results are impressive, a notable criticism of the symmetry model is that
it results in infeasible LPs almost surely once $n$ is a bit larger than $d$.

\paragraph{Smoothed LP Models} The \emph{smoothed} analysis framework,
introduced in the breakthrough work of Spielman and Teng~\cite{jour/jacm/ST04},
helps explain the performance of algorithms whose worst-case examples are in
essence \emph{pathological}, i.e.,~which arise from very brittle structures in
instance data. To get rid of these structures, the idea is to add a small amount
of noise to the data, quantified by a parameter $\sigma$, where the general goal
is then to prove an expected runtime bound over any \emph{smoothed instance}
that scales inverse polynomially with $\sigma$. Beyond the simplex method,
smoothed analysis has been successfully
applied to many other algorithms such as interior point
methods~\cite{jour/mapr/ST03}, Gaussian elimination~\cite{jour/siamjmaa/SST06},
Lloyd's $k$-means algorithm~\cite{jour/jacm/AMR11}, the $2$-OPT heuristic for the
TSP~\cite{jour/algorithmica/ERV14}, and much more.

The smoothed LP model, introduced by~\cite{jour/jacm/ST04}, starts with any base LP
\begin{equation*}
\max \vec c^\T \vec x,~ \vecb A \vec x \leq \vecb b, \hspace{5em}
\text{(Base LP)}
\end{equation*}
$\vecb A \in \R^{n \times d}$, $\vecb b \in \R^n$, $\vec c \in \R^d \setminus
\set{\vec 0}$, where the rows of $(\vecb A, \vecb b)$ are normalized to have $\ell_2$
norm at most $1$.  From the base LP, we generate the smoothed LP by adding
Gaussian perturbations to both the constraint matrix $\vecb A$ and the right hand
side $\vecb b$.  Precisely, the data of the smoothed LP is
\begin{equation*}
\vec A = \vecb A + \hat{\vec A},~ \vec b = \vecb b + \hat{\vec b},~ \vec c,
\hspace{5em} \text{(Smoothed LP Data)}
\end{equation*}
where the perturbations $\hat{\vec A}$,$\hat{\vec b}$ have i.i.d.~mean $0$,
variance $\sigma^2$ Gaussian entries. Note that the objective is not
perturbed in this model, though we require that $\vec c \neq \vec 0$. An LP
algorithm is said to have polynomial smoothed complexity if for any base LP data
$\vecb A,\vecb b,\vec c$ as above, we have
\begin{equation*}
\E_{\hat{\vec A},\hat{\vec b}}[T(\vec A,\vec b,\vec c)] = \poly(n,d,1/\sigma),
\hspace{5em} \text{(Smoothed Complexity)}
\end{equation*}
where $T(\vec A,\vec b,\vec c)$ is the runtime of the algorithm on a given
smoothed instance. Crucially, this complexity measure allows for an inverse
polynomial dependence on $\sigma$, the perturbation size. Focusing on
the simplex method, $T$ will measure the number of simplex pivots used
by the algorithm as a proxy for the runtime.

Spielman and Teng~\cite{jour/jacm/ST04} proved the first polynomial smoothed
complexity bound for the simplex method. In particular, they gave a two phase
shadow vertex method which uses an expected $\widetilde{O}(d^{55} n^{86}
\sigma^{-30} + d^{70}n^{86})$ number of pivots. This bound was substantially
improved by Deshpande and Spielman~\cite{conf/focs/DS05} and
Vershynin~\cite{jour/siamjc/Vershynin09}, where Vershynin gave the fastest such
method requiring an expected \[O\big(\log^2 n \cdot \log\log n \cdot
(d^3\sigma^{-4} + d^5\log^2 n + d^9\log^4 d)\big)\] number of pivots.

In all these works, the complexity of the algorithms is reduced in a black box
manner to a shadow bound for \emph{smoothed unit LPs}. In particular, a smoothed
unit LP has a base system $\vecb A \vec x \leq \vec 1$, where $\vecb A$ has
row norms at most $1$, and smoothing is performed only to $\vecb A$. Here the
goal is to obtain a bound on the expected shadow size with respect to any fixed
plane. Note that if $\vecb A$ is the zero matrix, then this is exactly
Borgwardt's Gaussian model, where he achieved the asymptotically tight bound of
$\Theta(d^{1.5} \sqrt{\log n})$ as $n\to\infty$~\cite{Borgwardt87}. For smoothed unit LPs, Spielman and
Teng~\cite{jour/jacm/ST04} gave the first bound of $O(d^3 n \sigma^{-6} + d^6 n
\log^3 n)$. Deshpande and Spielman~\cite{conf/focs/DS05} derived a bound of $O(d
n^2\log n \sigma^{-2} + d^2 n^2 \log^2n)$, substantially improving the dependence
on $\sigma$ while doubling the dependence on $n$. Lastly, Vershynin~\cite{jour/siamjc/Vershynin09} achieved a
bound of $O(d^3 \sigma^{-4} + d^5 \log^2 n)$, dramatically improving the
dependence on $n$ to poly-logarithmic, though still with a worse dependence on
$\sigma$ than~\cite{conf/focs/DS05}.

Before discussing the high level ideas for how these bounds are proved, we
overview how they are used algorithmically. In this context,
\cite{jour/jacm/ST04} and \cite{jour/siamjc/Vershynin09} provide two different
reductions to the unit LP analysis, each via an interpolation method.
Spielman and Teng first solve the smoothed LP with respect to an artificial
``somewhat uniform'' right hand side $\vec b'$, constructed to force a randomly
chosen basis of $\vec A$ to yield a vertex of the artificial system. From here
they use the shadow vertex method to compute a maximizer for right hand side $\vec b'$,
and continue via interpolation to derive an optimal solution for $\vec b$.  Here
the analysis is quite challenging, since in both steps the LPs are not quite
smoothed unit LPs and the used shadow planes correlate with the perturbations.
To circumvent these issues, Vershynin uses a \emph{random vertex} (RV)
algorithm, which starts with $\vec b' = \vec{1}$ (i.e.~a unit LP) and adds a random additional
set of $d$ inequalities to the system to induce an ``uncorrelated known
vertex''. From this random vertex, he proceeds similarly to Spielman and Teng,
but now at every step the LP is of smoothed unit type and the used shadow planes
are (almost) independent of the perturbations. In Vershynin's approach, the main
hurdle was to give a simple shadow vertex algorithm to solve unit LPs, which
correspond to the Phase 1 problem. An extremely simple method for this was in
fact already given in the early work of Borgwardt~\cite{Borgwardt87}, namely,
the dimension by dimension (DD) algorithm. The application of the DD algorithm
in the smoothed analysis context was however only discovered much
later by Schnalzger~\cite{thesis/Schnalzger14}. As it is both simple and not
widely known, we will describe the DD algorithm and its analysis in
Section~\ref{sec:algorithms}.

We note that beyond the above model, smoothed analysis techniques have been used
to analyze the simplex method in other interesting settings.
In~\cite{jour/siamjc/BCMRR15}, the successive shortest path algorithm for
min-cost flow, which is a shadow vertex algorithm, was shown to be efficient
when only the objective (i.e.~edge costs) is perturbed.
In~\cite{conf/stoc/KS06}, Kelner and Spielman used smoothed analysis techniques
to give a ``simplex like'' algorithm which solves arbitrary LPs in polynomial
time. Here they developed a technique to analyze the expected shadow size when
only the right hand side of an LP is perturbed.

\paragraph{Shadow Bounds for Smoothed Unit LPs}

Let $\vec{a}_1,\dots,\vec{a}_n \in \R^d, i \in [n]$, denote the rows of the
constraint matrix of the smoothed unit LP $\vec A \vec x \leq \vec 1$. The goal
is to bound the expected number of vertices in the projection of the feasible
polyhedron $P$ onto a fixed 2D plane $W$. As noticed by Borgwardt, by a
duality argument, this number of vertices is upper bounded by the number of edges
in the \emph{polar polygon} (see figure~\ref{fig:shadow} for an illustration).
Letting $Q := \conv(\vec{a}_1,\dots,\vec{a}_n)$, the convex hull of the rows,
the polar polygon can be expressed as $D := Q \cap W$.

We overview the different approaches used in~\cite{jour/jacm/ST04,conf/focs/DS05,
jour/siamjc/Vershynin09} to bound the number of edges of $D$. Let $\vec{u}_\theta$,
$\theta \in [0,2\pi]$, denote an angular parametrization of the unit circle in
$W$, and let $\vec{r}_\theta = \vec {u}_\theta \cdot \R_{\geq 0}$ denote the
corresponding ray. Spielman and Teng~\cite{jour/jacm/ST04} bounded the
probability that any two nearby rays $\vec {r}_\theta$ and $\vec {r}_{\theta+\eps}$
intersect different edges of $D$ by a linear function of $\eps$. Summing this
probability over any fine enough discretization of the circle upper bounds the
expected number of edges of $D$~\footnote{One must be a bit more careful
when $D$
does not contain the origin, but the details are similar.}. Their probability
bound proceeds in two steps, first they estimate the probability that the
Euclidean distance between the intersection of $\vec r_\theta$ with its corresponding
edge and the boundary of that edge is small (the \emph{distance lemma}), and
second they estimate the probability that angular distance is small compared to
Euclidean distance (the \emph{angle of incidence bound}).
Vershynin~\cite{jour/siamjc/Vershynin09} avoided the use of the angle of
incidence bound by measuring the intersection probabilities with respect to the
``best'' of three different viewpoints, i.e.~where the rays emanate from a
well-chosen set of three equally spaced viewpoints as opposed to just the
origin. This gave a much more efficient reduction to the distance lemma, and in
particular allowed Vershynin to reduce the dependence on $n$ from linear to
poly-logarithmic. Deshpande and Spielman~\cite{conf/focs/DS05} bounded different
probabilities to get their shadow bound. Namely, they bounded the probability
that nearby objectives $\vec{u}_\theta$ and $\vec{u}_{\theta+\eps}$ are
maximized at different vertices of $D$. The corresponding discretized sum over
the circle directly bounds the number of vertices of $D$, which is the same as
the number of edges.

\paragraph{Complexity in two dimensions} In two dimensions, the shadow size
reduces to the complexity of the convex hull. The smoothed complexity
of the convex hull was first considered by Damerow and
Sohler~\cite{damerow2004extreme},
who obtained an upper bound of $O((1+\sigma^{-d})\log^{3d/2} n)$
expected non-dominated points in $d$ variables,
bounding the smoothed complexity of a two-dimensional convex hull by
$O((1+\sigma^{-2})\log^{3/2} n)$ vertices.
Schnalzger~\cite{thesis/Schnalzger14} proved a complexity bound of
$O(\sigma^{-2} + \log n)$ based on the discretized sum
approach of Spielman and Teng~\cite{jour/jacm/ST04}.
The best general bounds were proved by Devillers, Glisse, Goaoc, and Thomasse
\cite{devillers2016smoothed}, who considered both i.i.d. Gaussian noise
and i.i.d. uniform noise drawn from a scaling of the sphere $\delta\bbS^{d-1}$.
In the Gaussian case, they prove an upper bound of
$O((1+\sigma^{-1})\sqrt{\log n})$ and lower bounds of $\Omega(n)$ for $0 \leq \sigma \leq \frac{1}{n^2}$,
$\Omega(\frac{1}{\sqrt{\sigma}}\sqrt[4]{\log(n\sqrt{\sigma})})$ for
$\frac{1}{n^2}\leq\sigma\leq\frac{1}{\sqrt{\log n}}$ and
$\Omega(\sqrt{\ln n})$ for $\frac{1}{\sqrt{\log n}} \leq \sigma$.

\FloatBarrier

\subsection{Results}

While the original proof of Spielman and Teng has now been substantially
simplified, the resulting analyses are still complex and the parameter
improvements have not been uniform. In this work, we give a ``best of all
worlds'' analysis, which is both much simpler and improves all prior parameter
dependencies. Our main contribution is a substantially improved shadow bound,
presented below.

{
\renewcommand{\arraystretch}{1.3}
\begin{table}
\label{fig:shadow-bound-hist}
\centering
\begin{tabular}{|l|l|l|}
\hline Works & Expected Number of Vertices & Model \\
\hline \cite{jour/mor/Borgwardt99} & $\Theta(d^2 n^{1/(d-1)})$ & RSD \\
\hline \cite{Borgwardt87} & $\Theta(d^{3/2} \sqrt{\log n} )$ & Gaussian, $n \rightarrow \infty$ \\
\hline \cite{jour/jacm/ST04} & $O(d^3 n \sigma^{-6} + d^6 n \log^3 n)$ & Smooth \\
\hline \cite{conf/focs/DS05} & $O(d n^2\log n ~ \sigma^{-2} + d^2 n^2 \log^2 n )$ & Smooth \\
\hline \cite{jour/siamjc/Vershynin09} & $O(d^3 \sigma^{-4} +  d^5 \log^2 n )$ & Smooth \\
\hline \textcolor{red}{This paper} & $\textcolor{red}{O(d^2 \sqrt{\log n} ~ \sigma^{-2} +
 d^{2.5}\log^{3/2} n (1+\sigma^{-1}))}$ & Smooth \\
\hline
\end{tabular}
\caption{Shadow Bounds. Logarithmic factors are simplified.
The Gaussian, $n \rightarrow \infty$ lower bound applies in
the smoothed model as well.}
\end{table}
}
We note that some of the bounds below (including ours) only hold for $d \geq 3$.
Recalling the models, the results in the Table~\ref{fig:shadow-bound-hist} bound
the expected number of vertices in the projection of a random polytope $\vec
A\vec x \leq \vec 1$, $\vec A \in \R^{n \times d}$, onto any fixed
$2$-dimensional plane. The models differ in the class of distributions examined
for $\vec A$. In the RSD model, the rows of $\vec A$ are distributed i.i.d.
according to an \emph{arbitrary} rotationally symmetric distribution. In the
Gaussian model, the rows of $\vec A$ are i.i.d. mean zero standard Gaussian
vectors.  Note that this is a special case of the RSD model.  In the smoothed
model, the rows of $\vec A$ are $d$-dimensional Gaussian random vectors with
standard deviation $\sigma$ centered at vectors of norm at most $1$, i.e.~the
expected matrix $\E[\vec A]$ has rows of $\ell_2$ norm at most $1$.  The $n
\rightarrow \infty$ in the table indicates that that bound only holds for $n$
large enough (compared to $d$).  The Gaussian, $n\to\infty$ model is a special
case of the smoothed analysis model, and hence the $\Omega(d^{3/2}\sqrt{\log
n})$ bound also holds in the smoothed model for $n$ big enough.

As can be seen, our new shadow bound yields a substantial improvement over
earlier smoothed bounds in all regimes of $\sigma$ and is also competitive in
the Gaussian model. For small $\sigma$, our bound
improves the dependence on $d$
from $d^3$ to $d^2$, achieves the same $\sigma^{-2}$ dependence
as~\cite{conf/focs/DS05}, and improves the dependence on $n$ to $\sqrt{\log n}$.
For $\sigma \geq 1$, our bound becomes $O(d^{2.5} \log^{3/2} n)$, which in
comparison to Borgwardt's optimal (asymptotic) Gaussian bound is only off by a
$d \log n $ factor. Furthermore, our proof is substantially simpler than
Borgwardt's and holds for all $n$ and $d$.
No interesting lower bounds for the small $\sigma$ regime are known for
$d\geq 3$, though the results of
\cite{devillers2016smoothed,devillers2015smoothed,Borgwardt87} suggest that
the correct lower bound might be much lower than current upper bounds.  We leave these questions as open problems.
{\renewcommand{\arraystretch}{1.3}
\begin{center}
\begin{table}
\centering
\label{fig:running times}
\begin{tabular}{|l|l|l|l|}
\hline Works & Expected Number of Pivots & Model & Algorithm \\
\hline \cite{Borgwardt87,thesis/Schnalzger14} &
$d\cdot\mbox{max shadow size}$
& Multiple\!\! & DD + Int. LP\!\! \\
\hline \cite{Borgwardt87,thesis/Hofner95,jour/mor/Borgwardt99} & $O(d^{2.5} n^{1/(d-1)})$ & \begin{tabular}{@{}c@{}}RSD, \\ $n\to\infty$\end{tabular} & DD \\
\hline \cite{jour/siamjc/Vershynin09} &
$O\big(\log^3 n \cdot (d^3\sigma^{-4}\! + d^5\log^2 n + d^9\log^4 d)\big)$\!\! & Smooth & RV + Int.LP \\
\hline \textcolor{red}{This paper} & $\textcolor{red}{O(d^2 \sqrt{\log n}
~ \sigma^{-2} + d^3 \log^{3/2}n)}$ & Smooth & \begin{tabular}{@{}c@{}}Symmetric RV\!\!\!\!\! \\ + Int. LP\end{tabular} \\
\hline
\end{tabular}
\caption{Runtime bounds. Logarithmic factors are simplified.}
\end{table}
\end{center}}

An interesting point of our analysis is that it is \emph{completely modular},
and that it gives bounds for perturbations other than Gaussians. In fact, in
our approach it is easier to obtain bounds for Laplace perturbations (see
Section~\ref{sec:shadow-bounds}) than for the Gaussian distribution. The range
of analyzable perturbations still remains limited however, our analysis doesn't
extend to bounded perturbations such as uniform $[-1/\sigma,1/\sigma]$ for
example. As is well known, LPs in practice tend to be sparse and hence don't
follow a Gaussian distribution (which yields a totally dense constraint matrix).
It is thus of considerable interest to understand the smoothed behavior of the
simplex method under wider classes of perturbations, such as perturbations with
restricted support.

From the algorithmic perspective, we describe the two phase interpolation
approach of Vershynin~\cite{jour/siamjc/Vershynin09}, which we instantiate using
two different Phase 1 algorithms to solve unit LPs. As a warmup,
we first describe Schnalzger's application of the dimension by dimension (DD)
algorithm~\cite{thesis/Schnalzger14}, as it yields the simplest known Phase 1
algorithm and is not widely known. Following this, we introduce a new variant of
Vershynin's random vertex (Symmetric RV) algorithm which induces an artificial
(degenerate) random vertex by adding $2d-2$ inequalities placed symmetrically
around a randomly chosen objective. The symmetry condition ensures that this
random vertex optimizes the chosen objective with probability $1$. Vershynin's
original approach added $d$ random inequalities, which only guaranteed that
the induced vertex optimized the chosen objective with $1-1/\poly(d)$
probability. Via a more careful analysis of the RV algorithm combined with the
additional guarantees ensured by our variant, we derive a substantially improved
complexity estimate. Specifically, our Symmetric RV algorithm runs in time
$O(d^2\sqrt{\log n}~\sigma^{-2} + d^3\log^{3/2} n)$, which is faster than both
the original RV algorithm and Borgwardt's dimension by dimension algorithm in
all parameter regimes. We defer further discussion of this to
section~\ref{sec:algorithms} of the paper.

\FloatBarrier

\subsection{Techniques: Improved Shadow Bound}
\label{sec:proof-sketch}

We now give a detailed sketch of the proof of our improved shadow bound.  Proofs
of all claims can be found in section~\ref{sec:shadow-bounds}. The outline of the
presentation is as follows. To begin, we explain our general edge counting
strategy, where we depart from the previously discussed analyses. In particular, we
adapt the approach of Kelner and Spielman (KS)~\cite{conf/stoc/KS06}, who analyzed a smoothing model where only the right-hand side is perturbed, to
the present setting. Following this, we present a parametrized shadow bound,
which applies to any class of perturbations for which the relevant
parameters are bounded. The main motivation of the abstraction in the
parametrized model is us to clearly identify the relevant properties of the
perturbations we need to obtain shadow bounds. Lastly, we give the high-level
idea of how we estimate the relevant quantities in the KS approach within the
parametrized model.

\paragraph{Edge Counting Strategy} Our goal is to compute a bound
on the expected number of edges in the polygon $Q \cap W$, where $W$ is the
two-dimensional shadow plane, $Q := \conv(\vec{a}_1,\dots,\vec{a}_n)$ and
$\vec{a}_1,\dots,\vec{a}_n \in \R^d$ are the smoothed constraints of a unit LP.
Recall that this is an upper bound on the shadow size.

In~\cite{conf/stoc/KS06}, Kelner and Spielman developed a very elegant and
useful alternative strategy to bound the expected number of edges, which can be
applied to many distributions over 2D convex polygons. Whereas they analyzed the
geometry of the primal shadow polygon, the projection of $P$ onto $W$,
we will instead work with the geometry of the polar polygon $Q \cap W$. The analysis begins with the
following elementary identity:

\begin{equation}
\label{eq:perim-to-edges1}
\E[{\rm perimeter}(Q \cap W)] = \E[\sum_{\vec e \in {\rm edges}(Q \cap W)} {\rm length}(\vec e)] \text{ .}
\end{equation}

Starting from the above identity, the approach first derives a good upper bound on the
perimeter and a lower bound on the right-hand side in terms of the number
of edges and the minimum edge length. The bound on the number of edges is then
derived as the ratio of the perimeter bound and the minimum edge length.

We focus first on the perimeter upper bound. Since $Q \cap W$ is convex, any containing
circle has larger perimeter. Furthermore, we clearly have $Q \cap W \subseteq
\pi_W(Q)$, where $\pi_W$ is the orthogonal projection onto $W$. Combining these
two observations, we derive the first useful inequalities:

\begin{equation}
\label{eq:perim-bound}
\E[{\rm perimeter}(Q \cap W)] \leq \E[2\pi \max_{\vec{x} \in Q \cap W} \norm{\vec{x}}] \leq
\E[2\pi \max_{i \in [n]} \|\pi_W(\vec a_i)\|] \text{ .}
\end{equation}

To extract the expected number of edges from the right hand side
of~\eqref{eq:perim-to-edges1}, we first note that every edge of $Q \cap W$ is derived
from a facet of $Q$ intersected with $W$ (see Figure~\ref{fig:shadow} for an
illustration). Assuming non-degeneracy, the possible facets of $Q$ are $F_I :=
\conv(\vec a_i : i \in I)$, where $I
\subseteq [n]$ is any subset of size $d$. Let $E_I$ denote the event that $F_I$
induces an edge of $Q \cap W$, more precisely, that $F_I$ is a facet of $Q$ and that
$F_I \cap W \neq \emptyset$. From here, we get that
\begin{equation}
\label{eq:perim-to-edges2}
\begin{split}
\E[\sum_{\vec e \in {\rm edges}(Q \cap W)} {\rm length}(\vec e)]
&= \sum_{\abs{I}=d} \E[{\rm length}(F_I \cap W) \mid E_I] \Pr[E_I] \\
&\geq \min_{\abs{I}=d} \E[{\rm length}(F_I \cap W) \mid E_I] \cdot \sum_{\abs{I}=d} \Pr[E_I] \\
&= \min_{\abs{I}=d} \E[{\rm length}(F_I \cap W) \mid E_I] \cdot \E[\abs{{\rm
edges}(Q \cap W)}]~.
\end{split}
\end{equation}

Combining~\eqref{eq:perim-to-edges1},~\eqref{eq:perim-bound},~\eqref{eq:perim-to-edges2},
we derive the following fundamental bound:
\begin{equation}
\label{eq:main-bound}
\E[\abs{{\rm edges}(Q \cap W)}] \leq \frac{\E[2\pi \max_{i \in [n]} \norm{\pi_W(\vec
a_i)}]}{\min_{\abs{I}=d} \E[{\rm length}(F_I \cap W) \mid E_I]}~.
\end{equation}

In the actual proof, we further restrict our attention to potential edges having
probability $\Pr[E_I] \geq 2\binom{n}{d}^{-1}$ of appearing, which helps
control how extreme the conditioning on $E_I$ can be. Note that the edges appearing with
probability smaller than $2\binom{n}{d}^{-1}$ contribute at most $2$ to the expectated
number of edges, and
hence can be ignored. Thus our task now directly reduces to showing that the
maximum perturbation is not too large on average, an easy condition, while
ensuring that the edges that are not too unlikely to appear are reasonably long on
average, the more difficult condition.

We note that applying the KS approach already improves the situation with
respect to the maximum perturbation size compared to earlier analyses,
as~\cite{jour/jacm/ST04,conf/focs/DS05,jour/siamjc/Vershynin09} all require a
bound to hold with high probability as opposed to on expectation. For this purpose, they
enforced the condition $1/\sigma \geq \sqrt{d \log n}$ (for Gaussian
perturbations), which we do not require here.

\paragraph{Bound for Parametrized Distributions}
We now present the parameters of the pertubation distributions we use to obtain
our bounds on the enumerator and denominator of \ref{eq:main-bound}. We also
discuss how these parameters behave for the Gaussian and Laplace distribution.

Let us now assume that $\vec{a}_1,\dots,\vec{a}_n \in \R^d$ are independently
distributed. As before we assume that the centers $\vecb{a}_i : =
\E[\vec{a}_i]$, $i \in [n]$, have norm at most $1$. We denote the perturbations
by $\vech{a}_i := \vec a_i-\vecb a_i$, $i \in [n]$. We will assume for
simplicity of the presentation that all the perturbations
$\vech{a}_1,\dots,\vech{a}_n$ are i.i.d. according to a distribution
with probability density $\mu$ (in general, they could each have a distinct
distribution).

At a high-level, the main properties we require of the distribution are
that it be smooth and that it have sufficiently strong tail bounds. We formalize these requirements via
the following $4$ parameters, where we let $\vec X \sim \mu$ below:

\begin{enumerate}
\item $\mu$ is an $L$-log-Lipschitz probability density function, that is, \\
$\abs{\log \mu(\vec x) - \log \mu(\vec y)} \leq L \norm{\vec
x-\vec y}$, $\forall \vec x,\vec y \in \R^d$.
\item The variance of $\vec X$, when restricted to any line $l \subset \R^d$, is at least $\tau^2$.
\item The cutoff radius $R_{n,d} > 0$ is such that $\Pr[\norm{\vec X} \geq R_{n,d}] \leq \frac{1}{d\binom{n}{d}}$.
\item The $n$-th deviation $r_n$ is such that, for all $\vec \theta \in \R^d$, $\norm{\vec \theta}=1$, and $\vec X_1,\dots,\vec X_n$
i.i.d., we have $\E[\max_{i \in [n]}
\abs{\langle \vec X_i, \vec \theta \rangle}] \leq r_n$.
\end{enumerate}
We refer the reader to subsection~\ref{sub:dist-params} for more formal
definitions of these parameters. We note that these parameters naturally arise
from our proof strategy and directly expose the relevant quantities for our
shadow bound.

The first two parameters are smoothness related while the last two relate to
tail bounds. Using these four parameters, we will derive appropriate bounds
for the enumerator and denominator in \eqref{eq:main-bound}.
Assuming the above parameter bounds for
$\vech{a}_1,\dots,\vech{a}_n$, our main ``plug and play'' bound on the expected shadow
size is as follows (see Theorem~\ref{thm:abstractedbound}):
\begin{equation}
\label{eq:main-shadow-bnd}
\E[\abs{{\rm edges}(\conv(\vec a_1,\dots,\vec a_n) \cap W)}] = O(\frac{d^{1.5} L}{\tau}(1+R_{n,d})(1+r_n)) \text{ .}
\end{equation}

We can use this parametrized bound to prove the shadow bound for Gaussian and
Laplace distributed noise. For the variance $\sigma^2$ Gaussian distribution in
$\R^d$, it is direct to verify that $\tau = \sigma$ for any line (since every
line restriction results in a 1D variance $\sigma^2$ Gaussian), and from
standard tail bounds that $R_{n,d} = O(\sigma \sqrt{d \log n})$ and $r_n =
O(\sigma \sqrt{\log n})$. The only parameter that cannot be bounded directly is
the log-Lipschitz parameter $L$, since $\norm{\vec{x}/\sigma}^2/2$, the log of
the Gaussian density, is quadratic. For Laplace distributed perturbations
however, this last difficulty is completely avoided. Here a comparably sized Laplace
perturbation (i.e.~same expected norm) has density proportional to
$e^{-(\sqrt{d}/\sigma) \norm{\vec x}}$ which is by definition log-Lipshitz
with $L=\sqrt{d}/\sigma$. The other parameters are somewhat worse, it can be
shown that $R_{n,d} = O(\sigma \sqrt{d} \log n)$, $r_n = O(\sigma \log n)$ and
$\tau \geq \sigma/\sqrt{d}$, where in particular $\tau$ is a $\sqrt{d}$-factor
smaller than the Gaussian. Thus, for Laplace perturbations our parametrized
bound applies directly and yields a bound of $O(d^{2.5}\sigma^{-2})$ is the
small $\sigma$ regime.

To apply our analysis to the Gaussian setting, we start with the fact, noted in
all prior analyses, that the Gaussian is locally smooth within any fixed radius.
In particular, within radius $R_{n,d}$ of the mean, the Gaussian density is
$O(\sqrt{d \log n}/\sigma)$-log-Lipschitz. As events that happen with
probability $\ll \binom{n}{d}^{-1}$ have little effect on the expected shadow
bound (recall that the shadow is always bounded by $\binom{n}{d}$), one can hope
to condition on each perturbation living inside the $R_{n,d}$ radius ball. This
is in fact the approach taken in the prior
analyses~\cite{jour/jacm/ST04,conf/focs/DS05,jour/siamjc/Vershynin09}. This
conditioning however does not ensure full log-Lipshitzness and causes problems
for points near the boundary. Furthermore, the conditioning may also decrease
line variances for lines near the boundary.

To understand why this is problematic, we note that the main role of the
smoothness parameters $L$ and $\tau$ is to ensure enough ``wiggle-room'' to
guarantee that edges induced by any fixed basis are long on expectation. Using
the above conditioning, it is clear that edges induced by facets whose
perturbations occur close to the $R_{n,d}$ boundary must be dealt with
carefully. To avoid such difficulties altogether, we leverage the local
log-Lipshitzness of the Gaussian in a ``smoother'' way. Instead of conditioning,
we simply replace the Gaussian density with a globally $O(\sqrt{d \log
n}/\sigma)$-log-Lipshitz density which has statistical distance $\ll
\binom{n}{d}^{-1}$ to the Gaussian (thus preserving the shadow bound) and also
yields nearly identical bounds for the other parameters. This distribution will
consist of an appropriate gluing of a Gaussian and Laplace density, which we
call the Laplace-Gaussian distribution (see section~\ref{sec:shadow-gaussian}
for details). Thus, by slightly modifying the distribution, we are able to use
our parametrized model to obtain shadow bounds for Gaussian perturbations in a
black box manner.



\paragraph{Bounding the Perimeter and Edge Length} We now briefly describe
how the perimeter and minimum edge length in \eqref{eq:main-bound} are bounded in our parametrized
perturbation model to obtain \eqref{eq:main-shadow-bnd}. As this is the most technical part of the analysis, we refer
the reader to the proofs in section~\ref{sec:shadow-bounds} and give only a
very rough discussion here. As above, we will assume that the perturbations
satisfy the bounds given by $L,\tau,R_{n,d},r_n$.

For the perimeter bound, we immediately derive the bound
\[
\E[\max_{i \in [n]} \norm{\pi_W(\vec{a}_i)}]
\leq 1 + \E[\max_{i \in [n]} \|\pi_W(\vech{a}_i)\|] \leq 1 + 2r_n
\]
by the triangle inequality. From here, we must bound the minimum expected edge
length, which requires the majority of the work. For this task, we provide a
clean analysis, which shares high-level similarities with the Spielman and Teng
distance lemma, though our task is simpler. Firstly, we only need to
show that an edge is large on average, whereas the distance lemma has the more
difficult task of proving that an edge is unlikely to be small. Second, our
conditioning is much milder.  Namely, the distance lemma conditions a facet
$F_I$ on intersecting a specified ray $\vec{r}_\theta$, whereas we only
condition $F_I$ on intersecting $W$.  This conditioning gives the edge much more
``wiggle room'', and is the main leverage we use to get the factor $d$
improvement.

Let us fix $F := F_{[d]} = \conv(\vec a_1,\dots,\vec a_d)$ as the potential
facet of interest, under the assumption that $E := E_{[d]}$, i.e.,
that $F$ induces an edge of $Q \cap W$, has probability at least $2\binom{n}{d}^{-1}$. Our
analysis of the edge length conditioned on $E$ proceeds as follows:

\begin{enumerate}
\item Show that if $F$ induces an edge, then under this conditioning $F$ has
\emph{small diameter} with good probability, namely its vertices are all at
distance at most $O(1 + R_{n,d})$ from each other
(Lemma~\ref{lem:boundedprojectedsimplex}). This uses the tailbound defining
$R_{n,d}$ and the fact that $E$ occurs with non-trivial probability.
\item Condition on $F$ being a facet of $Q$ by fixing its containing affine
hyperplane $H$ (Lemma~\ref{lem:conditiononhyperplane}). This is standard and
is achieved using a change of variables analyzed by Blaschke (see
section~\ref{sec:blaschke} for details).
\item Let $l := H \cap W$ denote the line which intersects $F$ to form an edge of
$Q \cap W$. Show that on average the longest chord of $F$ parallel to $l$ is long.
We achieve the bound $\Omega(\tau/\sqrt{d})$ (Lemma~\ref{lem:heightofsimplex})
using that the vertices of $F$
restricted to lines parallel to $l$ have variance at least $\tau^2$.
\item Show that on average $F$ is pierced by $l$ through a chord that is not
too much shorter than the longest one. Here we derive the final bound
on the expected edge length of
\[
\E[{\rm length}(F \cap W) \mid E] = \Omega((\tau/\sqrt{d}) \cdot 1/(d L(1 + R_{n,d}))) \text{ (Lemma~\ref{lem:convexcombinationspace})}
\]
using the fact that the distribution of the vertices is $L$-log-Lipschitz and that
$F$ has diameter $O(1+R_{n,d})$.
\end{enumerate}
This concludes the high-level discussion of the proof.

\subsection{Related work}

\paragraph{Structured Polytopes} An important line of work has been to
study LPs with good geometric or combinatorial properties. Much work has been
done to analyze primal and dual network simplex algorithms for fundamental
combinatorial problems on flow polyhedra such as bipartite
matching~\cite{jour/or/Hung83}, shortest
path~\cite{jour/networks/DGKK79,jour/or/GHK90}, maximum
flow~\cite{jour/mapr/GH90,jour/mapr/GGT91} and minimum cost
flow~\cite{report/Orlin84,jour/algorithmica/GH92,jour/mapr/OPT93}. Generalizing
on the purely combinatorial setting, LPs where the constraint matrix $\vec A \in \Z^{n
\times d}$ is totally unimodular (TU), i.e.~the determinant of any square
submatrix of $\vec A$ is in $\set{0,\pm 1}$, were analyzed by Dyer and
Frieze~\cite{jour/mapr/DF94}, who gave a random walk based simplex algorithm
which requires $\poly(d,n)$ pivots. Recently, an improved random walk approach was
given by Eisenbrand and Vempala~\cite{jour/mapr/EV17}, which works in the more
general setting where the subdeterminants are bounded in absolute value by
$\Delta$, who gave an $O(\poly(d,\Delta))$ bound on the number of Phase II pivots (note
that there is no dependence on $n$). Furthermore, randomized variants of the
shadow vertex algorithm were analyzed in this setting
by~\cite{conf/stacs/BGR15,jour/dcg/DH16}, where in
particular~\cite{jour/dcg/DH16} gave an expected $O(d^5 \Delta^2 \log(d \Delta))$
bound on the number of Phase I and II pivots.
Another interesting class of structured polytopes
comes from the LPs associated with Markov Decision Processes (MDP), where
simplex rules such as Dantzig's most negative reduced cost correspond to
variants of policy iteration. Ye~\cite{jour/mor/Ye11} gave polynomial bounds for
Dantzig's rule and Howard's policy iteration for MDPs with a fixed discount
rate, and Ye and Post~\cite{jour/mor/PY15} showed that Dantzig's rule converges
in strongly polynomial time for deterministic MDPs with variable discount rates.

\paragraph{Diameter Bounds} Another important line of research has been to
establish diameter bounds for polyhedra, namely to give upper bounds on the
shortest path length between any two vertices of a polyhedron as a function of
the dimension $d$ and the number of inequalities $n$. For any simplex method pivoting on
the vertices of a fixed polytope, the diameter is clearly a lower bound on the
worst-case number of pivots. The famous Hirsch conjecture from 1957,
posited that for polytopes (bounded polyhedra) the correct bound should be
$n-d$. This precise bound was recently disproven by
Santos~\cite{jour/am/Santos12}, who gave a $43$ dimensional counter-example,
improved to $20$ in~\cite{jour/plms/MSW15}, where the diameter is about $1.05(n-d)$ (these counter-examples can also be extended to infinite families).
However, the possibility of a polynomial (or even linear) bound is still left
open, and is known as the polynomial Hirsch conjecture. From this standpoint,
the best general results are the $O(2^d n)$ bound by
Barnette~\cite{jour/dm/Barnette74} and Larman~\cite{jour/plms/Larman70}, and the
quasi-polynomial bound of Kalai and
Kleitman~\cite{jour/bams/KK92}, recently refined by
Todd~\cite{jour/sjdm/Todd14} and Sukegawa~\cite{jour/dcg/sukegawa18} to $(n-d)^{\log_2 O(d/\log d)}$. As above, such bounds have been studied for
structured classes of polytopes. In particular, the diameter of polytopes with
bounded subdeterminants was studied by various
authors~\cite{jour/mapr/DF94,jour/dcg/BSEHN14,jour/dcg/DH16}, where the best
known bound of $O(d^3 \Delta^2 \log(d \Delta))$ was given
in~\cite{jour/dcg/DH16}. The diameters of other classes such as $0/1$
polytopes~\cite{jour/mapr/Naddef89}, transportation
polytopes~\cite{jour/mor/Balinski84,jour/combinatorica/BHS06,jour/jct/LKOS09,jour/mapr/BLF17}
and flag polytopes~\cite{jour/mor/AB14} have also been studied.

\subsection{Organization}

Section~\ref{sec:prelims} contains basic definitions and background material.
The proofs of our shadow bounds are given in section~\ref{sec:shadow-bounds},
\ref{sec:shadow-laplace}
and~\ref{sec:shadow-gaussian} respectively. The details regarding the
two phase shadow vertex algorithm we use, which rely in an almost
black box way on the shadow
bound, are presented in section~\ref{sec:algorithms}.

%% file: preliminaries.tex
\section{Preliminaries}
\label{sec:prelims}

\subsection{Notation and basic definitions}
\hspace{1em}\vspace{1em}

\begin{itemize}
 \item Vectors are printed in bold to contrast with scalars: $\vec x =
(x_1,\dots,x_d) \in \R^d$. The space $\R^d$ comes with a standard basis $\vec
e_1,\dots,\vec e_d$. We abbreviate $\vec 1 := (1,1,\dots,1)$ and $\vec 0 :=
(0,0,\dots,0)$. Vector inequalities are defined coordinate-wise: $\vec v \leq \vec w$ if and only if $v_i \leq w_i$ for all $i \leq d$.
 \item We abbreviate $[n] := \set{1,\dots,n}$ and $\binom{[n]}{d} = \set{I
\subset [n] \mid \abs{I} = d}$. For $a, b \in \R$ we have intervals $[a,b] =
\set{r \in \R : a \leq r \leq b}$ and $(a, b) = \set{r \in \R : a < r < b}$.
\item For $x > 0$, we define $\log x$ to be the logarithm base $e$ of $x$.
\item For a set $C \subseteq \R^n$, we denote its topological boundary by
$\partial C$.
\item For a set $A$, we use the notation $\mathbbm{1}[\vec x \in
A]$ to denote the indicator function of $A$, i.e., $\mathbbm{1}[\vec x \in A] = 1$
if $\vec x \in A$ and $0$ otherwise.
\item For $A, B \subset \R^d$ we write the Minkowski sum $A + B =
\set{\vec a + \vec b : \vec a \in A, \vec b \in B}$. For a vector $\vec v \in
\R^d$ we write $A + \vec v = A + \set{\vec v}$. For a set of scalars $S \subset
\R$ we write $\vec v \cdot S = \set{s\vec v : s \in S}$.
\item The inner product of $\vec x$ and $\vec y$ is written with as $\vec
x^{\mathsf{T}} \vec y = \sum_{i=1}^d x_iy_i$. We use the $\ell_2$-norm
$\norm{\vec x}_2 = \sqrt{\inner{\vec x}{\vec x}}$ and the $\ell_1$-norm
$\norm{\vec x}_1 = \sum_{i=1}^d \abs{x_i}$. Every norm without subscript is the
$\ell_2$-norm. The unit sphere in $\R^d$ is $\bbS^{d-1} = \set{\vec x \in \R^d : \norm{\vec x} = 1}$ and the unit ball is $\calB_2^d = \set{\vec x \in \R^d : \norm{\vec x} \leq 1}$.
 \item A set $V + \vec p$ is an affine subspace if $V \subset \R^d$ is a linear subspace. If $S \subset \R^d$ then the affine hull $\aff(S)$ is the smallest affine subspace containing $S$. We say $\dim(S) = k$ if $\dim(\aff(S)) = k$.
 \item For any linear or affine subspace $V \subset \R^d$ the orthogonal projection onto $V$ is denoted by $\pi_V$.
 \item For a linear subspace $V \subseteq \R^d$, we denote its orthogonal complement by $V^\perp = \set{\vec x \in \R^d : \inner{\vec v}{\vec x} = 0,\, \forall\, \vec v \in V}$. For $\vec v \in \R^d$ we abbreviate $\vec v^\perp := \linsp(\vec v)^\perp$.
 \item We write $\vol_k(S)$ for the $k$-dimensional volume of $S$. The $1$-dimensional volume of a line segment $l$ will also be written as $\length(l)$.
 \item We say vectors $\vec a_1,\dots,\vec a_k$ in $\R^d$ are \emph{affinely independent} if there is no $(k-2)$-dimensional affine subspace containing all of $\vec a_1,\dots,\vec a_k$. Algebraically, $\vec a_1,\dots,\vec a_k$ are affinely independent if the system $\sum_{i\leq k} \lambda_i \vec a_i = \vec 0, \sum_{i\leq k} \lambda_i = 0$ has no non-trivial solution.
 \item For $\vec A \in \R^{n\times d}$ a matrix and $B \subset [n]$ we write $\vec A_B \in \R^{\abs{B}\times d}$ for the submatrix of $\vec A$ consisting of the rows indexed in $B$, and for $\vec b \in \R^n$ we write $\vec b_B$ for the restriction of $\vec b$ to the coordinates indexed in $B$.
\end{itemize}

\subsection{Convexity}
A polyhedron $P$ is of the form $P = \set{\vec x \in \R^d : \vec A\vec x \leq
\vec b}$ for $\vec A \in \R^{n\times d}, \vec b \in \R^n$. A face $F \subseteq
P$ is a convex subset such that  if $\vec x,\vec y \in P$ and for $\lambda \in
(0,1)$  $\lambda \vec x + (1-\lambda)\vec y \in F$, then $\vec x, \vec y \in F$.
In particular, a set $F$ is a face of the polyhedron $P$ if and only if there
exists $I \subset [n]$ such that $F$ coincides with $P$ intersected with $\vec
a_i^\T \vec x = b_i, \forall i \in I$. A zero-dimensional face is called a
vertex, one-dimensional face is called an edge, and a $\dim(P)-1$-dimensional
face is called a facet. We use the notation $\vertices(P)$ to denote the set of
vertices of $P$ and $\edges(P)$ for the set of edges of $P$.

A set $S \subset \R^d$ is convex if for all $\vec x,\vec y \in S, \lambda
\in [0,1]$ we have $\lambda \vec x + (1-\lambda)\vec y \in S$. We write
$\conv(S)$ to denote the convex hull of $S$, which is the intersection of all
convex sets $T \supset S$. In a $d$-dimensional vector space, the convex hull equals

\[
\conv(S) = \set{\sum_{i=1}^{d+1} \lambda_i \vec s_i : \lambda_1,\dots,\lambda_{d+1} \geq 0, \sum_{i=1}^{d+1} \lambda_i = 1, \vec s_1,\dots,\vec s_{d+1} \in S}.
\]
For $\vec x, \vec y \in \R^d$ the line segment between $\vec x$ and $\vec y$ is
denoted $[\vec x, \vec y] = \conv(\set{\vec x,\vec y})$.

We will need the following classical comparison inequality for surface areas of
convex sets (see for example \cite[Chapter 7]{bonnesen1987theory}).

\begin{lemma}[Monotonicity of Surface Area] \label{lem:monotonicity}
If $K_1 \subseteq K_2 \subset \R^d$ are compact full-dimensional convex sets,
then $\vol_{d-1}(\partial K_1) \leq \vol_{d-1}(\partial K_2)$.
\end{lemma}

\subsection{Random Variables}~

For a random variable $X \in \R$, we denote its expectation (mean) by $\E[X]$
and its variance by $\Var(X) := \E[(X-\E[X])^2]$. For a random vector $\vec X
\in \R^n$, we define its expectation (mean) $\E[\vec X] :=
(\E[X_1], \dots, \E[X_n])$ and its variance (expected squared distance from the
mean) $\Var(\vec X) := \E[\norm{\vec X - \E[\vec X]}^2]$.

For jointly distributed $X \in \Omega_1, Y \in \Omega_2$, we will often minimize
the expectation of $X$ over instantiations $y \in A \subset \Omega_2$. For this,
we use the notation
\begin{displaymath}
\min_{Y \in A} \E[X\mid Y] := \min_{y \in A}\E[X\mid Y=y].
\end{displaymath}
If $\mu$ is a probability density function, we write $x \sim \mu$ to
denote that $x$ is a random variable distributed with probability density
$\mu$.

For an event $E \subseteq \Omega$ in a measure space, we write $E^c :=
\Omega \setminus E$ to denote its complement.

\subsubsection{Gaussian distribution}
\label{sub:gaussian}
\begin{definition}
The \emph{Gaussian distribution}\index{Gaussian distribution} or \emph{normal distribution}\index{normal distribution|see {Gaussian distribution}} $N_d(\vecb a, \sigma)$\index{$N_d(\vecb a, \sigma)$|see {Gaussian distribution}} in $d$ variables with mean $\vecb a$ and standard deviation $\sigma$ has density
$(2\pi)^{-d/2} e^{-\norm{\vec x - \bar{\vec a}}^2/(2\sigma^2)}$. We abbreviate $N_d(\sigma) = N_d(\vec 0, \sigma)$.
\end{definition}

Important facts about the Gaussian distribution include:
\begin{itemize}
 \item Given a $k$-dimensional affine subspace $W \subseteq \R^d$, if $\vec X$ is $N_d(\vecb a, \sigma)$-distributed then both the orthogonal projection $\pi_W(\vec X)$ and the restriction of $\vec X$ to $W$ are $N_k(\pi_W(\vecb a), \sigma)$-distributed in $W$.
 \item For $\vec X \sim N_d(\vecb a,\sigma)$ we have $\E[\vec X] = \vecb a$ and $\E[(\inner{(\vec X - \vecb a)}{\vec\theta})^2] = \sigma^2$ for all $\vec \theta \in \bbS^{d-1}$.
 \item The expected squared distance to the mean is $\E[\norm{\vec X - \vecb a}^2] = d\sigma^2$.
 \item The moment generating function of $X \sim N_1(0,\sigma)$ is
$\E[e^{\lambda X}] = e^{\lambda^2\sigma^2/2}$, for all $\lambda \in \R$, and that
of $X^2$ is $\E[e^{\lambda X^2}] = 1/\sqrt{1-2\lambda\sigma}$ for $\lambda < 1/(2\sigma)$. \index{moment generating function}
\end{itemize}

We will need the following tail bound for Gaussian random variables. We include a proof for completeness.

\begin{lemma}[Gaussian tail bounds]
\label{lem:gaussian-tails}
For $\vec{X} \in \R^d$ distributed as $N_d(\vec{0},\sigma)$, $t \geq 1$,
\begin{equation}
\label{eq:gauss-full-d}
\Pr[\norm{\vec{X}} \geq t \sigma \sqrt{d}] \leq e^{-(d/2)(t-1)^2} \text{ .}
\end{equation}
For $\vec \theta \in \bS^{d-1}$ and $t \geq 0$,
\begin{equation}
\label{eq:gauss-1d}
\Pr[\abs{\pr{\vec X}{\vec \theta}} \geq t \sigma] \leq 2e^{-t^2/2} \text{ .}
\end{equation}
\end{lemma}
\begin{proof} By homogeneity, we may without loss of generality assume that $\sigma=1$.
\vspace{1em}
\paragraph{ Proof of~\eqref{eq:gauss-full-d}}
\begin{align*}
\Pr[\norm{\vec{X}} \geq \sqrt{d} t]
&= \min_{\lambda \in (0,1/2)} \Pr[e^{\lambda \norm{\vec{X}}^2} \geq e^{\lambda t^2 d}] \\
&\leq \min_{\lambda \in (0,1/2)} \E[e^{\lambda \norm{\vec{X}}^2}] e^{-\lambda t^2 d} \qquad \text{(Markov's inequality)}\\
&= \min_{\lambda \in (0,1/2)} \left(\prod_{i=1}^d\E[e^{\lambda X_i^2}]\right) e^{-\lambda t^2 d} \qquad \text{(Independence of coefficients)}\\
&= \min_{\lambda \in (0,1/2)} \left(\frac{1}{1-2\lambda}\right)^{d/2} e^{-\lambda t^2 d} \\
&\leq e^{-(d/2)(t^2-2\log t-1)}  \qquad \text{(setting $\lambda = \frac{1}{2}(1-1/t^2)$)} \\
&\leq e^{-(d/2)(t-1)^2}  \quad \text{ (since $\log t \leq t-1$ for $t \geq 1$).}
\end{align*}

\paragraph{ Proof of~\eqref{eq:gauss-1d}}
\begin{align*}
\Pr[\abs{\pr{\vec{X}}{\vec \theta}} \geq t] &= 2 \Pr[\pr{\vec{X}}{\vec \theta} \geq t] \\
&\leq 2 \min_{\lambda > 0} \E[e^{\lambda \pr{\vec{X}}{\vec \theta}}] e^{-\lambda t} \\
&= 2 \min_{\lambda > 0} e^{\lambda^2/2 - \lambda t} \leq 2e^{-t^2/2} \text{ , setting $\lambda = t$.}
\end{align*}
\end{proof}

\subsubsection{Laplace distribution}
\label{sub:laplace}
\index{Laplace distribution}
Our shadow bounds will hold for a general class of distributions with
bounds on certain parameters. We illustrate this for the
$d$-dimensional Laplace distribution.

\begin{definition}
 The \emph{Laplace distribution} $L_d(\vecb a,\sigma)$ or \emph{exponential distribution} in $\R^d$ with mean vector $\bar{\vec a}$ has probability
density function
\[
\frac{\sqrt{d}^d}{(d-1)!\sigma^d \vol_{d-1}(\bbS^{d-1})} e^{-\norm{\vec x - \bar{\vec a}}\sqrt{d}/\sigma}.
\]
We abbreviate $L_d(\sigma) = L_d(\vec 0,\sigma)$.
We have normalized the distribution to have expected norm $\sqrt{d}\sigma$.
Additionally, the variance along any direction is $\sigma^2 (1+\frac{1}{d})$.
\end{definition}

The norm of a Laplace distributed random variable follows a \emph{Gamma distribution}.

\begin{definition}
\index{Gamma distribution}
The Gamma distribution $\Gamma(\alpha,\beta), \alpha \in \N, \beta \in \R,$ on the non-negative
real numbers has probability density
$\frac{\beta^\alpha}{(\alpha - 1)!}t^{\alpha-1}e^{-\beta t}$. The moment
generating function of the Gamma distribution is $\E_{X \sim \Gamma(\alpha,\beta)}[e^{\lambda X}] =
(1-\lambda/\beta)^{-\alpha}$ for $\lambda < \beta$.
\end{definition}

One can generate a $d$-dimensional Laplace distribution
$L_d(\sigma)$ as the product of an independent scalar and vector. The
vector $\vec \theta$ is sampled uniformly from the sphere $\bbS^{d-1}$.
The scalar $s \sim \Gamma(d,\sqrt{d}/\sigma)$ is sampled from the Gamma
distribution. The product $s\vec \theta$ has a $L_d(\sigma)$-distribution.

We will need the following tail bound for Laplace distributed random variables.
We include a proof for completeness.

\begin{lemma}[Laplace tail bounds]
\label{lem:laplace-tails}
For $\vec X \in \R^d$, $d \geq 2$, distributed as $(\vec{0},\sigma)$-Laplace and $t \geq 1$,
\begin{equation}
\label{eq:laplace-full-d}
\Pr[\norm{\vec X} \geq t \sigma \sqrt{d}] \leq e^{-d(t-\log t -1)} \text{ .}
\end{equation}
In particular, for $t \geq 2$,
\begin{equation}
\label{eq:laplace-full-d-2}
\Pr[\norm{\vec X} \geq t \sigma \sqrt{d}] \leq e^{-dt/7} \text{ .}
\end{equation}
For $\vec \theta \in \bS^{d-1}$, $t \geq 0$,
\begin{equation}
\label{eq:laplace-1d}
\Pr[\abs{\pr{\vec X}{\vec \theta}} \geq t \sigma] \leq
\begin{cases}
2 e^{-t^2/16} :  0 \leq t \leq 2\sqrt{d} \\
e^{-\sqrt{d}t/7} : t \geq 2\sqrt{d}
\end{cases}\text{ .}
\end{equation}
\end{lemma}
\begin{proof} By homogeneity, we may without loss of generality assume that $\sigma=1$.
\vspace{1em}

\paragraph{Proof of~\eqref{eq:laplace-full-d}}
\begin{align*}
\Pr[\norm{\vec X} \geq \sqrt{d} t] &= \min_{\lambda \in (0,\sqrt{d})} \Pr[e^{\lambda\norm{\vec X}} \geq e^{\lambda\sqrt{d} t }] \\
&\leq \min_{\lambda \in (0,\sqrt{d})} \E[e^{\lambda
\norm{\vec X}}] e^{-\lambda\sqrt{d} t} \qquad \text{(Markov's inequality)} \\
&\leq \min_{\lambda \in (0,\sqrt{d})} (1-\lambda/\sqrt{d})^{-d}
e^{-\lambda \sqrt{d} t} \\
&= e^{-d(t-\log t-1)} \text{ , setting $\lambda = \sqrt{d}(1-1/t)$.}
\end{align*}
For the case $t \geq 2$, the desired inequality follows from the fact that
$t-\log t-1 \geq t/7$ for $t \geq 2$, noting that $(t-\log t -1)/t$ is an
increasing function on $t \geq 1$.

\paragraph{Proof of~\eqref{eq:laplace-1d}} For $t \geq 2\sqrt{d}$, we
directly apply equation~\eqref{eq:laplace-full-d-2}:
\begin{align*}
\Pr[\abs{\pr{\vec X}{\vec \theta}} \geq t \sigma] &\leq \Pr[\norm{\vec X} \geq t \sigma]
\leq e^{-\sqrt{d}t/7} .
\end{align*}
For $t \leq 2\sqrt{d}$, express $\vec X = s \cdot \vec \omega$ for $s \sim
\Gamma(d,\sqrt{d}/\sigma)$, $\vec\omega \in \bbS^{d-1}$ uniformly sampled.
\begin{align*}
 \Pr[\abs{\pr{s \vec \omega}{\vec \theta}} \geq t \sigma] &\leq
\Pr[\abs{\pr{\vec \omega}{\vec \theta}} \geq t/(2\sqrt{d})] + \Pr[\abs{s} \geq 2\sqrt{d}\sigma] \\
 &\leq \Pr[\abs{\pr{\vec \omega}{\vec \theta}} \geq t/(2\sqrt{d})] + e^{-d/4}.
\end{align*}

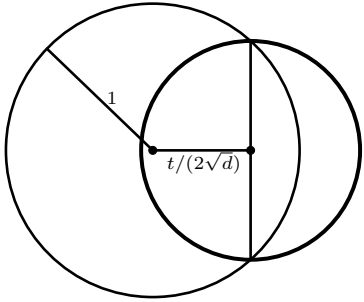
\begin{figure}[ht]
\centering
\definecolor{uuuuuu}{rgb}{0.26666666666666666,0.26666666666666666,0.26666666666666666}
\definecolor{ududff}{rgb}{0.5333333333333333,0.5333333333333333,0.5333333333333333}
\begin{tikzpicture}[line cap=round,line join=round,>=triangle 45,x=2.0cm,y=2.0cm]
\clip(-1.1349613321827652,-1.122903714025955) rectangle (1.5236261177954022,1.0915989725928248);
\draw [line width=1.pt] (0.,0.)-- (0.6504276358303109,0.);
\draw [line width=1.pt] (0.,0.) circle (0.9760091934107995);
\draw [line width=1.pt] (0.,0.)-- (-0.7036899537401479,0.6763241786507327);
\draw [line width=1.5pt] (0.6504276358303109,0.) circle (0.7276935042795091);
\draw [line width=1.pt] (0.6504276358303109,0.7276935042795091)-- (0.6504276358303109,-0.727693504279509);
\begin{scriptsize}
\draw [fill=black] (0.,0.) circle (1.5pt);
\draw [fill=black] (0.6504276358303109,0.) circle (1.5pt);
\draw[color=black] (0.3394000821703834,-0.08229745651812818) node {$t/(2\sqrt{d})$};
\draw[color=black] (-0.2704763261765013,0.35) node {$1$};
\end{scriptsize}
\end{tikzpicture}
\caption{The small sphere has at least as much surface area as combined surface area of the enclosed sphere cap and the opposite cap together
by the monotonicity of surface area (Lemma \ref{lem:monotonicity}).}
\label{fig:spherecap}
\end{figure}
For the first term we follow~\cite[Lemma 2.2]{jour/flge/Ball97}, where the second
line is illustrated in Figure~\ref{fig:spherecap}:
\begin{align*}
\Pr[\abs{\pr{\vec \omega}{\vec \theta}} \geq t/(2\sqrt{d})]
&= \frac{\vol_{d-1}(\set{\vec \omega \in \bbS^{d-1} : \abs{\pr{\vec \omega}{\vec
\theta}} \geq t/(2\sqrt{d})})}{\vol_{d-1}(\bbS^{d-1})} \\
&\leq \frac{\vol_{d-1}(\sqrt{1-\frac{t^2}{4d}}\bbS^{d-1})}{\vol_{d-1}(\bbS^{d-1})} \\
&= (1-\frac{t^2}{4d})^{(d-1)/2} \\
&\leq e^{-t^2(d-1)/(8d)} \leq e^{-t^2/16}.
\end{align*}
The desired conclusion follows since $e^{-t^2/16}+e^{-d/4} \leq 2e^{-t^2/16}$ for
$0 \leq t \leq 2\sqrt{d}$.
\end{proof}

\subsection{Change of variables}
\label{sec:blaschke}

In section~\ref{sec:shadow-bounds} we make use of a change of variables that
was analyzed by Blaschke~\cite{jour/bmsrs/Blaschke35}, and is standard in the study of convex hulls.

\index{change of variables}
Recall that a change of variables affects a probability distribution.
Let the vector $\vec y \in \R^d$ be a random variable with density $\mu$. If $\vec y =
\phi(\vec x)$ and $\phi$ is invertible, then the induced density on $\vec x$ is
\[
\mu(\phi(\vec x)) \abs*{\det \left(\frac{\partial \phi(\vec x)}{\partial \vec x}\right) },
\]
where $\abs*{\det \left(\frac{\partial \phi(\vec x)}{\partial \vec x}\right)  }$
is the Jacobian of $\phi$. We describe a particular change of variables which
has often been used for studying convex hulls, and, in particular, by
Borgwardt~\cite{Borgwardt87} and Spielman and Teng \cite{jour/jacm/ST04} for deriving shadow bounds.

For affinely independent vectors $\vec a_1,\dots,\vec a_d \in \R^d$ we have the coordinate transformation
\[
(\vec a_1,\dots,\vec a_d) \mapsto (\vec \theta, t, \vec b_1,\dots,\vec b_d),
\]
where $\vec \theta \in \bbS^{d-1}$ and $t \geq 0$ satisfy $\inner{\vec \theta}{\vec a_i} = t$ for every $i \in \set{1,\dots,d}$ and the vectors $\vec b_1,\dots,\vec b_d \in \R^{d-1}$ parametrize the positions of $\vec a_1,\dots,\vec a_d$ within the hyperplane $\set{\vec x \in \R^d \mid \inner{\vec \theta}{\vec x} = t}$. We coordinatize the hyperplanes as follows:

Fix a reference vector $\vec v \in \bbS^{d-1}$, and pick an isometric embedding $h : \R^{d-1} \to \vec v^\perp$. For
any unit vector $\vec \theta \in \bbS^{d-1}$, define the map $R'_{\vec \theta} :
\R^d \to \R^d$ as the unique map that rotates $\vec v$ to $\vec \theta$ along ${\rm
span}(\vec v, \vec \theta)$ and fixes
the orthogonal subspace $\linsp(\vec v,\vec \theta)^\perp$.
We define $R_{\vec \theta} = R'_{\vec \theta} \circ h$. The change of variables from $\vec \theta \in \bbS^{d-1},t > 0, \vec b_1,\dots,\vec b_d \in \R^{d-1}$ to  $\vec a_1,\dots,\vec a_d$ takes the form
\[
(\vec a_1,\dots,\vec a_d) = (R_{\vec \theta} \vec b_1 + t\vec \theta, \dots, R_{\vec \theta} \vec b_d + t\vec \theta).
\]
The change of variables as specified above is not uniquely defined when $\vec a_1,\dots,\vec a_d$ are affinely dependent, when $t = 0$ or when $\vec \theta = - \vec v$.\index{$R_{\vec\theta}$}

\begin{theorem} \label{prelim:blaschke} \index{Blaschke's theorem} \index{change of variables}
 Let $\vec \theta \in \bbS^{d-1}$ be a unit vector, $t > 0$ and
$\vec b_1,\dots,\vec b_d \in \R^{d-1}$. Consider the map \[(\vec \theta, t, \vec b_1,\dots,\vec b_d) \mapsto (\vec a_1,\dots,\vec a_d) = (R_{\vec \theta} \vec b_1 + t\vec \theta, \dots, R_{\vec \theta} \vec b_d + t\vec \theta).\]
 The Jacobian of this map equals
 \[\abs*{\det \left(\frac{\partial \phi(\vec x)}{\partial \vec x}\right)  } = (d-1)!\vol_{d-1}(\conv(\vec b_1,\dots,\vec b_d)).\]
\end{theorem}

\subsection{Shadow vertex algorithm}
We briefly introduce the shadow vertex algorithm. For proofs of the statements below,
see \cite{thesis/huiberts}. An alternative exposition about the shadow vertex algorithm
can be found in \cite{Borgwardt87}.

Let $P = \set{\vec x \in \R^d : \vec A\vec x \leq \vec b}$ be a polyhedron,\index{$P$}
and let $\vec a_1,\dots,\vec a_n \in \R^d$ correspond to the rows of $\vec A$. We
call a set $B \subseteq [n]$ a basis of $\vec A \vec x \leq \vec b$ if $\vec A_B$
is invertible. This implies that $\abs{B} = d$.
We say $B$ is a feasible basis if $\vec x_B = \vec A_B^{-1}
\vec b_B$ satisfies $\vec A \vec x_B \leq \vec b$.
The point $\vec x_B$ is a vertex of $P$.
We say a feasible basis $B$ is optimal for an objective
$\vec c \in \R^d$ if $\vec c^{\T} \vec A_B^{-1} \geq \vec 0$, which happens if
and only if $\max_{\vec x \in P} \inner{\vec c}{\vec x} = \inner{\vec c}{\vec
x_B}$.

\begin{algorithm}
\caption{Shadow vertex algorithm for non-degenerate polyhedron and shadow.}
\label{alg:shadow-simplex}
\begin{algorithmic}
\REQUIRE $P = \set{ \vec{x} \in \R^d : \vec A\vec{x} \leq \vec{b} }$,
    $\vec{c}, \vec{d} \in \R^d$,
    feasible basis $B \subseteq [n]$ optimal for $\vec{d}$.
\ENSURE Return optimal basis $B \subseteq [n]$ for $\vec{c}$ or \emph{unbounded}.
\STATE  $\lambda_0 \gets 0$.
\STATE  $i \gets 0$.
  \LOOP
  \STATE   $i \gets i + 1$.
  \STATE   $\lambda_i :=$ maximum $\lambda \leq 1$ such that $\vec c_\lambda^\T
\vec A_B^{-1} \geq \vec 0$. 
    \IF{$\lambda_{i} = 1$}
      \RETURN B.
    \ENDIF
  \STATE   $k := k \in B$ such that $(\vec c_{\lambda_{i}}^{\mathsf{T}} \mathbf A_B^{-1})_k = 0$.
  \STATE   $\vec x_B := \vec A_B^{-1}\vec b_B$.
  \STATE   $s_i :=$ supremum $s > 0$ such that $\vec A(\vec x_B -  s \vec
A_B^{-1} \vec e_k) \leq \vec b$. 
    \IF{\emph{$s_i = \infty$}}
      \RETURN \emph{unbounded}.
    \ENDIF
   \STATE  $j := j \in [n] - B$ such that $\inner{\vec a_j}{(\vec x_B -  s_i \vec A_B^{-1} \vec e_k)} = b_j$.
   \STATE  $B \gets B \cup \set{j}\setminus\set{k}$.
  \ENDLOOP
\end{algorithmic}
\end{algorithm}

The shadow vertex algorithm is a pivot rule for the simplex method. Given a
feasible basis $B \subseteq [n]$, an objective $\vec d \in \R^d$ for which $B$
is optimal, and an objective function $\vec c \in \R^d$ to optimize, where $\vec
c$ and $\vec d$ are linearly independent, the shadow
vertex algorithm (Algorithm \ref{alg:shadow-simplex}) specifies which pivot
steps to take to reach an optimal basis for $\vec c$.  We note that there are
many possible choices for starting objective $\vec d$.

We parametrize $\vec c_\lambda := (1-\lambda) \vec d + \lambda \vec c$ and start at
$\lambda = 0$. The shadow vertex rule increases $\lambda$ until there are $j \neq
k \in [n]$ such that a new feasible basis $B \cup \set{j} \setminus \set{k}$ is optimal for $\vec c_\lambda$,
and repeat with increased $\lambda$ and new basis $B$ until $\lambda = 1$.

The index $k \in B$ is such that the coordinate
for $k$ in $\vec c_\lambda^{\T} \vec A_B^{-1}$ first lowers to $0$, and $j \not \in B$ is
such that $B \cup \set{j} \setminus \set{k}$ is a feasible basis: we follow the edge
$\vec A_B^{-1} \vec b_B - \vec A_B^{-1} \vec e_k \R_+$ until we hit the first constraint
$\vec a_j^\T \vec x \leq b_j$, and then replace $k$ by $j$ to get the new basis
$B \cup \set{j} \setminus \set{k}$.

Changing the current basis from $B$ to $B \cup \set{j} \setminus \set{k}$ is called a pivot step. As soon as $\lambda = 1$ we have $\vec c_\lambda = \vec c$, at which moment the current basis is optimal for our objective $\vec c$. If at some point no choice of $j$ exists, then an unbounded ray has been found. \\

\begin{figure}[t]
 \centering
 \ifarxiv
  \includegraphics[width=0.45\textwidth]{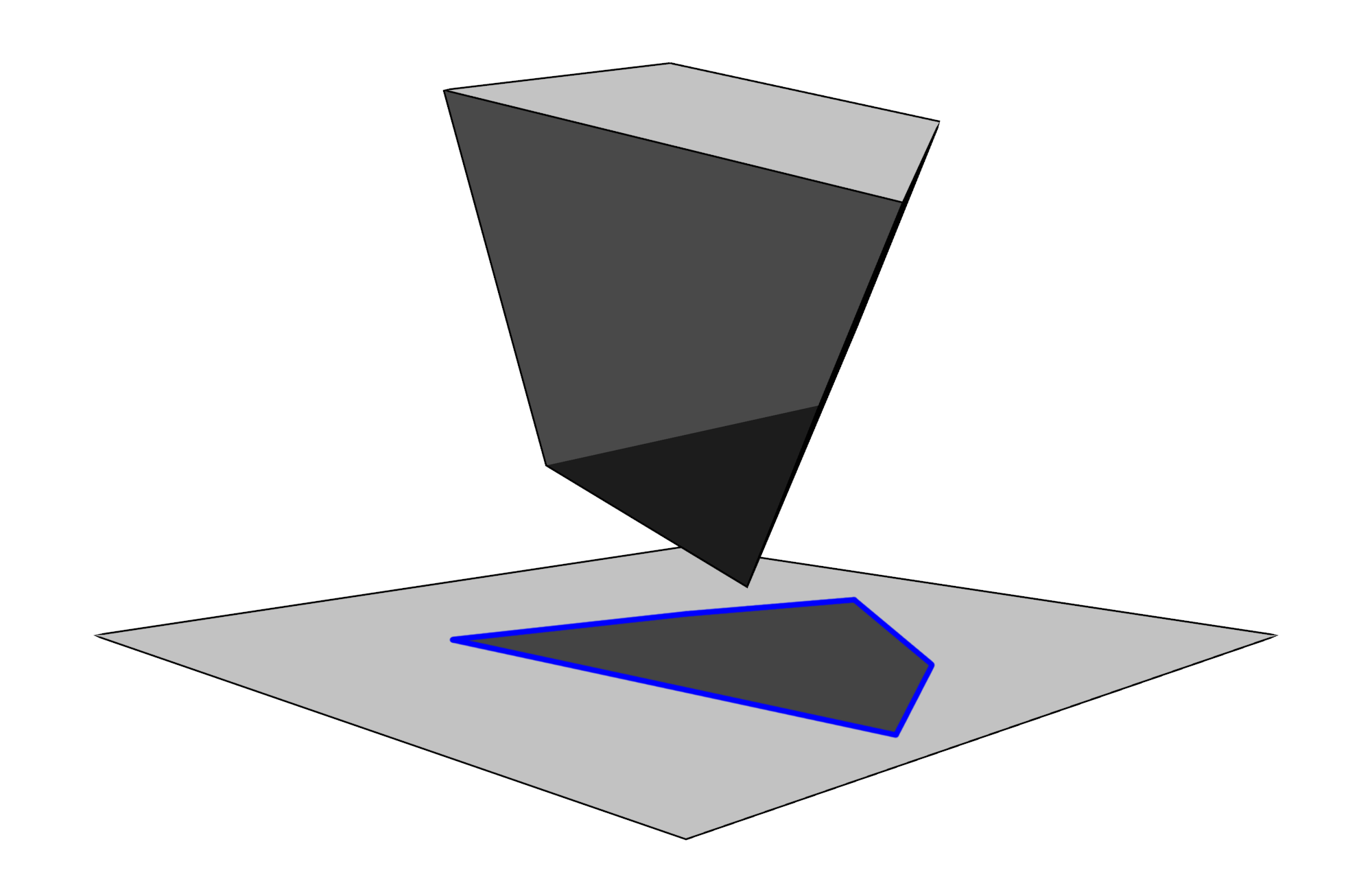}
  \includegraphics[width=0.45\textwidth]{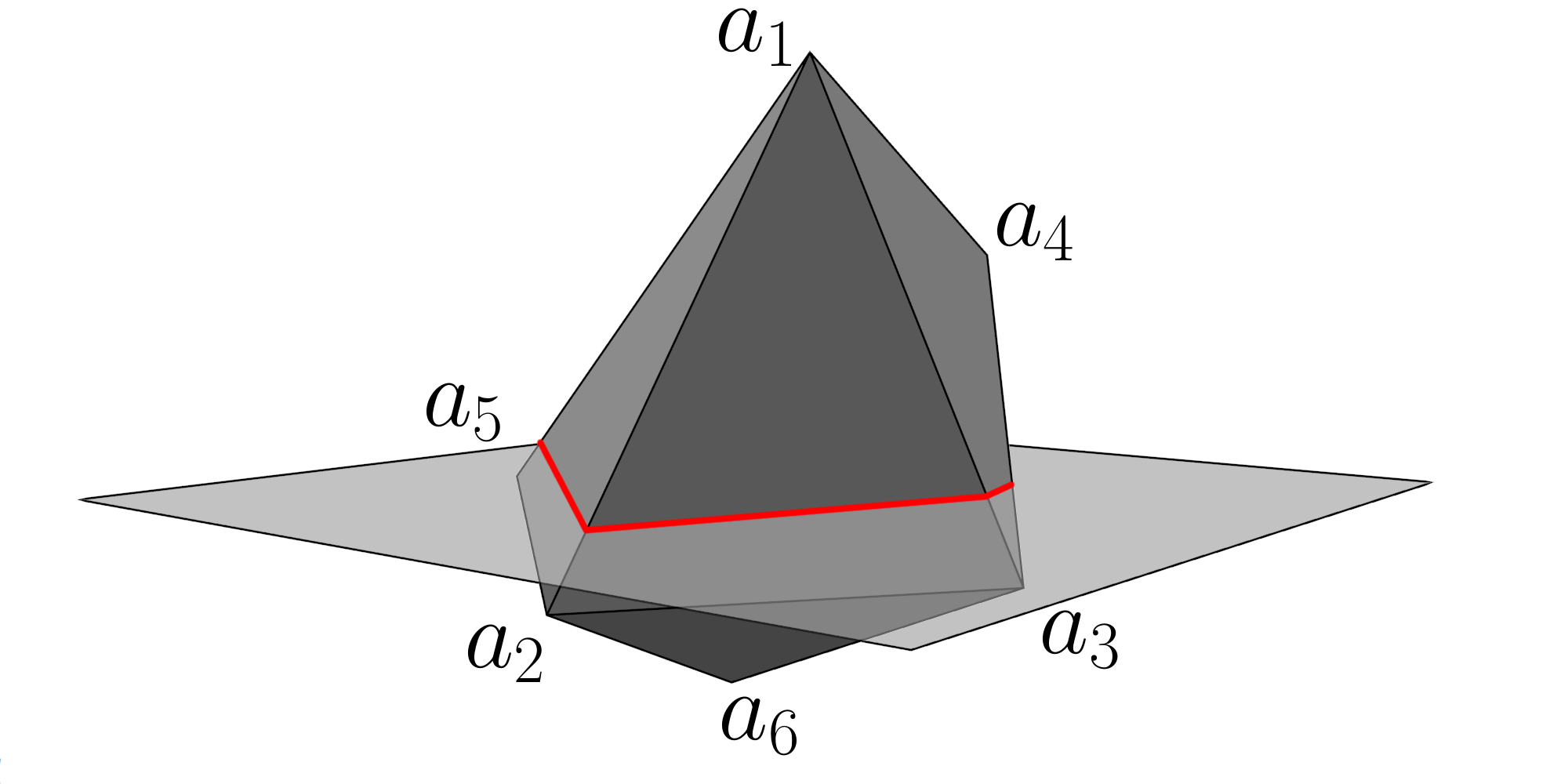}
 \else
  \includegraphics[width=0.45\textwidth]{\main/images/shadow-blue}
  \includegraphics[width=0.45\textwidth]{\main/images/intersect-revised}
 \fi
 \caption{On the left, a polytope and its shadow. On the right, the corresponding polar polytope intersected with the plane. There are as many edges marked blue as there are edges marked red. \label{fig:shadow}}
\end{figure}

\begin{definition}
We say that the system $\vec A \vec x \leq \vec b$ is \emph{non-degenerate}
if $n \geq d$, any $B \in \binom{[n]}{d}$ is a basis, and every vertex of the corresponding
polyhedron $P$ is tight at exactly $d$ linearly independent inequalities.
When the description $\vec A \vec x \leq \vec b$ is clear, we say that $P$ is
non-degenerate to mean that its describing system is.
\end{definition}

\begin{definition}
We say that the shadow of $P$ on a two-dimensional linear subspace $W$ is
\emph{non-degenerate} if $\dim(\pi_W(P))=2$ and for every face $F$ of $P$ such that $\pi_W(F)$ is
a face of $\pi_W(P)$ and $\dim(\pi_W(F)) \leq 1$, we have that $\dim(\pi_W(F)) = \dim(F)$.
\end{definition} \index{degenerate!shadow}\index{non-degenerate shadow}

If both the polyhedron and the shadow are non-degenerate, each pivot step can
be performed in $O(nd)$ time (see the pseudo-code for
Algorithm~\ref{alg:shadow-simplex}). Under the distribution models we examine,
degeneracy occurs with probability $0$.

The shadow vertex rule is called as such because the visited vertices are in
correspondence with vertices on the relative boundary of the orthogonal projection
$\pi_W(P)$ of $P$ onto $W = \linsp(\vec d, \vec c)$, where we denote $\pi_W(P)$
as the shadow of $P$ on $W$. See Figure~\ref{fig:shadow}.
We call the total number of vertices of the projection the
\emph{shadow size}, and it is the key geometric estimate
in our analysis of the simplex method.

\begin{lemma} \label{lem:shadow-vertices-are-vertices-of-shadow}
For a polyhedron $P = \set{\vec x \in \R^d: \vec A \vec x \leq \vec b}$
having a non-degenerate shadow on $W$, the vertices of $P$ optimizing
objectives in $W \setminus \set{\vec 0}$ are in one-to-one correspondence with the
vertices of $\pi_W(P)$ under the map $\pi_W$.
\end{lemma}

We will consider non-degenerate polyhedra of the form $\set{\vec x \in \R^d : \mathbf
A\vec x \leq \vec 1}$,
in which case $\vec 0$ is always contained in the polyhedron. The problem thus
has a known feasible solution. It is instructive to look at the geometry of shadow
paths on such polyhedra from a \emph{polar perspective}. For any non-degenerate polyhedron $P = \set{\vec
x \in \R^d : \mathbf A\vec x \leq \vec 1}$, we look at the polar polytope, defined as the
convex hull $Q := \conv(\vec a_1,\ldots,\vec a_n)$ of the
constraint vectors. For any index-set $I \subseteq [n], \abs{I}=d$, if the (unique) solution
$\vec x_I$ to the equations \[\inner{\vec a_i}{\vec x} = 1 \qquad \forall i \in
I \] is a vertex of the original polyhedron $P$, then the set
$\conv(\vec a_i: i \in I)$ forms a facet of the polytope $\conv(\vec
a_1,\ldots,\vec a_n)$. Conversely, if $\conv(\vec a_i: i\in I)$ induces a facet of
$Q\ := \conv(\vec 0, \vec a_1,\ldots,\vec a_n)$ (note the inclusion of $\vec 0$), then
$\vec x_I$ is a vertex of $P$. The addition of $\vec 0$ to the polar of $P$
allows us to detect unboundedness. Precisely, the facets of the extended polar
$\conv(\vec 0, \vec a_1,\dots, \vec a_n)$ containing $\vec 0$ are in one
to one correspondence with unbounded edges of $P$. $P$ is bounded, i.e.~a
polytope, if and only if $\vec 0$ is in the interior of $Q$. In this case
$Q=Q'$, and hence every facet of $Q$ is associated to a vertex of $P$.

In the polar perspective, a pivot step moves from one facet of $Q'$ to a
neighboring facet. The shadow vertex algorithm moves the objective $\vec
c_\lambda$ along the line segment $[\vec d,\vec c]$ and keeps track of which
facet of $Q'$ is intersected by the ray $\vec c_\lambda \R^+$. If we move to a
facet of $Q'$ containing $\vec 0$, we may conclude that the LP with objective
$\vec c$ is in fact unbounded. Since we can only visit such facets at the end of
a shadow path, we will be able to control the length of shadow paths using only
the geometry of $Q$, which will help simplify our analyses. The main bound on
the size of the shadow we will use is given in the following lemma.

\begin{lemma} \label{lem:polar}
 Let $P = \set{\vec x \in \R^d: \vec A \vec x \leq \vec 1}$ be a non-degenerate
 polyhedron with a non-degenerate shadow on $W$. Then
 \[\abs{\vertices(\pi_W(P))} \leq \abs{{\rm edges}(\conv(\vec a_1,\dots, \vec a_n) \cap W)}.\]
\end{lemma}

The number of pivot steps taken in a shadow path is bounded from above
by the number of edges in the intersection $\conv(\vec a_1,\dots,\vec
a_n) \cap \linsp(\vec d, \vec c)$. Hence it suffices that we prove an upper bound
on this geometric quantity. The following theorem summarizes the properties we will
use of the shadow vertex algorithm.

\begin{theorem} \label{thm:geometric-characterization-shadow-path-size}
Let $P = \set{\vec x \in \R^d: \vec A \vec x \leq \vec b}$ denote a
non-degenerate polyhedron, and let $\vec a_1,\dots,\vec a_n \in \R^d$ be the rows
of $\vec A$. Let $\vec c, \vec d \in \R^d$ denote two
objectives inducing a non-degenerate shadow on $P$, and let $W =
\linsp(\vec d,\vec c)$. Given a feasible basis $I \in \binom{[n]}{d}$ for
$\vec A \vec x \leq \vec b$ which is optimal
for $\vec d$, Algorithm~\ref{alg:shadow-simplex} (shadow vertex) finds a feasible basis $J \in \binom{[n]}{d}$
optimal for $\vec c$ or declares unboundedness in a number of pivot steps bounded by
$\abs{\vertices(\pi_W(P))}$, where $\pi_W$ is the orthogonal projection
onto $W$. In particular, when $\vec b = \vec 1$, the number of pivots is at most
\[
\abs{\edges(\conv(\vec a_1,\dots,\vec a_n) \cap W)}~.
\]
\end{theorem}

%% file: shadowbounds.tex
\section{Shadow Bounds}
\label{sec:shadow-bounds}

In this section, we derive our new and improved shadow bounds for Laplace and Gaussian
distributed perturbations. We achieve these results by first proving a shadow
bound for parametrized distributions as described in the next section, and then
specializing to the case of Laplace and Gaussian perturbations. The bounds we obtain are
described below.

\begin{theorem} \label{thm:gaussian} \index{Gaussian shadow bound}
Let $W \subset \R^d$ be a fixed two-dimensional subspace, $ n \geq d \geq 3$ and let $\vec
a_1,\dots,\vec a_n \in \R^d$, be independent Gaussian random vectors with
variance $\sigma^2$ and centers of norm at most $1$. Then the expected number of
edges is bounded by \[\E[\abs{\edges(\conv(\vec a_1,\dots, \vec a_n) \cap W)}]
\leq \mathcal{D}_g(n,d,\sigma),\]
where the function $\calD_g(d,n,\sigma)$ is defined as
\[\calD_g(d,n,\sigma) := O(d^2\sqrt{\log n}~\sigma^{-2} +
d^{2.5}\log n~\sigma^{-1} + d^{2.5}\log^{1.5} n).\]
\end{theorem}

Our bound applies more generally for distributions satisfying certain parameters.
We illustrate this with a shadow bound for perturbations distributed according to
the Laplace distribution. This will serve as a good warm-up exercise for the
slightly more involved analysis of the Gaussian distribution.

\begin{theorem} \label{thm:laplace} \index{Laplace shadow bound}
Let $W \subset \R^d$ be a fixed two-dimensional subspace, $ n \geq d \geq 3$ and let $\vec
a_1,\dots,\vec a_n \in \R^d$, be independent Laplace distributed random vectors with
parameter $\sigma$ and centers of norm at most $1$. Then the expected number of
edges is bounded by \[\E[\abs{\edges(\conv(\vec a_1,\dots,\vec a_n) \cap W)}]
= O(d^{2.5}\sigma^{-2} + d^{3}\log n~\sigma^{-1} + d^{3}\log^2 n).\]
\end{theorem}

The proofs of Theorems~\ref{thm:gaussian} and~\ref{thm:laplace} are given in
respectively subsections~\ref{sec:shadow-gaussian} and~\ref{sec:shadow-laplace}.

\subsection{Parametrized Shadow Bound}
In this section, we prove a shadow bound theorem for any noise distribution
that has non-trivial bounds on certain parameters. The parameters we will use are defined below.

\subsubsection{Distribution parameters}
\label{sub:dist-params}
\index{distribution parameters}
\begin{definition}\label{def:log-Lipschitz}
A distribution with density $\mu$ on $\R^d$ is \emph{$L$-log-Lipschitz} if for all $\vec x, \vec y \in \R^d$ we have $\abs{\log(\mu(\vec x)) - \log(\mu(\vec y))} \leq L\norm{\vec x-\vec y}$. Equivalently, $\mu$ is $L$-log-Lipschitz if $\mu(\vec x)/\mu(\vec y) \leq \exp(L\norm{\vec x-\vec y})$ for all $\vec x, \vec y \in \R^d$. \index{log-Lipschitz}\index{$L$|see {log-Lipschitz}}
\end{definition}

\begin{definition} \label{def:linestddev}
 Given a probability distribution with density $\mu$ on $\R^d$, we define the
line variance $\tau^2$ as the infimum of the variances when restricted to any fixed line $l \subset \R^d$:
 \[\tau^2 = \inf_{\mbox{\scriptsize \emph{line} $l \subset \R^d$}} \Var(\vec X \sim \mu \mid \vec X \in l).\]
\end{definition}

Both the log-Lipschitz constant and the minimal line variance relate to how
``spread out'' the probability mass is. The log-Lipschitzness of a random
variable gives a lower bound on the line variance, which we prove in
Lemma~\ref{lem:linevariance-vs-lipschitz}.

\begin{definition} \label{def:deviation}
 Given a distribution with probability density $\mu$ on $\R^d$ with expectation $\E_{\vec X \sim \mu}[\vec X] = \vec y$ we define the $n$-th deviation $\deviation$ to be the smallest number such that for any unit vector $\vec \theta \in \R^d$,
 \[\int_{\deviation}^\infty \Pr_{\vec X \sim \mu}[\abs{\inner{(\vec X - \vec y)}{\vec \theta}} \geq t] dt \leq \deviation/n.\]
 Note that as $\deviation$ increases to $\infty$, the left-hand side goes to $0$ and the right-hand side goes to $\infty$. We see that there must exist a number satisfying this inequality, so $\deviation$ is well-defined.
\end{definition}

The $n$-th deviation will allow us to give bounds on the expected maximum size $\E[\max_{i \leq n} \abs{\inner{\vec x_i}{\vec\theta}}]$ of $n$ separate perturbations in a given direction $\vec\theta$. We formalize this in Lemma~\ref{lem:deviation-vs-expectednorm}.

\begin{definition} \label{def:cutoffnorm}
 Given a distribution with probability density $\mu$ on $\R^d$ with expectation $\E_{\vec x \sim \mu}[\vec x] = \vec y$, we define, for all $1 > p > 0$, the cutoff radius $R(p)$ as the smallest number satisfying
 \[\Pr_{\vec x \sim \mu}[\norm{\vec x - \vec y} \geq R(p)] \leq p.\]
\end{definition}
The cutoff radius of interest is $\cutoffnorm := R(\frac{1}{d\binom{n}{d}})$. The cutoff radius tells us how concentrated the probability mass of the random variable is, while the log-Lipschitzness tells us how spread out the probability mass is. These quantities cannot both be arbitrarily good (small) at the same time. We formalize this notion in Lemma~\ref{lem:lipschitz-vs-cutoff}.

\begin{lemma}\label{lem:deviation-vs-expectednorm}
If $\vec x_1,\dots,\vec x_n$ are each distributed with mean $\vec 0$ and $n$-th deviation at most $r_n$, then for any $\vec \theta \in \bbS^{d-1}$,
\[
\E[\max_{i\in [n]} \abs{\inner{\vec \theta}{\vec x_i}}] \leq 2 r_n .
\]
\end{lemma}
\begin{proof}
We rewrite the expectation as
\[
\E[\max_{i \in [n]} \abs{\inner{\vec \theta}{\vec x_i}}] = \int_0^\infty \Pr[\max_{i \in [n]} \abs{\inner{\vec \theta}{\vec x_i}} \geq t] \d t .
\]
We separately bound the integral up to $r_n$ and from $r_n$ to $\infty$. Since a probability is at most $1$ we have
\[
\int_0^\deviation \Pr[\max_{i \in [n]} \abs{\inner{\vec \theta}{\vec x_i}} \geq t] \d t \leq \deviation,
\]
and by definition of the $n$-th deviation and the union bound:
\begin{align*}
\int_\deviation^\infty \Pr[ \max_{i \in [n]} \abs{\inner{\vec \theta}{\vec x_i}} \geq t] \d t &\leq \sum_{i\in [n]} \int_\deviation^\infty \Pr[\abs{\inner{\vec \theta}{\vec x_i}} \geq t] \\
&\leq \deviation.
\end{align*}
Together these estimates yield the desired inequality,
\[
\E[\max_{i \leq n} \abs{\inner{\vec \theta}{\vec x_i}}] \leq 2\deviation.
\]
\end{proof}

\begin{lemma} \label{lem:linevariance-vs-lipschitz}
 If a distribution with probability density $\mu$ is $L$-log-Lipschitz, then its line variance satisfies
$\tau \geq 1/(\sqrt{e}L)$.
\end{lemma}
\begin{proof}
 Let $\vec v + \vec w\R$ be a line and assume that $\E[\vec x\mid \vec x \in \vec v + \vec w\R] = \vec v$ and $\norm{\vec w} = 1$. We show that with probability at least $1/e$, $\vec x$ has distance at least $1/L$ from $\vec v$.
 Conditioning on $\vec x \in \vec v + \vec w\R$, the induced probability mass is proportional to $\mu(\vec x)$.
 We can bound the fraction of the induced probability mass that is far away from the expectation by the
following calculation:
 \begin{align*}
  \int_{-\infty}^\infty \mu(\vec v + \gamma \vec w) \d \gamma &= \int_{-\infty}^0 \mu(\vec v + \gamma \vec w) \d \gamma + \int_0^\infty \mu(\vec v + \gamma \vec w) \d \gamma \\
  &= \int_{-\infty}^{-1/L} \mu(\vec v + (\gamma + 1/L) \vec w) \d \gamma + \int_{1/L}^\infty \mu(\vec v + (\gamma - 1/L) \vec w) \d \gamma \\
  &\leq e \int_{-\infty}^{-1/L} \mu(\vec v + \gamma \vec w) \d \gamma + e \int_{1/L}^\infty \mu(\vec v + \gamma \vec w) \d \gamma.
 \end{align*}
The integral on the first line exists because it is the integral of a
continuous non-negative function, and,
if the integral were infinite, then the integral along every parallel line would be infinite by log-Lipschitzness,
contradicting the fact that $\mu$ has integral $1$ over $\R^d$.
 
 Hence, $\Pr[\norm{\vec x - \vec v} \geq 1/L \mid \vec x \in \vec v + \vec w\R] = \frac{\int_{-\infty}^{-1/L} \mu(\vec v + \gamma \vec w) \d \gamma + \int_{1/L}^\infty \mu(\vec v + \gamma \vec w) \d \gamma}{\int_{-\infty}^\infty \mu(\vec v + \gamma \vec w) \d \gamma} \geq 1/e,$
 and we can lower bound the variance
 \[\Var(\vec x \mid \vec x \in \vec v + \vec w\R) \geq \frac{1}{e} (1/L)^2.\]
 Since the line $\vec v + \vec w\R$ was arbitrary, it follows that $\tau \geq 1/(\sqrt{e} L)$.
\end{proof}

\begin{lemma} \label{lem:lipschitz-vs-cutoff}
For a $d$-dimensional distribution with probability density $\mu$, where $d \geq 3$, with parameters $L,
\cutoffnormnosub$ as described above, we have the inequality $L\cutoffnormnosub(1/2) \geq d/3$.
\end{lemma}
\begin{proof}
 Let $\bar{R} := \cutoffnormnosub(1/2)$. If $L\bar{R} \geq d$, we are already
done, so we may assume that $L\bar{R} < d$. Also, without loss of generality, we
may assume that $\mu$ has mean $\vec 0$. For $\alpha > 1$ to be chosen later we know
 \begin{align*}
  1 &\geq \int_{\alpha \bar R \calB_2^d} \mu(\vec x) \d \vec x \\
  &= \alpha^d \int_{\bar R \calB_2^d} \mu(\alpha \vec x)\d \vec x \\
  &\geq \alpha^d e^{-(\alpha-1) L \bar R} \int_{\bar R \calB_2^d}\mu(\vec x)\d \vec x \\
  &= \frac{\alpha^d}{2}e^{-(\alpha-1) L \bar R}.
\end{align*}
 Taking logarithms, we find
 \[0 \geq d \log(\alpha) - (\alpha-1) L \bar R - \log(2).\]
 We choose $\alpha = \frac{d}{L \bar R} > 1$ and look at the resulting inequality:
 \[0 \geq d \log(\frac{d}{L\bar R}) - d + L\bar{R} - \log(2).\]
 For $d \geq 3$, this inequality can only hold if $L\bar R \geq d/3$, as needed.
\end{proof}

\subsubsection{Proving a shadow bound for parametrized distributions}
The main result of this subsection is the following parametrized shadow bound.

\begin{theorem}[Parametrized Shadow Bound] \label{thm:abstractedbound} \index{parametrized shadow bound}
Let $\vec a_1,\ldots,\vec a_n \in \R^d$, where $n \geq d \geq 3$, be independently distributed according
to $L$-log-Lipschitz distributions with centers of norm at
most $1$, line variances at least $\tau^2$, cutoff radii at most $\cutoffnorm$
and $n$-th deviations at most $\deviation$. For any fixed two-dimensional linear
subspace $W \subset \R^d$, the expected number of edges satisfies
 \[\E[\abs{\edges(\conv(\vec a_1,\dots,\vec a_n) \cap W)}] \leq O(\frac{d^{1.5}L}{\tau}(1+\cutoffnorm)(1+\deviation)).\]
\end{theorem}

The proof is given at the end of the subsection. It  will be derived from the
sequence of lemmas given below. We refer the reader to
subsection~\ref{sec:proof-sketch} for a high-level overview
of the proof.

In the rest of the subsection, $\vec a_1,\dots,\vec a_n \in \R^d$, where $n \geq d \geq 3$, will
be as in Theorem~\ref{thm:abstractedbound}. We use $Q := \conv(\vec a_1,
\dots, \vec a_n)$ to denote the convex hull of the constraint vectors and $W$ to
denote the two-dimensional shadow plane. \index{$Q$}

The following non-degeneracy conditions on $\vec a_1,\dots,\vec a_n$ will hold with probability $1$, because $\vec a_1,\dots,\vec a_n$ are independently distributed with continuous distributions.
\begin{enumerate}
 \item Every $d+1$ vectors from $\vec a_1,\dots,\vec a_n$ are affinely independent. Thus, every facet of $Q$ is the convex hull of exactly $d$ vectors from $\vec a_1,\dots,\vec a_n$.
 \item Any $d$ distinct vectors $\vec a_{i_1},\dots,\vec a_{i_d}$, $i_1,\dots,i_d \in
[n]$, have a unique hyperplane through them. This hyperplane intersects $W$ in
a one-dimensional line, does not contain the origin $\vec 0$, and its unit
normal vector pointing away from the origin is not $- \vec e_1$. Note that the
last two conditions imply that Blaschke's transformation on $\vec a_{i_1},\dots,\vec a_{i_d}$ is uniquely defined with $\vec e_1$ as reference vector.
 \item For every edge $e \subset Q \cap W$ there is a unique facet $F$ of $Q$ such that $e = F \cap W$.
\end{enumerate}
In what follows we will always assume the above conditions hold.

For our first lemma, in which we bound the number of edges in terms of two
different expected lengths, we make a distinction between possible edges with
high probability of appearing versus edges with low probability of appearing.
The sets with probability at most $2\binom{n}{d}^{-1}$ to form an edge, together
contribute at most $2$ to the expected number of edges, as there are only
$\binom{n}{d}$ bases.

For a basis with probability at least $2\binom{n}{d}^{-1}$ of forming an edge,
we can safely condition on it forming an edge without forcing very unlikely
events to happen. Because of this, we will later be able to condition on
the vertices not being too far apart.

\begin{definition} \label{def:edgeevent}
 For $I \in \binom{[n]}{d}$, let $E_I$ denote the event that $\conv(\vec a_i : i \in I) \cap W$ forms an edge of $Q \cap W$.\index{$E$}
\end{definition}

\begin{definition} \label{def:badsets}
 We define the set $B \subseteq \binom{[n]}{d}$ to be the set of those $I
\subseteq [n]$ satisfying $\abs{I} = d$ and $\Pr[E_I] \geq 2 \binom{n}{d}^{-1}$.
\end{definition}

The next lemma is inspired by Theorem 3.2 of \cite{conf/stoc/KS06}.

\begin{lemma} \label{lem:edgesasfraction}
The expected number of edges in $Q \cap W$ satisfies \begin{displaymath}
\E[\abs{\edges(Q \cap W)}] \leq 2 + \frac{\E[\perimeter(Q \cap W)]}{\min_{I \in
B} \E[\length(\conv(\vec a_i : i \in I)\cap W) \mid E_I]}. \end{displaymath}
\end{lemma}

\begin{proof}
 We give a lower bound on the perimeter of the intersection $Q \cap W$ in terms of the number of edges. By our non-degeneracy assumption, every edge can be uniquely represented as $\conv(\vec a_i : i \in I) \cap W$, for $I \in \binom{[n]}{d}$. From this we derive the first equality, and we continue from that:
 \begin{align*}
  \E[\perimeter(Q \cap W)] &= \sum_{I \in \binom{[n]}{d}} \E[\length(\conv(\vec a_i : i \in I)\cap W) \mid E_I] \Pr[E_I] \\
  &\geq \sum_{I \in B} \E[\length(\conv(\vec a_i : i \in I)\cap W) \mid E_I] \Pr[E_I] \\
  &\geq \min_{I \in B} \E[\length(\conv(\vec a_i : i \in I)\cap W) \mid E_I]
\sum_{J \in B} \Pr[E_J].
 \end{align*}
 The first line holds because whenever $E_I$ holds, $\conv(\vec a_i : i \in I)\cap W$ is an edge of $Q \cap W$, and every edge of $Q \cap W$ is formed by exactly one face $F_J$, by the non-degeneracy conditions we have assumed. By construction of $B$ and linearity of expectation, $\sum_{J \in B} \Pr[E_J] \geq \sum_{J \in \binom{[n]}{d}} \Pr[E_J] - 2 = \E[\abs{\edges(Q \cap W)}] - 2$.
 By dividing on both sides of the inequality, we can now conclude
\begin{displaymath}\E[\abs{\edges(Q \cap W)}] \leq 2 + \frac{\E[\perimeter(Q
\cap W)]}{\min_{I \in B} \E[\length(\conv(\vec a_i : i \in I)\cap W) \mid
E_I]}.\end{displaymath}
\end{proof}

Given the above, we may now restrict our task to proving an upper bound on the
expected perimeter and a lower bound on the minimum expected edge length, which
will be the focus on the remainder of the subsection.

The perimeter is bounded using a standard convexity argument. A convex shape has perimeter no more than that of any circle containing it. We exploit the fact that all centers have norm at most $1$ and the expected perturbation sizes are not too big along any fixed axis.

\begin{lemma} \label{lem:perimeterbound} \index{perimeter bound}
The expected perimeter of $Q \cap W$ is bounded by
 \[\E[\perimeter(Q \cap W)] \leq 2\pi(1+4\deviation),\]
\end{lemma}\index{n-th deviation@$n$-th deviation}
where $r_n$ is the $n$-deviation bound for $\vec a_1,\dots,\vec a_n$.
\begin{proof}
By convexity, the perimeter is bounded from above by $2 \pi$ times the norm of
the maximum norm point. Let $\hat{\vec a}_i := \vec a_i-\E[\vec a_i]$ denote the
perturbation of $\vec a_i$ from the center of its distribution, recalling that
$\norm{\E[\vec a_i]} \leq 1$ by assumption. We can now
derive the bound
\begin{align*}
  \E[\perimeter(Q \cap W)] &\leq 2 \pi \E[\max_{\vec x \in Q \cap W} \norm{\vec x}]\\
  &= 2 \pi \E[\max_{\vec x \in Q \cap W} \norm{\pi_W(\vec x)}] \\
  &\leq 2 \pi \E[\max_{\vec x \in Q} \norm{\pi_W(\vec x)}] \\
  &= 2 \pi \E[\max_{i \in [n]} \norm{\pi_W(\vec a_i)}] \\
  &\leq 2\pi\left(1 + \E[\max_{i \leq n} \norm{\pi_W(\hat{\vec
a}_i)}]\right)~,
 \end{align*}
where the last inequality follows since $\vec a_1,\dots,\vec a_n$ have centers
of norm at most $1$. Pick an orthogonal basis $\vec v_1,\vec v_2$ of $W$. By the
triangle inequality the expected perturbation size satisfies
 \begin{align*}
  \E[\max_{i \leq n} \norm{\pi_W(\hat{\vec a}_i)}] &\leq \sum_{j \in
\set{1,2}} \E[\max_{i \leq n} \abs{\inner{\vec v_j}{\hat{\vec a}_i}}].
  \end{align*}
  Each of the two expectations satisfies, by Lemma~\ref{lem:deviation-vs-expectednorm},
  $\E[\max_{i \leq n} \abs{\inner{\vec v_j}{\hat{\vec a}_i}}] \leq 2\deviation$, thereby
  concluding the proof.
\end{proof}

The rest of this subsection will be devoted to finding a suitable lower bound on
the denominator $\E[\length(\conv(\vec a_i : i \in I)\cap W) \mid E_I]$
uniformly over all choices of $I \in B$. Without loss of generality we assume
that $I = [d]$ and write $E := E_{[d]}$.

\begin{definition}[Containing hyperplane]
Define $H = \aff(\vec a_1,\dots,\vec a_d) = t\vec \theta + \vec \theta^\perp$, where
$\vec \theta \in \bbS^{d-1}$, $t > 0$ to be the hyperplane containing $\vec
a_1,\dots,\vec a_d$. Define $l = H \cap W$. From our non-degeneracy conditions
we know that $l$ is a line. Express $l = \vec p + \vec \omega
\cdot \R$, where $\vec \omega \in \bbS^{d-1}$ and $\vec p \in \vec \omega^\perp$.
\end{definition}

To lower bound the length $\E[\length(\conv(\vec a_1,\dots,\vec a_d)\cap W) \mid
E]$ we will need the pairwise distances between the different $\vec a_i$'s for
$i \in \set{1,\dots,d}$ to be small along ${\vec \omega}^\perp$.  This will
allow us to get ``wiggle room'' around each vertex of $\conv(\vec a_1,\dots,\vec
a_d)$ that
is proportional to the size of the facet.

\begin{definition}[Bounded diameter event]
 We define the event $D$ to hold exactly when $\norm{\pi_{\vec \omega^\perp}(\vec a_i) -
\pi_{\vec \omega^\perp}(\vec a_j)} \leq 2 + 2\cutoffnorm$ for all $i,j \in [d]$. \index{$D$}
\end{definition}

We will condition on the event $D$. This will not change the expected length by
much, because the probability that $D$ does not occur is small compared to
the probability of $E$ by our assumption that $\Pr[E] \geq \frac{2}{d\binom{n}{d}}$.

\begin{lemma} \label{lem:boundedprojectedsimplex} The expected edge length satisfies
\[
\E[\length(\conv(\vec a_1,\dots,\vec a_d)\cap W) \mid E] \geq
\E[\length(\conv(\vec a_1,\dots,\vec a_d)\cap W) \mid D,E]/2.
\]
\end{lemma}
\begin{proof}
 Let the vector $\hat{\vec a}_i$ denote the perturbation $\vec a_i - \E[\vec
a_i]$. Since distances can only decrease when projecting, the complementary
event $D^c$ satisfies
 \begin{align*}
  \Pr[D^c] &= \Pr[\max_{i,j \leq d} \norm{\pi_{\vec \omega^\perp}(\vec a_i - \vec a_j)} \geq 2 + 2\cutoffnorm] \\
  &\leq \Pr[\max_{i,j \leq d} \norm{\vec a_i - \vec a_j} \geq 2 + 2\cutoffnorm] , \\
  \intertext{ by the triangle inequality and the bound of $1$ on the norms of
the centers, the line above is at most   }
  &\leq \Pr[\max_{i \leq d} \norm{\vec a_i} \geq 1 + \cutoffnorm] \\
  &\leq \Pr[\max_{i \leq d} \norm{\hat{\vec a}_i} \geq \cutoffnorm] \\
  &\leq \binom{n}{d}^{-1}.
 \end{align*}
By our assumption that $[d] \in B$, we know that $\Pr[E] \geq 2\binom{n}{d}^{-1}$. In
particular, it follows that $\Pr[E \cap D] \geq \Pr[E] - \Pr[D^c] \geq
\Pr[E]/2$. Thus, we may conclude that
\[
\E[\length(\conv(\vec a_1,\dots,\vec a_d)\cap W) \mid E] \geq
\E[\length(\conv(\vec a_1,\dots,\vec a_d)\cap W) \mid D,E]/2.
\]

\end{proof}

For the rest of this section, we use a change of variables on $\vec a_1,\dots,\vec a_d$.
The non-degeneracy conditions we have assumed at the start of this section make the
change of variables well-defined.

\begin{definition}[Change of variables] \index{change of variables}
\label{def:change-of-variables}
Recall the change of variables mapping $(\vec a_1,\dots,\vec a_d) \mapsto (\vec \theta,
t, \vec b_1,\dots,\vec b_d)$ for $\vec \theta \in \bbS^{d-1}, t > 0, \vec
b_1,\dots,\vec b_d \in \R^{d-1}$ from Theorem~\ref{prelim:blaschke}. We abbreviate
$\bar \mu_i(\vec \theta, t, \vec b_i) = \mu_i(R_\theta(\vec b_i) + t\vec \theta
)$ and we write $\bar \mu_i(\vec b_i)$ when the values of $\vec \theta, t$ are
clear. \index{$\bar \mu$}
\end{definition}
By Theorem~\ref{prelim:blaschke} of Blaschke \cite{jour/bmsrs/Blaschke35}
we know that for any fixed values of $\vec \theta, t$ the vectors $\vec
b_1,\dots,\vec b_d$ have joint probability density proportional to
\begin{equation}
\label{eq:change-var-disft}
\vol_{d-1}(\conv(\vec b_1,\dots,\vec b_d)) \prod_{i=1}^d \bar \mu_i(\vec b_i)
~.
\end{equation}

We assumed that $\vec a_1,\dots,\vec a_d$ are affinely independent,
so $\vec b_1,\dots,\vec b_d$ are affinely independent as well.

%

In the next lemma, we condition on the hyperplane $H = t\vec \theta + \vec
\theta^\perp$ and from then on we restrict our attention to what happens inside
$H$. Conditioned on $\vec a_1,\dots,\vec a_d$ lying in $H$, the set
$\conv(\vec a_1,\dots,\vec a_d)$ is a facet of $Q$ if and only if
all of $\vec a_{d+1},\dots,\vec a_n$ lie on one side of $H$. This means that
the shape of $\conv(\vec a_1,\dots,\vec a_d)$ in $H$ does not influence the
event that it forms a facet, so in studying this convex hull we can then
ignore $\vec a_{d+1},\dots,\vec a_n$.

We identify the hyperplane $H$ with $\R^{d-1}$ and define $\bar l
= \vecb p + \bar{\vec \omega}\cdot \R \subset \R^{d-1}$ corresponding to $l =
\vec p + \vec \omega \cdot \R$ by $\vecb p = R_{\vec \theta}^{-1}(\vec p - t\vec
\theta)$, $\bar{\vec \omega} = R_{\vec \theta}^{-1}(\vec \omega)$. We define
$\bar E$ as the event that $\conv(\vec b_1,\dots,\vec b_d) \cap \bar l \neq
\emptyset$. Notice that $E$ holds if and only if $\bar{E}$ and $\conv(\vec
a_1,\dots, \vec a_d)$ is a facet of $Q$. See Figure~\ref{fig:whatiswhat}.

\begin{figure}
\centering
\definecolor{ududff}{rgb}{0,0,0}
\definecolor{ffqqqq}{rgb}{1.,0.,0.}
\definecolor{ududff}{rgb}{0,0,0}
\definecolor{ffqqqq}{rgb}{1.,0.,0.}
\definecolor{ttffqq}{rgb}{0.2,1.,0.}
\begin{tikzpicture}[line cap=round,line join=round,>=triangle 45,x=2cm,y=2cm]
\clip(7.598900553265315,-11.69523380416249) rectangle (10.670516876517425,-9.374760686455328);
\draw [line width=1.pt] (8.588596881659882,-10.161240367755855)-- (7.902200875568879,-10.436853181736783);
\draw [line width=1.pt] (7.902200875568879,-10.436853181736783)-- (9.0453978875529,-11.116868685704095);
\draw [line width=1.pt] (9.0453978875529,-11.116868685704095)-- (10.375537831265433,-10.469485157715456);
\draw [line width=1.pt] (10.375537831265433,-10.469485157715456)-- (9.94374057767173,-10.266900092558686);
\draw [line width=1.pt] (10.078871049685484,-9.907470302926917)-- (9.431482863676065,-11.629440159465107);
\draw [line width=1.pt] (9.431482863676065,-11.629440159465107)-- (8.310571567942524,-10.931957472490662);
\draw [line width=1.pt] (8.310571567942524,-10.931957472490662)-- (8.385473358107015,-10.72432137603898);
\draw [line width=1.pt] (8.588596881659882,-10.161240367755855)-- (8.868401138736406,-9.385591834091308);
\draw [line width=1.pt] (8.868401138736406,-9.385591834091308)-- (10.078871049685484,-9.907470302926917);
\draw [line width=1.pt] (8.588596881659882,-10.161240367755855)-- (9.75381414296485,-10.77208014155058);
\draw [line width=1.pt] (9.048772993781144,-10.402477673079575)-- (8.852364176718286,-10.489066506408344);
\draw [line width=1.pt,dashed] (9.418557708224563,-10.239454304346433)-- (9.048772993781144,-10.402477673079575);
\draw [line width=1.pt,dashed] (9.197448530179818,-9.916764417056248)-- (9.94374057767173,-10.266900092558686);
\draw [line width=1.pt,dashed] (9.197448530179818,-9.916764417056248)-- (8.588596881659882,-10.161240367755855);
\draw [line width=1.pt,dashed] (8.588596881659882,-10.161240367755855)-- (8.385473358107015,-10.72432137603898);
\draw [line width=1.pt,color=ffqqqq] (9.04877299378121,-10.402477673079611)-- (9.244111991223017,-10.504879888505949);
\draw [line width=1.pt,color=blue] (8.852364176718286,-10.489066506408344)-- (9.151456249707275,-10.645859012214327);

\draw [line width=1.pt] (8.852364176718286,-10.489066506408344)-- (9.074411396767482,-10.76308559392418);
\draw [line width=1.pt] (9.074411396767482,-10.76308559392418)-- (9.244111991223013,-10.504879888505947);
\draw [line width=1.pt,dashed] (9.244111991223013,-10.504879888505947)-- (9.418557708224563,-10.239454304346433);

\draw[color=black] (8.15,-10.215089379210024) node {$H$};
\draw[color=black] (8.9,-9.702218790082534) node {$W$};
\draw[color=black] (9.55,-10.575786607569185) node {$l$};
\draw [fill=ududff] (9.074411396767482,-10.76308559392418) circle (1.5pt);
\draw[color=ududff] (9.106333506362482,-10.85781130927557) node {$a_1$};
\draw [fill=ududff] (8.852364176718286,-10.489066506408344) circle (1.5pt);
\draw[color=ududff] (8.75,-10.424356542110575) node {$a_2$};
\draw [fill=ududff] (9.418557708224563,-10.239454304346433) circle (1.5pt);
\draw[color=ududff] (9.443289107643025,-10.17252656641669) node {$a_3$};
\end{tikzpicture}
\caption{$\vec a_1,\dots,\vec a_d$ are conditioned for $\conv(\vec
a_1,\dots,\vec a_d)$ to intersect $W$ and lie in $H$. The red line corresponds
to induced edge. The blue line represents the longest chord parallel to $\ell$.
 \label{fig:whatiswhat}}
\end{figure}
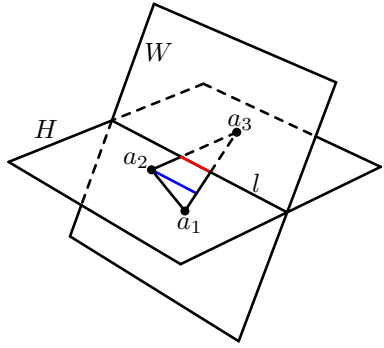

We will condition on the shape of the projected simplex.

\begin{definition}[Projected shape]
\label{def:projected-shape}
We define the projected shift variable by $\vec x := {\vec x}_{\vec \omega}(\vec
b_1) = \pi_{{\bar {\vec \omega}}^\perp}(\vec b_1)$ and shape variable $S:=S_{\omega}(\vec b_1,\dots,\vec b_d)$ by
\[
S_{\vec \omega}(\vec b_1,\dots,\vec b_d) =
(\vec 0, \pi_{{\bar {\vec \omega}}^\perp}(\vec b_2) - \vec x,\dots,
   \pi_{{\bar {\vec \omega}}^\perp}(\vec b_d) - \vec x) \text{ .}
\]
We index $S = (\vec s_1,\dots,\vec s_d)$, so $\vec s_i \in {\bar {\vec
\omega}}^\perp$ is the $i$-th vector in $S$, and furthermore define the diameter
function $\diam(S) = \max_{i,j\in [d]}\norm{\vec s_i - \vec s_j}$. We will
condition on the shape being in the set of allowed shapes
\[
\calS :=\! \set{\!(\vec s_1,\dots\!,\vec s_d) \!\in\! ({\bar {\vec \omega}}^\perp\!)^d: \vec s_1 = 0, \diam(S)
\!\leq\! 2 + 2\cutoffnorm, \mathrm{rank}(\vec s_2,\dots\!,\vec s_d) = d-2}\!.
\]
\end{definition}
Observe that $S \in \calS$ if and only if the event $D$ holds. To justify the
rank condition on $\vec s_2,\dots,\vec s_d$, note that by our non-degeneracy
conditions, we have that $\vec b_1,\dots,\vec b_d$ are affinely independent. In
particular, they do not all lie in a $d-2$-dimensional affine subspace. This
means that $\vec s_1,\dots,\vec s_d$ do not all lie in a $d-3$-dimensional
affine subspace, from which it follows that $\mathrm{rank}(\vec s_2,\dots\!,\vec
s_d) = d-2$ (recalling that $\vec s_1 = \vec 0$).
 
\begin{lemma} \label{lem:conditiononhyperplane}
Let $\vec \theta \in \bS^{d-1}, t > 0, \vec b_1,\dots,\vec b_d \in \R^{d-1}$
denote the change of variables of $\vec a_1,\dots,\vec a_n \in \R^d$ as
in Definition~\ref{def:change-of-variables}. Then, the expected length satisfies
\begin{displaymath} \E[\length(\conv(\vec a_1,\dots\!,\vec a_d)\cap W) \!\mid\! D,\!E] \!\!\!~~ \geq\! \inf_{\vec
\theta, t, S \in \calS}\E[\length(\conv(\vec b_1,\dots\!,\vec b_d)\cap \bar l) \!\mid\! \vec
\theta, t, S,\! \bar E]. \end{displaymath}
\end{lemma}
\begin{proof}
To derive the desired inequality, we first understand the effect of conditioning
on $E$. Let $E_0$ denote the event that $F := \conv(\vec a_1,\dots,\vec a_d)$
induces a facet of $Q$. Note that $E$ is equivalent to $E_0 \cap \bar E$, where
$\bar E$ is as above. We now perform the change of variables from $\vec a_1,\ldots,\vec a_d \in \R^d$
to $\vec \theta \in \bbS^{d-1}, t \in \R_+$, $\vec b_1,\dots,\vec b_d \in
\R^{d-1}$ as in Definition~\ref{def:change-of-variables}. The set $F$ is a
facet of $Q$ if and only if $\inner{\vec \theta}{\vec a_{d+i}} \leq t$ for all $i \in [n-d]$
or $\inner{\vec \theta}{\vec a_{d+i}} \geq t$ for all $i \in [n-d]$.
Given this, we see that
\begin{equation}
\label{eq:cond-hyp-1}
\begin{split}
&\E[\length(\conv(\vec a_1,\dots,\vec a_d)\cap W) \mid D,E] \\
&=
\E[\length(\conv(\vec b_1,\dots,\vec b_d)\cap \bar l) \mid D,E_0,\bar{E}] \\
&= \frac{\E[~\mathbbm{1}[E_0] \cdot
\length(\conv(\vec b_1,\dots,\vec b_d)\cap \bar l) \mid
D,\bar{E}]}{\Pr[E_0 \mid D,\bar{E}]} \\
&= \frac{\E_{\vec \theta,t}[~\E[\mathbbm{1}[E_0]
\cdot
\length(\conv(\vec b_1,\dots,\vec b_d)\cap \bar l) \mid
\vec \theta,t,D,\bar{E}]~]}{\E_{\vec \theta,t}[~\Pr[E_0 \mid \vec \theta,t,D,\bar{E}]~]}
\end{split}
\end{equation}
Since $\vec a_1,\dots,\vec a_n$ are independent, conditioned on $\vec \theta,
t$, the random vectors $\vec b_1,\dots,\vec b_d$ are independent of $\inner{\vec \theta}{\vec
a_{d+1}},\dots,\inner{\vec \theta}{\vec a_{n}}$. Since the events $D$ and
$\bar{E}$ only depend on $\vec b_1,\dots,\vec b_d$, continuing
from~\eqref{eq:cond-hyp-1}, we get
that
\begin{align*}
&\frac{\E_{\vec \theta,t}[~\E[\mathbbm{1}[E_0]
\cdot
\length(\conv(\vec b_1,\dots,\vec b_d)\cap \bar l) \mid
\vec \theta,t,D,\bar{E}]~]}{\E_{\vec \theta,t}[~\Pr[E_0 \mid \vec \theta,t,D,\bar{E}]~]} \\
& = \frac{\E_{\vec \theta,t}[\Pr[E_0 \mid \vec \theta,t]
\cdot
\E[\length(\conv(\vec b_1,\dots,\vec b_d)\cap \bar l) \mid
\vec \theta,t,D,\bar{E}]~]}{\E_{\vec \theta,t}[~\Pr[E_0 \mid \vec \theta,t]]} \\
& \geq \inf_{\vec \theta \in \bS^{d-1},t > 0} \E[\length(\conv(\vec b_1,\dots,\vec b_d)\cap \bar l) \mid
\vec \theta,t,D,\bar{E}] .
\end{align*}
The last inequality uses that $\frac{\int f(x) g(x) \d x}{\int f(x) \d x} \geq \inf g(x)$ if $f$ is non-negative and has finite integral.

Lastly, since the event $D$ is equivalent to $S := S_{\omega}(\vec
b_1,\dots,\vec b_d) \in \calS$ as in Definition~\ref{def:projected-shape},
we have that 
 \begin{equation*} \E[\length(\conv(\vec a_1,\dots\!,\vec a_d)\cap W) \!\mid\!
D,\!E] \!\!\!~~ \geq \inf_{\vec
\theta, t, S \in \cal S} \E[\length(\conv(\vec b_1,\dots\!,\vec b_d)\cap \bar l) \!\mid\! \vec
\theta, t, S,\! \bar E]. \end{equation*}
\end{proof}

\begin{definition}[Kernel combination] \index{kernel combination}
\label{def:kernel-combination}
For $S \in \calS$, define the combination $\vec z :=
\vec z(S)$ to be the unique (up to sign) $\vec z = (z_1,\dots,z_d) \in \R^d$
satisfying \[\sum_{i=1}^d z_i \vec s_i = \vec 0,\, \sum_{i=1}^d z_i = 0,\,\norm{\vec
z}_1 = 1.\]
\end{definition}

To justify the above definition, it suffices to show that the system of equations
(i) $\sum_{i=1}^d z_i
\vec s_i = \vec 0, \sum_{i=1}^d z_i = 0$ has a one-dimensional solution space.
Since $\vec s_1,\dots,\vec s_d$ live in a $d-2$ dimensional space, the solution
space has dimension at least $1$ by dimension counting. Next, note that $\vec z$
is a solution to (i) iff $z_1 = - \sum_{i=2}^d z_i$ and (ii) $\sum_{i=2}^d z_i
\vec s_i = \vec 0$ (since $\vec s_1 = \vec 0$). Thus, the solution space of (i)
and (ii) have the same dimension. Given our assumption that $\mathrm{rank}(\vec
s_2,\dots,\vec s_d) = d-2$, it follows that (ii) is one-dimensional, as needed.

Observe that for $S := S_{\vec \omega}(\vec
b_1,\dots,\vec b_d)$, $\vec z$ satisfies $\pi_{{\bar {\vec \omega}}^\perp}(\sum_{i=1}^d z_i
\vec b_i) = \vec 0$.

The vector $\vec z$ provides us with a unit to measure lengths in ``convex
combination space''. We make this formal with the next definition:

\begin{definition}[Chord combinations] \label{def:convcombspace} \index{chord combinations} \index{$\norm{C_S(\vec q)}_1$}
 We define the set of convex combinations of the shape $S = (\vec s_1,\dots,\vec s_d) \in
\calS$
that equal $\vec q \in \bar{\vec \omega}^\perp$ by \[C_S(\vec q) := \set{(\lambda_1,\dots,\lambda_d) \geq \vec 0
~:~ \sum_{i=1}^d \lambda_i = 1, ~\sum_{i=1}^d \lambda_i \vec s_i = \vec q}
\subset \R^d.\] When $S$ is clear we drop the subscript.
\end{definition}

Observe that $C(\vec q)$ is a line segment of the form $C(\vec q) =
\vec \lambda_{\vec q} +  \vec z \cdot [0,d_{\vec q}]$. We write $\norm{C(\vec q)}_1$
for the $\ell_1$-diameter of $C(\vec q)$. Since $C(\vec q)$ is a line segment,
$\norm{C(\vec q)}_1 = d_{\vec q}$. We prove two basic properties of $\norm{C(\vec q)}_1$
as a function of $\vec q$.

\begin{lemma}[Properties of chord combinations] \label{lem:chord-combinations}
Let $\vec y := \vec y(S) = \sum_{i = 1}^d \abs{z_i}\vec s_i$, with $\vec z := \vec z(S)$ as in Definition~\ref{def:convcombspace}. Then the following holds:
\begin{itemize} \index{$\vec y(S)$}
 \item $\norm{C(\vec q)}_1$ is a concave function for $\vec q \in \conv(S)$.
 \item $\max_{\vec q \in \conv(S)} \norm{C(\vec q)}_1 = \norm{C(\vec y)}_1 = 2$.
\end{itemize}
\end{lemma}
\begin{proof}
For the first claim, take $\vec x,\vec y\in \conv(S)$. Let $\vec \alpha \in C(\vec x)$ and $\vec \beta \in C(\vec y)$.
Then we see that, for all $\gamma \in [0,1]$,
\[\gamma\vec \alpha + (1-\gamma)\vec \beta \geq \vec 0,\quad \sum_{i=1}^d \gamma\alpha_i + (1-\gamma)\beta_i = 1,\quad \sum_{i=1}^d (\gamma\alpha_i + (1-\gamma)\beta_i)\vec s_i = \gamma \vec x + (1-\gamma)\vec y,\]
from which we derive that \[\gamma C(\vec x) + (1-\gamma)C(\vec y) \subseteq C(\gamma\vec x + (1-\gamma)\vec y),\]
and hence $\norm{C(\gamma\vec x + (1-\gamma)\vec y)}_1 \geq \norm{\gamma C(\vec x) + (1-\gamma)C(\vec y)}_1 = \gamma \norm{C(\vec x)}_1 + (1-\gamma)\norm{C(\vec y)}_1$.

For the second claim, we look at the
combination $\vec y := \sum_{i=1}^n  \abs{z_i}\vec s_i \in \conv(S)$. For all $\gamma
\in [-1,1]$,  we have $\sum_{i=1}^d (|z_i|+\gamma z_i)\vec s_i = \vec y$,
$\sum_{i=1}^d |z_i|+\gamma z_i = \norm{\vec z}_1 = 1$ and $|z_i|+\gamma z_i \geq 0$,
$\forall i \in [d]$. Hence, $\norm{C(\vec y)}_1 \geq 2$. Now suppose there is some $\vec y'$ with $\norm{C(\vec y')}_1 > 2$. That means there is some convex combination $\vec\lambda = (\lambda_1,\dots,\lambda_d) \geq \vec 0$, $\norm{\vec\lambda}_1 = 1$, with $\sum_{i=1}^d \lambda_i \vec s_i = \vec y'$ such that $\vec \lambda + \vec z > \vec 0$ and $\vec \lambda - \vec z > \vec 0$. Let $I \cup J$ be a partition of $[d]$ such that $z_i \geq 0$ for $i \in I$ and $z_j < 0$ for $j \in J$.
We know that $\sum_{i=1}^d z_i = 0$, so $\sum_{i \in I} z_i = -\sum_{j \in J} z_j$. This makes $1 = \norm{\vec z}_1 = \sum_{i\in I} z_i + \sum_{j \in J} - z_i = 2\sum_{i\in I} z_i$, so $\sum_{i \in I}z_i = 1/2$.
The combination $\vec \lambda$ satisfies \[\sum_{i\in I} \lambda_i > \sum_{i\in I} z_i = 1/2, \quad \sum_{j\in J}\lambda_j > \sum_{j\in J}-z_j = 1/2,\]
so $\norm{\vec \lambda}_1 > 1$. By contradiction we conclude that $\max_{\vec q \in \conv(S)} \norm{C(\vec q)}_1 = 2$.
\end{proof}

The $\ell_1$-diameter $\norm{C(\vec q)}_1$ specified by $\vec q
\in \conv(S(\vec b_1,\dots,\vec b_d))$ directly relates to the length of the
chord $(\vec q + \vec x + {\bar {\vec \omega}} \cdot \R) \cap \conv(\vec b_1,\dots,\vec b_d)$,
which projects to $\vec q + \vec x$ under $\pi_{{\bar {\vec \omega}}^\perp}$. Specifically, $\norm{C(\vec q)}_1$
measures how long the chord is compared to the longest chord through the simplex. The exact relation is
given below.

\begin{lemma} \label{lem:length-as-product}
 Let $(h_1,\dots,h_d) = (\inner{{\bar {\vec \omega}}}{\vec b_1},
\dots,\inner{{\bar {\vec \omega}}}{\vec b_d})$, $(\vec s_1,\dots,\vec s_d) =
S(\vec b_1,\dots,\vec b_d)$,  $\vec x = \pi_{{\bar{\vec \omega}}^\perp}(\vec
b_1)$. For any $\vec q \in \conv(S)$ the following equality holds:
 \[\length((\vec x + \vec q + {\bar {\vec \omega}}\cdot \R)\cap \conv(\vec b_1,\dots,\vec b_d)) =
\norm{C(\vec q)}_1 \cdot \abs{\sum_{i=1}^d z_i h_i}.\]
\end{lemma}
\begin{proof}
By construction there is a convex combination $\lambda_1,\dots,\lambda_d\geq 0$,
$\sum_{i=1}^d\lambda_i = 1$ satisfying $\sum_{i=1}^d \lambda_i \vec s_i = \vec
q$ such that $C(\vec q) = [\vec\lambda, \vec\lambda + \norm{C(\vec q)}_1\vec z]$ and hence \[
(\vec x + \vec q + {\bar {\vec \omega}} \cdot \R)\cap  \conv(\vec b_1,\dots,\vec b_d) =
[\sum_{i=1}^d \lambda_i \vec b_i, \sum_{i=1}^d (\lambda_i + \norm{C(\vec q)}_1 z_i)\vec b_i].
\]
From this we deduce
\begin{align*}
\!  \length((\vec x + \vec q + {\bar {\vec \omega}} \cdot \R) \cap \conv(\vec b_1,\dots,\vec b_d)) &=
\norm*{\sum_{i=1}^d (\lambda_i + \norm{C(\vec q)}_1 z_i)\vec b_i - \sum_{i=1}^d \lambda_i \vec b_i}\\
  &= \norm*{\sum_{i=1}^d \norm{C(\vec q)}_1 z_i\vec b_i} \\
  &= \norm{C(\vec q)}_1 \cdot \abs{\sum_{i=1}^d z_i h_i}.
 \end{align*}
 The third equality follows from the definition of $z_1,\dots,z_d$: as
$\pi_{{\bar {\vec \omega}}^\perp}(\sum_{i=1}^d z_i \vec b_i) = \vec 0$,
we must have
$\norm{\sum_{i=1}^d z_i \vec b_i} = \norm{\sum_{i=1}^d z_i h_i \bar{\vec
\omega}} = \abs{\sum_{i=1}^d z_i h_i}$.
\end{proof}

We can view the terms in the above product as follows:
the length of the longest chord of $\conv(\vec b_1,\dots,\vec b_d)$ parallel to $\bar l$ is $2\abs{\sum_{i=1}^d z_i h_i}$,
and the ratio of the length of the chord $\conv(\vec b_1,\dots,\vec b_d)\cap \bar l$ to the length of the longest chord parallel to $\bar l$ equals $\norm{C(\vec q)}_1/2$. This follows from Lemma~\ref{lem:chord-combinations} since $\norm{C(\vec q)}_1$ achieves a maximum value of $2$ at $\vec q = \vec y$. As discussed in the high-level description, we will bound the expected values of these two quantities separately.

The term $\abs{\sum_{i=1}^d z_i h_i}$ can also be used to simplify the volume term in the probability density of $\vec b_1,\dots,\vec b_d$ after we condition on the shape $S$. We prove this in the next lemma.

\begin{lemma} \label{lem:factorizevolume}
For fixed $\vec \theta \in \bbS^{d-1}, t > 0, S \in \calS$, define $\vec x \in {\bar {\vec
\omega}}^\perp,h_1,\dots,h_d \in \R$ conditioned on $\vec \theta,t,S$ to have joint probability density function proportional to
\begin{displaymath} \abs{\sum_{i=1}^d z_i h_i} \cdot \prod_{i=1}^d \bar
\mu_i(\vec x + \vec s_i + h_i{\bar {\vec \omega}}), \end{displaymath}
where $\vec z:=\vec z(S)$ is as in Definition~\ref{def:kernel-combination}. Then for $\vec
b_1,\dots,\vec b_d \in \R^{d-1}$ distributed as in Lemma~\ref{lem:conditiononhyperplane},
conditioned on $\vec \theta, t$ and the shape $S = (\vec s_1,\dots, \vec s_d)$,
where $\vec s_1 = \vec 0$, we have equivalence of the distributions
\begin{displaymath} (\vec b_1,\dots,\vec b_d)\mid \vec\theta,t,S \equiv (\vec x + \vec s_1 +
h_1{\bar {\vec \omega}},\dots,\vec x + \vec s_d + h_d{\bar {\vec \omega}}) \mid
\vec \theta,t,S. \end{displaymath}
\end{lemma}
\begin{proof}
 By Definition~\ref{def:change-of-variables}, the variables $\vec b_1,\dots,\vec b_d$ conditioned on $\vec
\theta, t$ have density proportional to \[\vol_{d-1}(\conv(\vec b_1,\dots,\vec
b_d))\prod_{i=1}^d\bar \mu_i(\vec b_i).\] We make a change of variables from
$\vec b_1,\dots,\vec b_d$ to $\vec x, \vec s_2,\dots,\vec s_d \in {\bar {\vec
\omega}}^\perp, h_1,\dots,h_d \in \R$, defined by
\begin{displaymath} (\vec b_1,\dots,\vec b_d) =
(\vec x + h_1{\bar {\vec \omega}}, \vec x + \vec s_2 + h_d{\bar {\vec \omega}},
\dots,\vec x + \vec s_d + h_d{\bar {\vec \omega}}). \end{displaymath}
 Recall that any invertible linear transformation has constant Jacobian.
 We observe that
 \begin{displaymath} \vol_{d-1}(\conv(\vec b_1,\dots,\vec b_d)) = \int_{\conv(S)} \length((\vec x
+ \vec q + {\bar {\vec \omega}}\cdot \R) \cap \conv(\vec b_1,\dots,\vec b_d)) \d
\vec q. \end{displaymath}
 By Lemma \ref{lem:length-as-product} we find
 \begin{displaymath} \vol_{d-1}(\conv(\vec b_1,\dots,\vec b_d)) =
\abs{\sum_{i=1}^d z_i h_i} \int_{\conv(S)} \norm{C(\vec q)}_1 \d \vec q.
\end{displaymath}
The integral of $\norm{C(\vec q)}_1$ over $\conv(S)$ is independent of $\vec
x, h_1,\dots,h_d$. Thus, for fixed $\vec \theta \in \bS^{d-1},t > 0, S \in
\calS$, the random variables $\vec x, h_1,\dots,h_d$ have joint probability
density proportional to
\begin{displaymath} \abs{\sum_{i=1}^d z_i h_i} \cdot \prod_{i=1}^d \bar
\mu_i(\vec x + \vec s_i + h_i{\bar {\vec \omega}}).\end{displaymath}
\end{proof}

Recall that $\bar l = \vecb p + {\bar {\vec \omega}} \cdot \R$. The event $\bar E$
that $\conv(\vec b_1,\dots,\vec b_d) \cap \bar l \neq \emptyset$ occurs if and
only if $\vecb p \in \vec x + \conv(S)$, hence if and only if $\vecb p - \vec x \in \conv(S)$.

\begin{lemma} \label{lem:factorizeedgelength}
Let $\vec \theta \in \bbS^{d-1}, t > 0, S \in \calS$ be fixed, and have random
variables
$\vec b_1,\dots, \vec b_d \in \R^{d-1}$,$h_1,\dots,h_d \in \R$, $\vec x \in
\omega^\perp$ be distributed as in
Lemma~\ref{lem:factorizevolume}. Define $\vec q := \vecb p
- \vec x$. Then, the expected edge length satisfies
\begin{align*}
 \E[\length(\conv(\vec b_1,\dots,\vec b_d)\cap \bar l) \mid \vec \theta, t, S,
\bar E] \geq &\E[\norm{C(\vec q)}_1\mid \vec \theta, t, S, \bar E] \\
 &\cdot \inf_{\vec x \in \bar{\omega}^\perp} \E[\abs{\sum_{i=1}^d z_i h_i} \mid \vec \theta, t, S, \vec x].
\end{align*}
\end{lemma}
\begin{proof}
 We start with the assertion of Lemma \ref{lem:length-as-product}:
 \[\length((\vec x + \vec q + \bar{\vec \omega}\cdot \R)\cap \conv(\vec b_1,\dots,\vec b_d)) =
\norm{C(\vec q)}_1 \cdot \abs{\sum_{i=1}^d z_i h_i}.\]
 We take expectation on both sides to derive the equality
 \begin{align*}
  \E[\length(\conv(\vec b_1,\dots,\vec b_d)\cap \bar l) \mid \vec \theta, t, S,
\bar E] &= \E[\norm{C(\vec q)}_1 \cdot \abs{\sum_{i=1}^d z_i h_i} \mid \vec
\theta, t, S, \bar E].
 \end{align*}
 Since $\norm{C(\vec q)}_1$ and $\abs{\sum_{i=1}^d z_i h_i}$ do not share any of their variables, we separate the two expectations:
 \begin{align*}
  \E[\norm{C(\vec q)}_1 \cdot \abs{\sum_{\mathclap{i=1}}^d z_i h_i} \mid\! \vec
\theta, t, S, \bar E] &= \E_{\vec x, h_1,\dots,h_d}[\norm{C(\vec q)}_1 \cdot
\abs{\sum_{i=1}^d z_i h_i} \mid\! \vec \theta, t, S,\bar E]  \\
  &\hspace{-3em}= \E_{\vec x}[\norm{C(\vec q)}_1 \E_{h_1,\dots,h_d}[\abs{\sum_{i=1}^d z_i
h_i} \mid \vec \theta, t, S, \vec x] \mid\! \vec \theta, t, S, \bar E] \\
  &\hspace{-3em}\geq \E_{\vec x}[\norm{C(\vec q)}_1\mid\! \vec \theta, t, S, \bar E] \inf_{\!\vec
x \in \bar{\vec \omega}^\perp\!\!} \E_{h_1,\dots,h_d}[\abs{\sum_{\mathclap{i=1}}^d z_i h_i} \mid\! \vec \theta, t, S, \vec x].
\end{align*}
~
\end{proof}

We will first bound the expected $\ell_1$-diameter of $C(\vec q)$, where $\vec q
= \vecb p - \vec x$, which depends on where $\vecb p - \vec x$ intersects the
projected simplex $\conv(S)$: where this quantity tends to get smaller as we
approach to boundary of $\conv(S)$. We recall that $\bar E$ occurs if and only
if $\vec q \in \conv(S)$.

\index{convex combination space} \index{chord combinations}
\begin{lemma}[Chord combination bound] \label{lem:convexcombinationspace} Let $\vec \theta \in
\bbS^{d-1}, t > 0$ and $S \in \calS$ be fixed. Let $\vec q = \vecb p - \vec
x$ be distributed as in Lemma~\ref{lem:factorizeedgelength}. Then, the expected $\ell_1$-diameter of $C(\vec q)$ satisfies
 \[\E[\norm{C(\vec q)}_1  \mid \vec \theta,t,S,\bar E] \geq \frac{e^{-2}}{dL(1+\cutoffnorm)}\]
\end{lemma}\index{cutoff norm}\index{log-Lipschitz}
\begin{proof}
\index{$\hat \mu$}
To get a lower bound on the expected value of $\norm{C(\vec q)}_1$, we will use the
concavity of
$\norm{C(\vec q)}_1$ over $\conv(S) = \conv(\vec s_1,\dots, \vec
s_d)$ and that $\max_{\vec q \in \conv(S)}
\norm{C(\vec q)}_1 = 2$. These facts are proven in Lemma~\ref{lem:chord-combinations}.
We show that shifting the projected simplex does not change the
probability density too much (using log-Lipschitzness), and use the properties of
$\norm{C(\vec q)}_1$ mentioned above.

Let $\hat \mu$ denote the probability density of $\vec q$ conditioned on $\vec
\theta,t,S,\bar E$. Note that $\hat \mu$ is supported on $\conv(S)$ and has
density proportional to \[\dotsint \prod_{i=1}^d \bar \mu_i(\bar{\vec p} -\vec q
+ \vec s_i + h_i {\bar {\vec \omega}}) \d h_1 \cdots \d h_d.\] We claim that
$\hat \mu$ is $dL$-log-Lipschitz. To see this, note that since $\bar
\mu_1,\dots,\bar \mu_d$ are $L$-log-Lipschitz, for $\vec v, \vec v' \in \conv(S)$
we have that
\begin{align*}
&\dotsint \prod_{i=1}^d \bar \mu_i(\vecb p -\vec v + \vec s_i + h_i {\bar {\vec
\omega}}) \d h_1 \cdots \d h_d \\
&\leq \dotsint
\prod_{i=1}^d e^{L\norm{\vec v'-\vec v}} \bar \mu_i(\vecb p -\vec v' + \vec s_i + h_i {\bar {\vec
\omega}}) \d h_1 \cdots \d h_d \\
&= e^{d L \norm{\vec v'-\vec v}} \dotsint \prod_{i=1}^d \bar \mu_i(\vecb p -\vec v' + \vec s_i + h_i {\bar {\vec
\omega}}) \d h_1 \cdots \d h_d \text{ , as needed .}
\end{align*}

Let $\alpha \in (0,1)$ be a scale factor to be chosen later, and let $\vec y = \vec y(S) \in \conv(S)$ be as in Lemma~\ref{lem:chord-combinations}. Now we can write
\begin{align}
  \E[\norm{C(\vec q)} \mid \vec \theta,t,S,\bar E] &=
\int_{\conv(S)} \norm{C(\vec q)}_1 \hat \mu(\vec q) \d \vec q  \nonumber \\
  &\geq \int_{\alpha\conv(S)+(1-\alpha)\vec y} \norm{C(\vec q)}_1\hat \mu(\vec
q) \d \vec q,\label{eq:comb-space-1}\\
  \intertext{because the integrand is non-negative. By concavity of $\norm{C(\vec q)}_1$
we have the lower bound $\norm{C(\alpha \vec q + (1-\alpha)\vec y)} \geq 2(1-\alpha)$ for
all $\vec q \in \conv(S)$. Therefore,~\eqref{eq:comb-space-1} is lower bounded by }
  &\geq \int_{\alpha \conv(S) + (1-\alpha) \vec y} 2(1-\alpha) \hat \mu(\vec q) \d
\vec q \nonumber\\
  &= 2\alpha^{d}(1-\alpha) \int_{\conv(S)}\hat \mu(\alpha \vec q + (1-\alpha)\vec y) \d \vec q
\nonumber\\
  &\geq 2\alpha^{d}(1-\alpha) e^{-\max_{\vec q \in \conv(S)} (1-\alpha)
\norm{\vec q -\vec y}\cdot
dL} \int_{\conv(S)} \hat \mu(\vec q) \d \vec q, \nonumber \\
  &= 2\alpha^{d}(1-\alpha) e^{-\max_{i \in [d]} (1-\alpha) \norm{\vec s_i-\vec y}\cdot
dL},\label{eq:comb-space-2}\\
  \intertext{where we used a change of variables in the first equality, the
$dL$-log-Lipschitzness of $\hat \mu$ in the second inequality, and the convexity
of the $\ell_2$ norm in the last equality. Using
the diameter bound of $2+2\cutoffnorm$ for $\conv(S)$,~\eqref{eq:comb-space-2}
is lower bounded by}
  &\geq 2\alpha^{d}(1-\alpha) e^{-(1-\alpha) dL(2+2\cutoffnorm)} .
\label{eq:comb-space-3} \\
  \intertext{Setting $\alpha = 1-\frac{1}{dL(2+2\cutoffnorm)} \geq 1-1/d$ (by
Lemma~\ref{lem:lipschitz-vs-cutoff}) gives a lower bound for~\eqref{eq:comb-space-3} of}
  &\geq e^{-2} \frac{1}{dL(1+\cutoffnorm)}~. \nonumber
 \end{align}
 \end{proof}

Recall that we have now fixed the position $\vec x$ and shape $S$ of the
projected simplex. The randomness we have left is in the positions
$h_1,\dots,h_d$ of $\vec b_1,\dots,\vec b_d$ along lines parallel to the vector
${\bar {\vec \omega}}$. As $\vec \theta$ and $t$ are also fixed, restricting
$\vec b_i$ to lie on a line is the same as restricting $\vec a_i$ to lie on a
line.

Thus, were it not for the correlation between $h_1,\dots,h_d$, i.e.,~the factor
$\abs{\sum_{i=1}^d z_i h_i}$ in the joint probability density function, each
$h_i$ would be independent and have variance $\tau^2$ by assumption, and thus
one might expect $\E[\abs{\sum_{i=1}^d z_i h_i}] = \Omega(\norm{\vec z} \tau)$.
The following lemmas establish this, and show that in fact, the correlation term
only helps.

\begin{lemma} \label{lem:EXsquaredoverEX}
 Let $X$ be a random variable with $\E\left[X\right] = \mu$ and $\Var(X) = \tau^2$. Then $X$ satisfies
 \[\frac{\E\left[X^2\right]}{\E\left[\abs{X}\right]} \geq (\abs{\mu}+\tau)/2.\]
\end{lemma}
\begin{proof}
 By definition one has $\E\left[X^2\right] = \mu^2 + \tau^2$. We will show that
$\E\left[\abs{X}\right] \leq \abs{\mu} + \tau$ so that we can use the fact that
$\mu^2 + \tau^2 \geq 2\abs{\mu}\tau$ to derive that $\mu^2 + \tau^2 \geq
(\abs{\mu} + \tau)^2/2$. It then follows that
$\E\left[X^2\right]/\E\left[\abs{X}\right] \geq (\abs{\mu} + \tau)/2$.
 
 The expected absolute value $\E[\abs{X}]$ satisfies \[\E\left[\abs{X}\right]
\leq \abs{\mu} + \E\left[\abs{X -\mu}\right] \leq \abs{\mu} + \E\left[(X
-\mu)^2\right]^{1/2}\] by Cauchy-Schwarz, hence $\E\left[\abs{X}\right] \leq
\abs{\mu} + \tau$.
\end{proof}

\begin{lemma}[Height of simplex bound] \label{lem:heightofsimplex} \index{height of simplex} Let $\vec \theta \in \bbS^{d-1}, t
\geq 0$, $S \in \calS, \vec x \in \bar{\omega}^\perp$ be fixed and let $\vec z
:= \vec z(S)$ be as in Definition~\ref{def:kernel-combination}. Then for $h_1,\dots,h_d \in \R$ distributed as
in Lemma~\ref{lem:factorizeedgelength}, the expected inner product satisfies
\[
\inf_{\vec x \in \bar{\vec \omega}^\perp} \E[\abs{\sum_{i=1}^d z_i
h_i} \mid \vec \theta, t, S, \vec x] \geq \tau/(2\sqrt{d}).
\]
\end{lemma}\index{line variance}
\begin{proof}
For fixed $\theta,t,S,\vec x$, let $g_1,\dots,g_d \in \R$ be
independent random variables with respective probability densities $\tilde
\mu_1,\dots,\tilde \mu_d$, where $\tilde \mu_i$, $i \in [d]$, is defined by
\[
\tilde \mu_i(g_i) := \bar \mu(\vec x + \vec s_i + g_i{\bar {\vec \omega}}) =
\mu(R_{\vec \theta}(\vec x + \vec s_i + g_i{\bar {\vec \omega}}) + t \vec
\theta) \text{ .}
\]
Note that, by assumption, the variables $g_1,\dots,g_d$ each have variance at least $\tau^2$.
We recall from Lemma~\ref{lem:factorizevolume} that the joint probability density of
$h_1,\dots,h_d$ is proportional to $\abs{\sum_{i=1}^d z_i h_i} \prod_{i=1}^d
\tilde\mu_i(h_i)$. Thus, we may rewrite the above expectation as
  \begin{align*}
  \E[\abs{\sum_{i=1}^d z_i h_i} \mid \vec \theta, t, S, \vec x] &=
\frac{\dotsint_{\R} \abs{\sum_{i=1}^d z_i h_i}^2 \prod_{i=1}^d \tilde\mu_i(h_i)
\d h_1 \cdots \d h_d}{\dotsint_{\R} \abs{\sum_{i=1}^d z_i h_i} \prod_{i=1}^d
\tilde\mu_i(h_i) \d h_1 \cdots \d h_d} \\
&= \frac{\E[\abs{\sum_{i=1}^d z_i g_i}^2]}{\E[\abs{\sum_{i=1}^d z_i g_i}]}~,
\end{align*}
where $g_1,\dots,g_d$ are distributed independently with densities $\tilde \mu_1,\dots,\tilde \mu_d$.
By the additivity of variance for independent random variables, we see that
\[\Var(\sum_{i=1}^d z_i g_i) = \sum_{i=1}^d z_i^2 \Var(g_i) \geq \tau^2 \norm{\vec z}^2 \geq \tau^2
\norm{\vec z}^2_1/d = \tau^2/d .\]
We reach the desired conclusion by applying Lemma~\ref{lem:EXsquaredoverEX}:
\[
\frac{\E[\abs{\sum_{i=1}^d z_i g_i}^2]}{\E[\abs{\sum_{i=1}^d z_i g_i}]} \geq
\frac{\abs{\E[\sum_{i=1}^d z_i g_i]} + \sqrt{\Var(\sum_{i=1}^d z_i g_i)}}{2}
\geq \tau/(2\sqrt{d}).
\]

\end{proof}

Using the bounds from the preceding lemmas, the proof of our main theorem is now given below.

\begin{proof}[Proof of Theorem~\ref{thm:abstractedbound} (Parametrized Shadow Bound)]
By Lemma~\ref{lem:edgesasfraction},
we derive the shadow bound by combining
an upper bound on $\E[\perimeter(Q \cap W)]$ and a uniform lower
bound on $\E[\length(\conv(\vec a_i : i \in I)\cap W)
\mid E_I]$ for all $I \in B$. For the perimeter upper bound, by
Lemma~\ref{lem:perimeterbound} we have that
\begin{equation}
\label{eq:perim-bnd-abstract}
\E[\perimeter(Q \cap W)] \leq 2\pi(1+4\deviation) .
\end{equation}
For the edge length bound, we assume w.l.o.g. as above that $I = [d]$.
Combining prior lemmas, we have that
\begin{equation}
\label{eq:edge-bnd-abstract}
\begin{split}
&\E[\length(\conv(\vec a_1,\dots,\vec a_d) \cap W) \mid E] \\
&\geq \frac{1}{2} \cdot \E[\length(\conv(\vec a_1,\dots, \vec a_d) \cap W) \mid D,E]
\quad \left(\text{ Lemma~\ref{lem:boundedprojectedsimplex} }\right) \\
&\geq \frac{1}{2} \cdot \inf_{\substack{\vec \theta \in \bbS^{d-1} \\ t > 0}}
\E[\length(\conv(\vec b_1,\dots, \vec b_d) \cap \bar l) \mid \vec \theta,t,S \in
\calS,\bar{E}]
\hspace{1em} \left(\text{ Lemma~\ref{lem:conditiononhyperplane} }\right)  \\
&\geq \frac{1}{2} \cdot \inf_{\substack{\vec \theta \in \bbS^{d-1} \\ t > 0,S \in \calS}}
\left( \E[\norm{C(\bar{\vec p}-\vec x)}_1 \mid \vec \theta,t,S,\bar E]
\cdot \inf_{\vec x \in \bar{\vec \omega}^\perp} \E[\abs{\sum_{i=1}^d z_i h_i} \mid
\vec \theta,t,S,\vec x] \right)\\
&\hspace{27em} \left(\text{ Lemma~\ref{lem:factorizeedgelength} }\right)  \\
&\geq \frac{1}{2} \cdot \frac{e^{-2}}{dL(1+R_{n,d})} \cdot \frac{\tau}{2\sqrt{d}}
\quad \left(\text{ Lemmas~\ref{lem:convexcombinationspace}
and~\ref{lem:heightofsimplex} }\right) .
\end{split}
\end{equation}
The theorem now follows by taking the ratio of~\eqref{eq:perim-bnd-abstract} and~\eqref{eq:edge-bnd-abstract}.
\end{proof}

\subsection{Shadow bound for Laplace perturbations}
\label{sec:shadow-laplace}

Theorem~\ref{thm:abstractedbound} is most naturally used to prove shadow bounds
or distributions where all parameters are bounded, which we illustrate here
for Laplace-distributed perturbations. The Laplace distribution is
defined in section~\ref{sub:laplace}. To achieve the shadow bound, we
use the abstract shadow bound as a black box, and we bound the necessary
parameters of the Laplace distribution below.

\begin{lemma}\label{lem:propertiesofLaplace}
 For $n \geq d \geq 3$, the Laplace distribution $L_d(\vecb a, \sigma)$, satisfies the following properties:
 \begin{enumerate}
  \item The density is $\sqrt{d}/\sigma$-log-Lipschitz.
  \item Its cutoff radius satisfies $\cutoffnorm \leq 14\sigma \sqrt{d} \log n$.
  \item The $n$-th deviation satisfies $\deviation \leq 7\sigma \log n$.
  \item The variance after restricting to any line satisfies $\tau \geq \sigma/\sqrt{de}$.
 \end{enumerate}
\end{lemma} \index{cutoff norm}\index{n-th deviation@$n$-th deviation}\index{line variance}\index{log-Lipschitz}
\begin{proof} By shift invariance of the parameters, we may assume w.l.o.g.~that $\vecb a = \vec 0$.
Let $\vec X$ be distributed as $L_d(\vec 0,\sigma)$ for use below.

\paragraph{1.} The density of the Laplace distribution is proportional to
$e^{-\norm{\vec x}\sqrt{d}/\sigma}$, for $\vec x \in \R^d$, and thus the logarithm of the density differs an additive constant from $-\norm{\vec x}\sqrt{d}/\sigma$, which is clearly $\sqrt{d}/\sigma$-Lipschitz.
 
\paragraph{2.} The second property follows from Lemma~\ref{lem:laplace-tails}:
\begin{align*}
  \Pr[\norm{\vec X} \geq 14 \sigma \sqrt{d} \log n]
&\leq e^{-2 d \log n} = n^{-2d} \\ & \leq \frac{1}{d\binom{n}{d}}.
\end{align*}

\paragraph{3.} Again from Lemma~\ref{lem:laplace-tails}. If $7\log n \geq
2\sqrt{d}$, we get that
 \begin{align*}
  \int_{7\sigma \log n}^\infty \Pr[\abs{\inner{\vec X}{\vec \theta}} \geq t] \d t &\leq \int_{7\sigma \log n}^\infty e^{-\sqrt{d} t/(7\sigma)} \d t \\
  &= \frac{7\sigma}{\sqrt{d}} n^{-\sqrt{d}\log n} \leq \frac{7 \sigma \log n}{n}.
 \end{align*}
 If $7\log n \leq 2\sqrt{d}$, then
 \begin{align*}
  \int_{7\sigma \log n}^\infty \Pr[\abs{\inner{\vec X}{\vec \theta}} \geq t] \d t
&= \int_{7\sigma \log n}^{2\sigma \sqrt{d}}  \Pr[\abs{\inner{\vec X}{\vec \theta}}
\geq t] \d t +  \int_{2\sigma \sqrt{d}}^\infty  \Pr[\abs{\inner{\vec X}{\vec
\theta}} \geq t] \d t \\
  &\leq \int_{7\sigma \log n}^{2\sigma \sqrt{d}} 2 e^{-t^2/(16\sigma^2)} \d t +  \int_{2\sigma \sqrt{d}}^\infty  e^{-\sqrt{d} t/(7\sigma)} \d t \\
  &\leq 4\sigma \sqrt{d} e^{-(7\log n)^2/16} + \frac{7\sigma}{\sqrt{d}} e^{-2d/7} \\
  &\leq 4\sigma \sqrt{d} / n^3 + 7\sigma/(\sqrt{d} n^{\sqrt{d}}) \leq \frac{7\sigma \log n}{n} .
 \end{align*}
 
\paragraph{4.} This follows from the $\sqrt{d}/\sigma$-log-Lipschitzness and Lemma~\ref{lem:linevariance-vs-lipschitz}.
\end{proof}

\begin{proof}[Proof of Theorem \ref{thm:laplace} (Shadow bound for Laplace perturbations)]
We get the desired result by plugging in the bounds from Lemma~\ref{lem:propertiesofLaplace}
for $L,\cutoffnorm,\deviation$ and $\tau$ into the upper bound
$O((d^{1.5}L/\tau)(1+\cutoffnorm)(1+\deviation))$ from Theorem~\ref{thm:abstractedbound}.
\end{proof}

\subsection{Shadow bound for Gaussian perturbations}
\label{sec:shadow-gaussian}
\index{Gaussian shadow bound}
In this subsection, we prove our shadow bound for Gaussian perturbations.

The Gaussian distribution is not log-Lipschitz, so we can not directly apply
Theorem \ref{thm:abstractedbound}. We will define a \emph{smoothed out} version
of the Gaussian distribution to remedy this problem, which we call the
Laplace-Gaussian distribution. The Laplace-Gaussian distribution, defined below,
matches the Gaussian distribution in every meaningful parameter,
while also being log-Lipschitz. We will first bound the shadow size for
Laplace-Gaussian perturbations, and then show that the expected number of
edges of $Q \cap W$ for Gaussian perturbations is at most $1$ larger.

\begin{definition}
We define a random variable $\vec X \in \R^d$ to be
\emph{$(\sigma,r)$-Laplace-Gaussian distributed with mean $\vecb a$}, or
$\vec X \sim LG_d(\vecb a, \sigma,r)$, if its density is proportional to $f_{(\vecb a, \sigma,r)} : \R^d \to \R_+$ given by
 \[f_{(\vecb a, \sigma,r)}(\vec x) = \begin{cases}
 e^{-\norm{\vec x - \vecb a}^2/(2\sigma^2)} & \mbox{ if } \norm{\vec x - \vecb a} \leq r\sigma \\
 e^{-\norm{\vec x - \vecb a}r/\sigma + r^2/2} & \mbox{ if } \norm{\vec x - \vecb a} \geq r\sigma .
 \end{cases}\]
Note that at $\norm{\vec x - \vecb a} = r\sigma$, both cases give the density
$e^{-r^2/2}$, and hence $f_{(\vecb a,\sigma,r)}$ is well-defined and continuous
on $\R^d$. For distributions with mean $\vec 0$, we abbreviate $f_{(\sigma, r)} := f_{(\vec
0, \sigma, r)}$ and $LG_d(\sigma,r) := LG_d(\vec 0,\sigma,r)$. \index{Laplace-Gaussian distribution} \index{$LG_d(\vecb a, \sigma, r)$|see {Laplace-Gaussian distribution}}
\end{definition}

Just like for the shadow size bound for Laplace perturbations, we need strong
enough tail bounds. We state these tail bounds here, and defer their proofs till
the end of the section.

\begin{lemma}[Laplace-Gaussian tail bounds]  \label{lem:lg-tails}
Let $\vec{X} \in \R^d$ be $(\sigma, r)$-Laplace-Gaussian distributed with
mean $\vec 0$, where $r:= c\sqrt{d \log n}$, $c \geq 4$. Then for $t \geq r$,
\begin{equation}
\label{eq:lg-full-d}
\Pr[\norm{\vec{X}} \geq \sigma t] \leq e^{-(1/4)r t} \text{ .}
\end{equation}
For $\vec\theta \in \bS^{d-1}$, $t \geq 0$,
\begin{equation}
\label{eq:lg-1d}
\Pr[\abs{\pr{\vec X}{\vec \theta}} \geq \sigma t] \leq
\begin{cases}
e^{-(1/4)r t} &: t \geq r \\
3e^{-t^2/4} &: 0 \leq t \leq r.
\end{cases}
\end{equation}
\end{lemma}

\begin{lemma} \label{lem:propertiesofLG}
 For $n \geq d \geq 3$, the $(\sigma,4\sqrt{d\log n})$-Laplace-Gaussian
distribution in $\R^d$ with mean $\vecb a$ satisfies the following properties:
 \begin{enumerate}
  \item The density is $4\sigma^{-1}\sqrt{d\log n}$-log-Lipschitz.
  \item Its cutoff radius satisfies $\cutoffnorm \leq 4\sigma \sqrt{d\log n}$.
  \item The $n$-th deviation is $\deviation \leq 4\sigma \sqrt{\log n}$.
  \item The variance after restricting to any line satisfies $\tau \geq \sigma/4$.
 \end{enumerate}
\end{lemma} \index{cutoff norm}\index{n-th deviation@$n$-th deviation}\index{line variance}\index{log-Lipschitz}
\begin{proof} As before, by shift invariance, we may assume w.l.o.g~that $\vecb
a = \vec 0$. Let $\vec X \sim LG_d(\sigma,4\sqrt{d\log n})$ and
let $r := 4\sqrt{d \log n}$.

\paragraph{1.} The gradient of the function $\log(f_{(\sigma, r)}(\vec x))$ has
norm bounded by $4\sigma^{-1}\sqrt{d\log n}$ wherever it is
defined, which by continuity implies $f_{(\sigma,r)}$ is $4\sigma^{-1}\sqrt{d\log n}$-log-Lipschitz.

\paragraph{2.} Applying the tail bound from Lemma~\ref{lem:lg-tails}, we get that
\begin{align*}
\Pr[\norm{\vec X} \geq 4\sigma \sqrt{d\log n}] &\leq e^{-4d\log n} \leq \frac{1}{d \binom{n}{d}}.
\end{align*}

\paragraph{3.} Again using Lemma~\ref{lem:lg-tails},
 \begin{align*}
  \int_{4\sigma \sqrt{\log n}}^\infty \Pr[\abs{\inner{\vec X}{\vec \theta}} \geq t] \d t
  &= \int_{4\sigma \sqrt{\log n}}^{r\sigma}  \Pr[\abs{\inner{\vec X}{\vec
\theta}} \geq t] \d t +  \int_{r\sigma}^\infty
\Pr[\abs{\inner{\vec X}{\vec \theta}} \geq t] \d t \\
  &\leq \int_{4\sigma \sqrt{\log n}}^{r\sigma}
3e^{-t^2/(4\sigma^2)} \d t + \int_{r\sigma}^\infty e^{-\sqrt{d\log n}t/\sigma} \d t \\
  &\leq 4\sigma \sqrt{d \log n} (3 n^{-4}) + \frac{\sigma}{\sqrt{d \log n}} n^{-4d} \\
  &\leq 4\sigma\sqrt{\log n}/n .
 \end{align*}

\paragraph{4.}  For the line variance, by rotational symmetry, we may
without loss of generality assume that $l := (\vec y, 0) + \vec e_d\R$, where $\vec y \in
\R^{d-1}$, and so $(\vec y, 0)$ is the point on $l$ closest to the origin. Since
$f_{(\sigma,r)}(\vec y,\lambda) = f_{(\sigma,r)}(\vec y,-\lambda)$ for
every $\lambda \in \R$, the expectation $\E[\vec X \mid \vec X \in l] = (\vec y,
0)$. Thus, $\Var(\vec X \mid \vec X \in l) = \E[X_d^2 \mid \vec X \in l]$.

Let $\bar{l} = (\vec y, 0) + [-\sigma,\sigma] \cdot \vec e_d$. Since $\abs{X_d}$ is
larger on $l \setminus \bar{l}$ than on $\bar{l}$, we clearly have $\E[X_d^2
\mid \vec X \in l] \geq \E[X_d^2 \mid \vec X \in \bar l]$, so it suffices to
lower bound the latter quantity.

For each $\vec y$ with $\norm{\vec y} \leq \sigma r$ we have for all $\lambda \in
[-\sigma,\sigma]$ the inequality
\begin{equation}
\label{eq:lg-var-1}
1 \geq \frac{f_{(\sigma,r)}(\vec
y,\lambda)}{f_{(\sigma,r)}(\vec y, 0)} \geq \frac{e^{-\norm{(\vec
y, \lambda)}^2/(2\sigma^2)}}{e^{-\norm{(\vec y, 0)}^2/(2\sigma^2)}} =
e^{-\lambda^2/(2\sigma^2)} \geq e^{-1/2} \text{ .}
\end{equation}
Given the above, we have that
\begin{equation}
\label{eq:lg-var-2}
\begin{split}
\E[X_d^2 \mid \vec X \in \bar l] &\geq (\sigma^2/4) \Pr[|X_d| \geq \sigma/2 \mid
\vec X \in \bar l] \\
&= (\sigma^2/4) \frac{\int_{\sigma/2}^\sigma f_{(\sigma,r)}(\vec y, t)
\d t}{\int_0^\sigma f_{(\sigma,r)}(\vec y,t) \d t} \\
&\geq (\sigma^2/4) \frac{\int_{\sigma/2}^\sigma f_{(\sigma,r)}(\vec y, 0)
e^{-1/2} \d t}{\int_0^\sigma f_{(\sigma,r)}(\vec y,0) \d t} \quad \left(\text{ by~\eqref{eq:lg-var-1} } \right) \\
&= (\sigma^2/4)(e^{-1/2}/2) \geq \sigma^2/16~\text{ , as needed .}
\end{split}
\end{equation}
For $\vec y$ with $\norm{\vec y} \geq \sigma r$, $\lambda \in [-\sigma,\sigma]$, we similarly have
\begin{align*}
\norm{(\vec y,\lambda)} &= \sqrt{\norm{\vec y}^2 + \lambda^2} \\
&\leq \norm{\vec y} + \frac{\lambda^2}{2\norm{\vec y}}
\leq \norm{\vec y} + \frac{\lambda^2}{2r\sigma}.
\end{align*}
In particular, we get that
\begin{equation}
\label{eq:lg-var-3}
1 \geq \frac{f_{(\sigma,r)}(\vec
y,\lambda)}{f_{(\sigma,r)}(\vec y, 0)} = \frac{e^{-\norm{(\vec
y, \lambda)}(r/\sigma)}}{e^{-\norm{(\vec y,0)}(r/\sigma)}} \geq
e^{-\lambda^2/(2\sigma^2)} \geq e^{-1/2} \text{ .}
\end{equation}
The desired lower bound now follows by combining~\eqref{eq:lg-var-2},~\eqref{eq:lg-var-3}.
\end{proof}

Given any unperturbed unit LP given by $\vec c, \bar {\vec a}_1,\dots,\bar {\vec a}_n$,
we denote by $\E_{N_d(\sigma)}$ the expectation when its vertices are
perturbed with noise distributed according to the Gaussian distribution of
standard deviation $\sigma$ and we write $\E_{LG_d(\sigma,r)}$ for the
expectation when its vertices are perturbed by $(\sigma,r)$-Laplace-Gaussian
noise. In the next lemma we prove that, for $r := 4\sqrt{d\log n}$, the expected
number of edges for Gaussian distributed perturbations is not much bigger
than the expected number for Laplace-Gaussian perturbations. We use the
strong tail bounds we have on the two distributions along with the knowledge
that restricted to a ball of radius $r\sigma$ the probability densities are equal.
Recall that we use $\hat{\vec a}_i$ to denote the perturbation
$\vec a_i - \E[\vec a_i]$.

\begin{lemma} \label{lem:mightaswellLG}
For $d \geq 3$, the number of edges in $\conv(\vec a_1,\dots,\vec a_n) \cap W$ satisfies
\[
\E_{N_d(\sigma)}[\abs{\edges(\conv(\vec a_1,\dots,\vec a_n))}] \leq 1 +
\E_{LG_d(\sigma,4 \sqrt{ d \log n})}[\abs{\edges(\conv(\vec a_1,\dots,\vec a_n))}].
\]
\end{lemma}
\begin{proof}
Let us abbreviate $\edges := \edges(\conv(\vec a_1,\dots,\vec a_n))$ and let $r
:= 4\sqrt{d \log n}$. We make use of the fact that $N_d(\sigma)$ and
$LG_d(\sigma,r)$ are equal when restricted to distance at most
$\sigma r$ from their centers.
 \begin{align}
  \E_{N(\sigma)}[\abs{\edges}] &= \Pr_{N_d(\sigma)}[\exists i \in [n] ~
\norm{\hat{\vec a}_i} > \sigma r]\E_{N_d(\sigma)}[\abs{\edges}
\mid \exists i \in [n] ~ \norm{\hat{\vec a}_i} > \sigma r] \nonumber \\
  &\,\,+ \Pr_{N_d(\sigma)}[\forall i \in [n] ~ \norm{\hat{\vec a}_i} \leq
\sigma r]\E_{N_d(\sigma)}[\abs{\edges} \mid \forall i \in [n] ~
\norm{\hat{\vec a}_i} \leq \sigma r]. \label{eq:mas-lg-1} \\
  \intertext{
By Lemma~\ref{lem:lg-tails}, the first probability is at most $n^{-4d}
\leq n^{-d}/4$, so we upper bound the first number of edges by
$\binom{n}{d}$ making a total contribution of less than $1/4$. Now we use the
fact that within radius $4 \sigma \sqrt{ d \log n}$ we have equality of densities
between $N_d(\sigma)$ and $LG_d(\sigma,r)$. Continuing from~\eqref{eq:mas-lg-1},}
  &\leq 1/4 + \E_{N_d(\sigma)}[\abs{\edges} \mid \forall i \in [n] ~
\norm{\hat{\vec a}_i} \leq \sigma r] \nonumber \\
   &= 1/4 + \E_{LG_d(\sigma,r)}[\abs{\edges} \mid \forall i \in
[n] ~ \norm{\hat{\vec a}_i} \leq \sigma r] \nonumber \\
  &\leq 1/4 + \E_{LG_d(\sigma,r)}[\abs{\edges}]/\Pr_{LG_d(\sigma,
r)}[\forall i \in [n] ~ \norm{\hat{\vec a}_i} \leq \sigma r].  \label{eq:mas-lg-2} \\
  \intertext{The inequality above is true by non-negativity of the number of
edges. Next we lower bound the denominator and continue~\eqref{eq:mas-lg-2},}
  &\leq 1/4 + \E_{LG_d(\sigma,r)}[\abs{\edges}]/(1-n^{-d}/4) \nonumber \\
  &\leq 1/4 + (1+n^{-d}/2)\E_{LG_d(\sigma,r)}[\abs{\edges}]. \label{eq:mas-lg-3}
 \end{align}
 The last inequality we deduce from the fact that $(1-\epsilon)(1+2\epsilon) = 1
+ \epsilon - 2\epsilon^2$, which is bigger than $1$ for $0 < \epsilon < 1/2$. Again using the trivial upper bound of $\binom{n}{d}$
edges, we arrive at our desired conclusion that
\[
  \E_{N_d(\sigma)}[\abs{\edges}] \leq 1 + \E_{LG_d(\sigma,r)}[\abs{\edges}].
  \]
  
\end{proof}

We now have all the ingredients to prove our bound on the expected number of edges for Gaussian perturbations.
\begin{proof}[Proof of Theorem \ref{thm:gaussian} (Shadow bound for Gaussian perturbations)]~
By Lemma \ref{lem:mightaswellLG}, we know that
\[\E_{N_d(\sigma)}[\abs{\edges(\conv(\vec a_1,\dots,\vec a_n))}] \leq 1 +
\E_{LG_d(\sigma,4 \sqrt{ d \log n})}[\abs{\edges(\conv(\vec a_1,\dots,\vec
a_n))}].\] We now derive the shadow bound for Laplace-Gaussian perturbations by
combining the parameter bounds in Lemma~\ref{lem:propertiesofLG} with the
parameterized shadow bound in Theorem~\ref{thm:abstractedbound}.
\end{proof}

We now prove the tail bounds for Laplace-Gaussian distributions. Recall that
we set $r:= c\sqrt{d \log n}$ with $c \geq 4$.

\begin{proof}[Proof of Lemma~\ref{lem:lg-tails} (Tail bound for Laplace-Gaussian distribution)] By homogeneity, we may
w.l.o.g.~assume that $\sigma = 1$. Define
auxiliary random variables $\vec{Y} \in \R^d$ distributed as
$(\vec{0},1/(c\sqrt{\log n}))$-Laplace and $\vec{Z} \in \R^d$ be distributed as
$N_d(\vec{0},1)$.

Since $\vec{X}$ has density proportional to $f_{(1,r)}(\vec{x})$, which equals
$e^{-\norm{\vec{x}}^2/2}$ for $\norm{\vec{x}} \leq r$ and $e^{-r\norm{\vec{x}}+r^2/2}$ for
$\norm{\vec{x}} \geq r$, we immediately see that
\begin{equation}
\label{eq:lg-equiv}
\begin{split}
\vec{Z} \mid \norm{\vec{Z}} \leq r & \equiv \vec{X} \mid \norm{\vec{X}} \leq r \\
\vec{Y} \mid \norm{\vec{Y}} \geq r & \equiv \vec{X} \mid \norm{\vec{X}} \geq r
\end{split}
\end{equation}
 
\paragraph{Proof of~\eqref{eq:lg-full-d}} By the above, for any $t \geq r$,
we have that
\begin{equation}
\label{eq:tailbound-1}
\Pr[\norm{\vec{X}} \geq t] = \Pr[\norm{\vec{Y}} \geq t] \cdot \frac{\Pr[\norm{\vec{X}}
\geq r]}{\Pr[\norm{\vec{Y}} \geq r]} \text{ .}
\end{equation}
For the first term, by the Laplace tail bound~\eqref{eq:laplace-full-d}, we get that
\begin{equation}
\label{eq:lg-full-d-2}
\Pr[\norm{\vec{Y}} \geq t] \leq e^{-rt - d\log(\frac{c\sqrt{\log n}t}{\sqrt{d}})-d} \text{ .}
\end{equation}
For the second term,
\begin{equation}
\label{eq:lg-full-d-3}
\begin{split}
\frac{\Pr[\norm{\vec{X}} \geq r]}{\Pr[\norm{\vec{Y}} \geq r]}
&=  e^{r^2/2} \frac{\int_{\R^n} e^{-r\norm{\vec x}} \d \vec x}
                   {\int_{\R^n} f_{(\sigma,r)}(\vec x) \d \vec x}
\leq  e^{r^2/2} \frac{\int_{\R^n} e^{-r \norm{\vec x}}\d \vec x}
                      {\int_{\R^n} e^{-\norm{\vec{x}}^2/2} \d \vec x}  \\
&\leq e^{r^2/2}\frac{r^{-d}d!\vol_d(\calB_2^d)}{\sqrt{2\pi}^d} \leq
e^{(d c^2 \log n)/2 }(\frac{\sqrt{e}}{c \sqrt{\log n}})^d \\
&\leq e^{(d c^2 \log n)/2} \text{ ,}
\end{split}
\end{equation}
where we have used the upper bound $\vol_d(\calB_2^d) \leq (2\pi e/d)^{d/2}$, $r
= c\sqrt{d\log n}$ and $c \geq \sqrt{e}$. Combining~\eqref{eq:lg-full-d-2},~\eqref{eq:lg-full-d-3} and
that $t \geq r$, $c \geq 4$, we get
\begin{equation}
\label{eq:lg-full-d-4}
\begin{split}
\Pr[\norm{\vec{X}} \geq t] &\leq e^{-rt - d\log(\frac{c\sqrt{\log
n}t}{\sqrt{d}})-d} \cdot e^{(d c^2 \log n)/2} \\
&\leq e^{-rt/2 - d\log(\frac{c\sqrt{\log n}t}{\sqrt{d}})-d}
 = e^{-d(\frac{rt}{2d}-\log(\frac{rt}{d})-1)} \\
&\leq e^{-d(\frac{rt}{4d})} = e^{-rt/4} ,
\end{split}
\end{equation}
where the last inequality follows from $x/2-\log(x)-1 \geq x/4$, for $x \geq rt/d \geq c^2 \geq 16$.

\paragraph{Proof of~\eqref{eq:lg-1d}} For $t \geq r$, using the
bound~\eqref{eq:lg-full-d}, we get
\begin{equation}
\label{eq:lg-1d-1}
\Pr[\abs{\pr{\vec X}{\vec \theta}} \geq t] \leq \Pr[\norm{\vec{X}} \geq t] \leq e^{-c\sqrt{d \log
n } t/4} \text{ . }
\end{equation}
For $t \leq r$, we see that
\begin{equation}
\label{eq:lg-1d-2}
\begin{split}
\Pr[\abs{\pr{\vec X}{\vec \theta}} \geq t] &\leq \Pr[\abs{\pr{\vec X}{\vec \theta}} \geq t, \norm{\vec{X}} \leq r]
+ \Pr[\norm{\vec{X}} \geq r] \\
&\leq \Pr[\abs{\pr{\vec X}{\vec \theta}} \geq t, \norm{\vec{X}} \leq r] +
e^{-r^2/4} \text{ .}
\end{split}
\end{equation}
By the identity~\eqref{eq:lg-equiv}, for the first term, using the Gaussian
tail bound~\eqref{eq:gauss-1d}, we have that
\begin{equation}
\label{eq:lg-1d-3}
\begin{split}
\Pr[\abs{\pr{\vec X}{\vec \theta}} \geq t, \norm{\vec{X}} \leq r]
&= \Pr[\abs{\pr{\vec Z}{\vec \theta}} \geq t, \norm{\vec{Z}} \leq r] \cdot
\frac{\Pr[\norm{\vec{X}} \leq r]}{\Pr[\norm{\vec{Z}} \leq r]} \\
&= \Pr[\abs{\pr{\vec Z}{\vec \theta}} \geq t, \norm{\vec{Z}} \leq r] \cdot
\frac{\int_{\R^n} e^{-\norm{\vec x}^2/2} \d \vec x}{\int_{\R^n} f_{(1,r)}(\vec x) \d \vec x} \\
&\leq \Pr[\abs{\pr{\vec Z}{\vec \theta}} \geq t] \leq 2e^{-t^2/2} \text{ .}
\end{split}
\end{equation}
The desired inequality~\eqref{eq:lg-1d} now follows directly
by combining~\eqref{eq:lg-1d-1},~\eqref{eq:lg-1d-2},~\eqref{eq:lg-1d-3}, noting
that $2e^{-t^2/2} + e^{-r^2/4} \leq 3 e^{-t^2/4}$ for $0 \leq t \leq r$.
\end{proof}

%% file: algorithms.tex
\section{Simplex algorithms}
\label{sec:algorithms}

In this section, we describe how to use the shadow bound
to bound the complexity of a complete shadow vertex simplex algorithm.
We restrict our
attention here to Gaussian perturbations, as the details for Laplace
perturbations are similar. We will follow the two-stage
interpolation strategy given by Vershynin in~\cite{jour/siamjc/Vershynin09}.
The Random Vertex algorithm of \cite{jour/siamjc/Vershynin09}
was shown in the same paper to work for any
$\sigma \leq \min(\frac{c_1}{\sqrt{d\log n}}, \frac{c_1}{d^{3/2}\log d})$
for some $c_1 > 0$.
The constraint on $\sigma$ is always achievable by scaling down the matrix
$\vec A$, though it will be reflected in the running time of the algorithm.

We will describe a modification of the RV algorithm that further relaxes the
condition on the perturbation size to
$\sigma \leq O(\frac{1}{\sqrt{d\log n}})$, for an expected
$O(d^2\sqrt{\log n}~\sigma^{-2} + d^3\log^{1.5} n)$ pivot steps. Hence our
algorithm is faster than both Vershynin's~\cite{jour/siamjc/Vershynin09} RV algorithm and Borgwardt's
dimension-by-dimension algorithm~\cite{Borgwardt87} for our shadow bound.

\noindent To recall, our goal is to solve the smoothed LP

\begin{align*}
\max & ~ \vec{c}^\T \vec{x} \hspace{11em} \text{(Smooth LP)} \\
     & ~ \vec{A} \vec{x} \leq \vec{b}
\end{align*}
where $\vec{A} \in \R^{n \times d}$, $\vec{b} \in \R^n$, $\vec c \in \R^{d}
\setminus \set{\vec 0}$ and $n \geq d \geq 3$.
Here each row $(\vec{a}_i,b_i)$, $i \in [n]$, of $(\vec{A},\vec{b})$ is a
variance $\sigma^2$ Gaussian random vector with mean $(\vecb{a}_i,\bar{b}_i) :=
\E[(\vec{a}_i,b_i)]$ of $\ell_2$ norm at most $1$. We will say that (Smooth LP)
is unbounded (bounded) if the system $\vec{c}^\T \vec{x} > 0, \vec{A} \vec{x}
\leq \vec 0$ is feasible (infeasible). Note that (Smooth LP) can be unbounded
and infeasible under this definition. If (Smooth LP) is bounded and
feasible, then it has an optimal solution.

For the execution of the algorithms as stated, we assume the non-degeneracy
conditions listed in Theorem~\ref{thm:geometric-characterization-shadow-path-size}.
That is, we assume both the feasible polyhedron and shadows to be non-degenerate.
These conditions hold with probability $1$.

\begin{theorem}\label{thm:algorithms} (Smooth LP) can be solved by a two-phase shadow
simplex method using an expected number of pivots of
$O(d^2\sqrt{\log n}~\sigma^{-2} + d^3\log^{1.5} n)$.
\end{theorem}
\begin{proof}
Combining Lemma~\ref{lem:smooth-via-int} and
Theorem~\ref{thm:success-bnd}, the expected number of simplex pivots is bounded by
\[
10 + \calD_g(d+1,n,\sigma/2) + 5\calD_g(d,n+2d-2,\min \{\sigma,\bar{\sigma}\}/5) \text{ ,}
\] 
where $\bar{\sigma}$ is as defined in \eqref{eq:sigmabar}. Noting that $1/\bar{\sigma} =
O(\sqrt{d\log n})$, by the smoothed Gaussian shadow
bound~(Theorem \ref{thm:gaussian}), the above is bounded by
\[
O(\calD_g(d,n,\sigma) + \calD_g(d,n,(\sqrt{d\log n})^{-1}))
=O(d^2 \sqrt{\log n} \sigma^{-2} + d^3 \log^{1.5} n) \text{ ,}
\]
as needed.
\end{proof}

\paragraph{Two-Phase Interpolation Method} Define the Phase I Unit LP:
\begin{align*}
\max & ~ \vec{c}^\T \vec{x} \hspace{10em} \text{(Unit LP)} \\
     & ~ \vec{A} \vec{x} \leq \vec{1}
\end{align*}
and the Phase II interpolation LP with parametric objective for $\theta \in (-\pi/2,\pi/2)$:
\begin{align*}
\max & ~ \cos(\theta) \vec{c}^\T \vec{x} + \sin(\theta) \lambda \hspace{3em} \text{(Int.~LP)} \\
     & ~ \vec{A} \vec{x} + (\vec{1}-\vec{b}) \lambda \leq \vec{1} \\
     & ~0 \leq \lambda \leq 1.
\end{align*}
The above form of interpolation was first introduced in the context of smoothed
analysis by Vershynin~\cite{jour/siamjc/Vershynin09}.

Let us assume for the moment that (Smooth LP) is bounded and feasible (i.e.,~has
an optimal solution). Since boundedness is a property of $\vec{A}$ and not
$\vec{b}$, note that this implies that (Unit LP) is also bounded (and clearly
always feasible).

To understand the Phase II interpolation LP, the key observation is that for $\theta$
sufficiently close to $-\pi/2$, the maximizer will be the optimal solution to
(Unit LP), i.e.,~will satisfy $\lambda=0$, and for $\theta$ sufficiently close to
$\pi/2$ the maximizer will be the optimal solution to (Smooth LP), i.e.,~will
satisfy $\lambda=1$. Thus given an optimal solution to the Phase I unit LP one
can initialize a run of shadow vertex starting at $\theta$ just above $-\pi/2$,
moving towards $\pi/2$ until the optimal solution to (Smooth LP) is found. The
corresponding shadow plane is generated by $(\vec{c},0)$ and $(\vec{0},1)$
(associating $\lambda$ with the last coordinate), and as usual the size of the
shadow bounds the number of pivots.

If (Smooth LP) is unbounded (i.e.,~the system $\vec{c}^\T \vec{x} > 0,
\vec{A}\vec{x} \leq \vec{0}$ is feasible), this will be detected during Phase I
as (Unit LP) is also unbounded. If (Smooth LP) is infeasible but bounded, then
the shadow vertex run will terminate at a vertex having $\lambda < 1$. Thus,
all cases can be detected by the two-phase procedure (see~\cite[Proposition
4.1]{jour/siamjc/Vershynin09} for a formal proof).

We bound the number of pivot steps taken to solve (Int.~LP) given a solution to (Unit LP), and after that we describe how to solve (Unit LP).

Consider polyhedron $P' = \set{(\vec x,\lambda) \in \R^{d+1} : \vec A \vec x + (\vec 1 -
\vec b)\lambda \leq \vec 1}$, the slab $H = \set{(\vec x,\lambda) \in
\R^{d+1} : 0 \leq \lambda \leq 1}$ and let $W = \linsp(\vec c, \vec
e_\lambda)$. In this notation, $P'\cap H$ is the feasible set of (Int.~LP)
and $W$ is the shadow plane of (Int.~LP). We bound the number of vertices
in the shadow $\pi_W(P'\cap H)$ of (Int.~LP) by relating it to $\pi_W(P')$.

The constraints defining $P'$ are of smoothed unit type. Namely, the rows of
$(\vec{A},\vec{1}-\vec{b})$ are variance $\sigma^2$ Gaussians centered at means of
norm at most $2$. We derive this from the triangle inequality. Thus,
we know $\pi_W(P')$ has at most $\calD_g(d+1,n,\sigma/2)$ expected vertices.
We divide $\sigma$ by $2$ because the centers have norm at most $2$.

Since the shadow plane contains the normal vector $(\vec{0},1)$ to the inequalities
$0 \leq \lambda \leq 1$, these constraints intersect the shadow plane $W$ at right angles.
It follows that $\pi_W(P'\cap H) = \pi_W(P') \cap H$. Adding $2$
constraints to a 2D polyhedron can add at most $2$ new edges, hence
the constraints on $\lambda$ can add at most $4$ new vertices.
By combining these observations, we directly derive the following
lemma of Vershynin~\cite{jour/siamjc/Vershynin09}.

\begin{lemma}
\label{lem:smooth-via-int}
If (Unit LP) is unbounded, then (Smooth LP) is unbounded. If (Unit LP) is bounded,
then given an optimal solution to (Unit LP) one can solve (Smooth LP) using at
most an expected $\calD_g(d+1,n,\sigma/2)+4$ shadow vertex pivots over (Int.~LP).
\end{lemma}

Given the above, our main task is now to solve (Unit LP), i.e.,~either to find an
optimal solution or to determine unboundedness.  The simplest algorithm is
Borgwardt's dimension-by-dimension (DD) algorithm, which was first used
in the context of smoothed analysis by Schnalzger~\cite{thesis/Schnalzger14}.
Due to its simplicity, we describe it briefly below as a warm-up.

\paragraph{DD algorithm}As outlined in the introduction, the DD algorithm
solves Unit LP by iteratively solving the restrictions:
\begin{align*}
\max &~\vec{c}^\T_k \vec{x} \hspace{5em} \text{(Unit LP$_k$)} \\
     &~\vec{A}\vec{x} \leq \vec 1 \\
     &~x_i = 0, ~\forall i \in \set{k+1,\dots,d}\text{, }
\end{align*}
where $k \in \set{1,\dots,d}$ and $\vec{c}_k := (c_1,\dots,c_k,0,\dots,0)$. The
main idea here is that the solution of (Unit LP$_k$), $k \in \set{1,\dots,d-1}$,
is generically on an edge of the shadow of (Unit LP$_{k+1}$) on the span of
$\vec{c}_k$ and $\vec{e}_{k+1}$, which is sufficient to initialize the shadow
simplex path in the next step. We note that Borgwardt's algorithm can be applied
to any LP with a known feasible point as long as appropriate
non-degeneracy conditions hold (which occur with probability $1$ for smoothed LPs). To avoid
degeneracy, we will assume that $\vec{c}_k \neq \vec{0}$ for all $k \in
\set{1,\dots,d}$, which can always be achieved by permuting the coordinates.
Note that (Unit LP$_1$) can be trivially solved, as the feasible region is an interval
whose endpoints are easy to compute.

\begin{theorem}[\cite{jour/zor/Borgwardt82}]
\label{thm:borgwardt-dd}
Let $W_k$, $k \in \{2,\dots,d\}$, denote the shadow of (Unit LP$_k$) on the span
of $\vec{c}_{k-1}$ and $\vec{e}_k$. Then, if each (Unit LP$_k$) and
shadow $W_k$ is non-degenerate, for $k \in \set{2,\dots,d}$, there exists a shadow
simplex method which solves (Unit LP) using at most $\sum_{k=2}^d \abs{\vertices(W_k)}$ number of pivots.
\end{theorem}

The shadow bounds of Theorem~\ref{thm:gaussian} only hold for $d \geq 3$,
though sufficiently good bounds have been proven by
\cite{thesis/Schnalzger14,devillers2016smoothed}
to derive the following immediate corollary.

\begin{corollary}
\label{cor:dd-phaseI}
The smoothed (Unit LP) can be solved by the DD algorithm using an expected $\sum_{k=2}^d
\calD_g(k,n,\sigma) = O(d^3\sqrt{\log n}~\sigma^{-2} + d^{3.5}\sigma^{-1}\log n + d^{3.5}\log^{3/2} n)$ number of shadow vertex pivots.
\end{corollary}

\paragraph{Random vertex method}
Vershynin's approach for initializing the shadow
simplex method on (Unit LP) is to add a random smoothed system of $d$ linear constraints
to its description. These constraints are meant to induce a
\emph{known random} vertex $\vec{v}$ and corresponding maximizing objective
$\vec{d}$ which are effectively uncorrelated with the original system. Starting
at this vertex $\vec{v}$, we then follow the shadow path induced by rotating $\vec{d}$
towards $\vec{c}$. The main difficulty with this approach is to guarantee that
the randomly generated system (i) adds a vertex which (ii) is optimized at $\vec d$ and
(iii) does not cut off the optimal solution or all unbounded rays. Fortunately, each of
these conditions is easily checkable, and hence if they fail (which will occur with
constant probability), the process can be attempted again.

One restriction imposed by this approach is that the perturbation size needs to
be rather small, namely
\begin{equation*}
\sigma \leq \sigma_1 := \frac{c_1}{\max \set{\sqrt{d \log n}, d^{1.5}\log d}}
\end{equation*}
in~\cite{jour/siamjc/Vershynin09} for some $c_1 > 0$. A more careful analysis of Vershynin's algorithm
can relax the restriction to
\begin{equation*}
\sigma \leq \sigma_2 := \frac{c_2}{\max \set{\sqrt{d \log n}, \sqrt{d} \log d}}
\end{equation*}
for some $c_2 > 0$. This restriction is necessary due to the fact that we wish
to predict the effect of smoothing the system, in particular, the smoothing
operation should not negate (i), (ii), or (iii).
Recall that one can always artificially decrease $\sigma$ by
scaling down the matrix $\vec{A}$ as this does not change the structure of (Unit
LP). The assumption on $\sigma$ is thus without loss of generality. When
stating running time bounds however, this restriction will be reflected by a
larger additive term that does not depend on $\sigma$.

We adapt the Random Vertex algorithm to make (ii) guaranteed to hold,
allowing us to relax the constraint on the perturbation size to
\begin{equation}
\label{eq:sigmabar}
\sigma \leq \bar{\sigma} := \frac{1}{36\sqrt{d \log n}}.
\end{equation}
Instead of adding $d$ constraints, each with their own perturbation,
we add $d-1$ pairs of constraints with mirrored perturbations.
This forces the desired objective to be maximized at the random vertex
whenever this vertex exists.

\begin{algorithm}
\caption{Symmetric Random Vertex algorithm}
\label{alg:upgraded-random-vertex}
\begin{algorithmic}[1]
	\REQUIRE $\vec{c} \in \R^d \setminus \set{\vec 0}$, $\vec{A} \in \R^{n \times d}$,
	$\vec{A}$ is standard deviation $\sigma \leq \bar{\sigma}$ Gaussian with
	rows having centers of norm at most $1$.
	\ENSURE Decide whether (Unit LP) $\max \vec{c}^\T \vec{x}, \vec{A} \vec{x} \leq \vec 1$ is
unbounded or return an optimal solution.
\STATE If some row of $\vec A$ has norm greater than $2$,
    solve $\max \vec c^\T\vec x$,st.$\vec A\vec x \leq \vec 1$ using
    any simplex method that takes at most $\binom{n}{d}$ pivot steps.
\LOOP
\STATE Let $l = 1/6\sqrt{\log d}$.
\STATE Sample a rotation matrix $\vec R \in O(d)$ uniformly at random.
\STATE Sample $\vec g_1,\dots,\vec g_{d-1} \sim N(\vec 0,\sigma^2\vec I)$ independently.
\STATE Set $\vec v_i^+ = \vec R(4\vec e_d + l\vec e_i + \vec g_i)$,
$\vec v_i^- = \vec R(4\vec e_d - l\vec e_i - \vec g_i)$
 for all $i \in [d-1]$.
\STATE Put $\vec V = (\vec v_1^+,\vec v_1^-,\vec v_2^+,\dots,
 \vec v_{d-1}^+, \vec v_{d-1}^-)^\T$,
 $\vec d = \vec R\vec e_d$.
\STATE Find $\vec x_0$ such that $\vec V\vec x_0 = \vec 1$.
\STATE If not $\vec A\vec x_0 < \vec 1$, restart the loop.
\STATE Solve $\sum_{i=1}^{d-1} \lambda_i\vec R(l\vec e_i + \vec g_i) = \vec c + \lambda_d\vec d$.
If $\lambda_d + \sum_{i=1}^{d-1}4\abs{\lambda_i} \leq 0$, restart the loop.
(This corresponds to $\vec x_0$ being optimal for $\vec c$.)
\STATE Follow the shadow path from $\vec d$ to $\vec c$ on
\begin{align*}
\max \vec c^\T&\vec x \\\tag{Unit LP'}
\vec A\vec x &\leq \vec 1 \\
\vec V\vec x &\leq \vec 1,
\end{align*}
starting from the vertex $\vec x_0$.
For the first pivot, follow the edge which is tight at the constraints in
$\vec B_0 = (\vec v_1^{\sgn(\lambda_1)},
 \dots,\vec v_{d-1}^{\sgn(\lambda_{d-1})})$. All other pivot steps are
 as in Algorithm~\ref{alg:shadow-simplex}.
\STATE  If (Unit LP') is unbounded, return ``unbounded''.
\STATE  If (Unit LP') is
 bounded and the optimal vertex $\vec{x}^*$ satisfies
 $\vec V\vec x^* < \vec 1$, return $\vec{x}^*$ as the
 optimal solution to (Unit LP).
\STATE  Otherwise, restart the loop.\;
\ENDLOOP
\end{algorithmic}
\end{algorithm}

We begin with some preliminary remarks for Algorithm~\ref{alg:upgraded-random-vertex}.
First, the goal of defining $\vec V$ is to create a new artificial LP, (Unit LP') $\max
\vec{c}^\T \vec x, \vec{A}\vec{x} \leq \vec{1}$, $\vec{V}\vec{x} \leq \vec{1}$, such
that $\vec{x_0}$ is a vertex of the corresponding system which maximizes $\vec{d}$.
On line 9 and 10, the algorithm checks if $\vec x_0$ is feasible and whether
it is not the optimizer of $\vec c$ on (Unit LP').
Having passed these checks, (Unit LP') is solved
via shadow vertex initialized at vertex $\vec{x}_0$ with objective $\vec{d}$.
An unbounded solution to (Unit LP') is always an unbounded
solution to (Unit LP).
Lastly, it is checked on line 13 whether the bounded solution (if it exists)
to (Unit LP') is a solution to (Unit LP).
Correctness of the algorithm's output is thus straightforward.
We do have to make sure that every step of the algorithm can be executed as described.
\begin{lemma}\label{lem:RVcorrectness}
In (Unit LP') as defined on lines 3-11 of Algorithm~\ref{alg:upgraded-random-vertex},
with probability $1$, $\vec x_0$ is well-defined, and, when entering the shadow
simplex routine, the point $\vec x_0$ is a shadow vertex and the edge defined by
$\vec B_0$ is a shadow edge on (Unit LP'). Moreover, $\vec x_0$ is the only degenerate
vertex.
\end{lemma}
\begin{proof}
Without loss of generality, we assume $\vec R = \vec I_{d\times d}$. With probability
$1$, the coefficients $\lambda_1,\dots,\lambda_d$ exist and are uniquely defined.

We now show that $\vec x_0$ is well-defined. Let $\vec x_0^+$ be the solution to
the following system of $d$ equalities
\begin{align}
\inner{\vec v_1^+}{\vec x_0^+} = 1, \inner{\vec v_2^+}{\vec x_0^+} = 1, \dots,
\inner{\vec v_{d-1}^+}{\vec x_0^+} = 1, \qquad \inner{8\vec e_d}{\vec x_0^+} = 2.
\end{align}
This system of equations almost surely has a single solution. We claim that
$\vec V\vec x_0^+=\vec 1$. By writing $\vec v_i^- = 8\vec e_d - \vec v_i^+$,
we find that $\inner{\vec v_i^-}{\vec x_0^+} = 1$ for all $i \in [d-1]$.
Therefore, $\vec x_0 = \vec x_0^+$ is indeed well-defined.

By definition, upon entering the shadow simplex routine, $\vec x_0$ satisfies
$\vec A\vec x < \vec 1$, $\vec V\vec x \leq \vec 1$ and is thus a vertex.

For all $t > 0$, define $\vec x_t$ to be the solution to $\inner{8\vec e_d}{\vec x_t} = 2-t,
\vec B_0\vec x_t = 1$. For any $\vec v_i^s \not\in \vec B_0$, we have
$\inner{\vec v_i^{s}}{\vec x_t} = 1-t < 1$. As the $\vec x_t$ lie on a line
and $\vec A\vec x_0 < \vec 1$, there exists some $\eps > 0$ such that $\vec x_t$ is
feasible for all $t \leq \eps$. Hence the constraints in $\vec B_0$ define an edge
of the feasible set.

The point $\vec x_0$ is tight at the inequalities $\vec V\vec x \leq \vec 1$, and
$\frac{1}{8d-8}\sum_{i=1}^{d-1}(\vec v_i^++\vec v_i^-) = \vec d$ is a corresponding
dual solution, so we know that $\vec x_0$ is optimal for objective $\vec d$ and thus
a shadow vertex.

Assume that $\vec x_0$ is not optimal for objective $\vec c$. One outgoing edge of
$\vec x_0$ is tight at the inequalities $\inner{\vec v_i^s}{\vec x}\leq 1$
for all $\vec v_i^s \in \vec B_0$ and that edge is on the shadow path exactly
if the cone spanned by $\vec B_0$ intersects $\cone(\vec c,\vec d)$ outside $\set{\vec 0}$.
This intersection is exactly the ray spanned by
\begin{align*}
\sum_{i=1}^{d-1}\abs{\lambda_i}\vec v_i^{\sgn(\lambda_i)}
&=  \sum_{i=1}^{d-1}\lambda_i(l\vec e_i + \vec g_i) + 4\abs{\lambda_i}\vec d \\
&= \vec c + \lambda_d \vec d + \sum_{i=1}^{d-1}4\abs{\lambda_i}\vec d,
\end{align*}
and we know that $\lambda_d + \sum_{i=1}^{d-1}4\abs{\lambda_i} > 0$
as otherwise we would have a certificate that
$\vec c \in \cone(\vec V\vec \lambda : \vec\lambda\in\R^d_+)$.
We conclude that $\sum_{i=1}^{d-1}\abs{\lambda_i}\vec v_i^{\sgn(\lambda_i)}$
is a non-negative linear combination of $\vec c,\vec d$ and
hence our description of the first shadow vertex pivot step is correct.

Lastly, we show that any vertex other than $\vec x_0$
is tight at exactly $d$ independently distributed constraint vectors.
Fix any basis $B$ such that there exists an $i \in [d-1]$ with
$\vec v_i^+,\vec v_i^- \in B$ and which does not define the vertex $\vec x_0$.
Let $\vec x_B$ be such that $\inner{\vec a}{\vec x_B} = 1$ for all $\vec a \in B$.
There exists some $j \in [d-1]$ such that both $\vec v_j^+,\vec v_j^- \notin B$,
for otherwise we would have $\vec x_B = \vec x_0$. We show that, almost surely,
$\inner{\vec v_j^+}{\vec x_B} > 1$ or $\inner{\vec v_j^-}{\vec x_B} > 1$,
which implies that $\vec x_B$ is almost surely not feasible.
We know that $\inner{\vec v_i^+}{\vec x_B} = \inner{\vec v_i^-}{\vec x_B} = 1$,
and hence $\inner{4\vec d}{\vec x_B} = 1$. It follows that
$\inner{\vec v_j^+}{\vec x_B} = 2 - \inner{\vec v_j^-}{\vec x_B}$,
The only way to have both $\inner{\vec v_j^+}{\vec x_B} \leq 1$ and
$\inner{\vec v_j^-}{\vec x_B} \leq 1$ would be if
$\inner{\vec v_j^+}{\vec x_B} = 1$. However, $\vec x_B$
and $\vec v_j^+$ are independently distributed and $\vec v_j^+$ has a continuous
probability distribution, so $\vec x_B$ is a vertex with probability $0$.
\end{proof}

To bound the expected running time of Algorithm~\ref{alg:upgraded-random-vertex},
we bound the expected number of pivot steps per iteration of the loop,
and the expected number of iterations of the loop.

First, we bound the expected shadow size in a single iteration. Because the
constraint vectors $\vec v_i^+,\vec v_i^-$ are not independently distributed for
any $i \in [d-1]$, we are unable to apply Theorem~\ref{thm:gaussian} in a
completely black-box way. As we show below, in this new setting, the proof of
Theorem~\ref{thm:gaussian} still goes through essentially without modification.

In the rest of this section, we abbreviate
\[\conv(\vec A, \vec V) := \conv(\vec a_1,\dots,\vec a_n,\vec v_1^+,\dots,\vec v_{d-1}^+, \vec v_1^-,\dots,\vec v_{d-1}^-).\]

\begin{lemma}\label{lem:unpackblackbox} Let $\vec A$ have independent
standard deviation $\sigma$ Gaussian rows with centers of norm at most $1$
and let $\vec V$ be sampled, independently from $\vec A$, as in lines
4-7 of Algorithm~\ref{alg:upgraded-random-vertex} with $l \leq 1$.
The shadow size $\E[\abs{\edges(\conv(\vec A, \vec V) \cap \linsp(\vec c,\vec d))}]$ is bounded by
$\calD_g(d,n+2d-2,\min(\sigma,\bar{\sigma})/5)+1$.
\end{lemma}
\begin{proof}
We fix the choice of $\vec R$. The distribution of constraint vectors is now
independent of the two-dimensional plane.
\[
\E[\abs{\edges(\conv(\vec A, \vec V)\cap\linsp(\vec c,\vec d))}]
\leq
\max_{\vec R} \E[\abs{\edges(\conv(\vec A, \vec V)\cap\linsp(\vec c,\vec R\vec e_d))}\mid \vec R].
\]

The rows of $\vec A$ have centers of norm at most $1$ and the rows
of $\vec V$ have centers of norm at most $4+l \leq 5$. After an appropriate
rescaling, we can assume all $n+2d-2$ constraints have expectations
of norm at most $1$ and
standard deviation $\sigma \leq \bar{\sigma}/5$.

To get the desired bound, we bound the number of edges other than
the one induced by $\vec x_0$,
$W \cap \set{\vec y \in \R^d : \inner{\vec y}{\vec x_0} = 1}$,
which yields the $+1$ in the final bound. The proof is essentially
identical to that of Theorem~\ref{thm:gaussian}, i.e.~we bound the ratio of the
expected perimeter divided by the minimum expected edge of the polar polygon. We
sketch the key points below. Firstly, notice that the perimeter bound in
Lemma~\ref{lem:perimeterbound} does not require independence of the
perturbations, so it still holds. For the minimum edge length, we
restrict to the bases $B$ as in Lemma~\ref{lem:edgesasfraction} (which also does
not require independence) after removing those which induce $\vec x_0$ as a
vertex (it has already been counted). By Lemma~\ref{lem:RVcorrectness}, the
remaining bases in $B$ contain at most one of each pair $\set{\vec v_i^-,\vec
v_i^+}$, $i \in [d-1]$, since bases containing two such vectors correspond to an
edge different from the one induced by $\vec x_0$ with probability $0$. In
particular, every basis in $B$ consists of only independent random vectors.

From here, the only remaining detail for the bound to go through is to
to check that the conclusion of Lemma~\ref{lem:conditiononhyperplane} still
holds, i.e., that the position of vectors within their containing
hyperplane does not affect the probability that these vectors form
a facet of the convex hull. Without loss of generality,
we consider the vectors $\vec a_1,\dots,\vec a_i$, $\vec v_1^+,\dots,\vec v_j^+
$ with $i+j=d$. Define $\vec\theta\in\bbS^{d-1},t\geq 0$
by $\inner{\vec\theta}{\vec a_k}=t$ for all $k \in [i]$, $\inner{\vec\theta}{\vec v_k^+}=t$
for all $k\in[j]$.
The set $\conv(\vec a_1,\dots,\vec a_i,\vec v_1^+,\dots,\vec v_j^+)$
is a facet of the convex hull of the constraint vectors when either (1)
$\inner{\vec\theta}{\vec a_k} < t$ for all $k > i$,
$\inner{\vec\theta}{\vec v_k^-} < t$ for all $k \in [j]$ and $\inner{\vec\theta}{\vec v_k^\pm} < t$
for all $k > j$
or (2) when
$\inner{\vec\theta}{\vec a_k} > t$ for all $k > i$,
$\inner{\vec\theta}{\vec v_k^-} > t$ for all $k \in [j]$ and $\inner{\vec\theta}{\vec v_k^\pm} > t$
for all $k > j$.
The only one of these properties that is not
independent of $\vec a_1,\dots,\vec a_i, \vec v_1^+,\dots,\vec v_j^+$ is whether
$\inner{\vec\theta}{\vec v_k^-} < t$ or $\inner{\vec\theta}{\vec v_k^-} > t$
for $k \in [j]$, but we know that
$\inner{\vec\theta}{\vec v_k^-} = 8\inner{\vec\theta}{\vec d} -
\inner{\vec\theta}{\vec v_k^+} = 8\inner{\vec\theta}{\vec d} - t$ for all $k \in
[j]$, and so the value $\inner{\vec\theta}{\vec v_k^-}$ does not depend on the
positions of $\vec a_1,\dots,\vec a_i, \vec v_1^+,\dots,\vec v_j^+$ within their
containing hyperplane.  We conclude that the expected number of edges is bounded by
$\calD_g(d,n+2d-2, \min(\sigma,\bar \sigma)/5)+1$.
\end{proof}

All that is left, is to show that the success probability of each
loop is lower bounded by a constant.

\begin{definition}
For a matrix $\vec M \in \R^{d \times d}$, we define its operator norm by
\[\norm{\vec M } = \max_{\vec x \in \R^d \setminus \set{\vec 0}} \frac{\norm{\vec M \vec x}}{\norm{\vec x}}\]
and its maximum and minimum singular values by
\[ s_{\max}(\vec M) = \norm{\vec M}, \qquad s_{\min}(\vec M) = \min_{\vec x \in
\R^d \setminus \set{\vec 0}}\frac{\norm{\vec M\vec x}}{\norm{\vec x}}.\]
\end{definition}

Using the Gaussian tailbound \eqref{eq:gauss-full-d} together with a $1/2$-net
on the sphere (which has size at most $8^d$, see e.g., \cite{matouvsek2002lectures}, page 314),
we immediately obtain the following tail bound for the operator
norm of random Gaussian matrices.

\begin{lemma}\label{lemma:gaussian-matrix-norm}
For a random $d \times d$ matrix $G$ with independent standard normal entries, one has
\[\Pr[\norm{\vec G} > 2 t \sqrt{d}] \leq 8^d e^{-d(t-1)^2/2}.\]
\end{lemma}

\begin{lemma}\label{lem:pr-feasible}
Let $\vec A \in \R^{n\times d}$ have rows of norm at most $2$
and $\sigma \leq \frac{l}{6\sqrt{d}}$. For $\vec x_0$ sampled
as in lines 4-8 of Algorithm~\ref{alg:upgraded-random-vertex},
with probability at least $0.98$, the point $\vec x_0$
satisfies $\vec A\vec x_0 < \vec 1$.
\end{lemma}
\begin{proof} Without
loss of generality, we assume $\vec R = \vec I_{d\times d}$.
We claim that, with sufficient probability,
$\norm{\vec x_0-\vec e_d/4} < 1/4$. Together with the triangle
inequality and the assumption that $\norm{\vec a_i} \leq 2$ for all
$i\in[n]$, this suffices to show $\vec A\vec x_0 < \vec 1$.

Elementary calculations show that $\vec x_0-\vec e_d/4$ satisfies
$\inner{\vec e_d}{(\vec x_0-\vec e_d/4)}=0$ and, for every $i \in [d-1]$,
$\inner{(l\vec e_i + \vec g_i)}{(\vec x_0-\vec e_d/4)} = -\inner{\vec g_i}{\vec e_d}/4$.
Let $\vec G$ be the matrix with rows consisting of the first $d-1$ entries of each
of $\vec g_1,\dots,\vec g_{d-1}$, and $\vec g$ be the vector consisting of
the $d$'th entries of $\vec g_1,\dots,\vec g_{d-1}$. From the above equalites we derive
\begin{align*}
\begin{pmatrix}
l\vec I_{d-1}+\vec G & \vec g \\ \vec 0^\T & 1
\end{pmatrix}(\vec x_0 - \vec e_d/4) &= \frac{1}{4}
\begin{pmatrix}-\vec g\\0\end{pmatrix}\\
\begin{pmatrix}
l\vec I_{d-1}+\vec G & \vec 0 \\ \vec 0^\T & 1
\end{pmatrix}(\vec x_0 - \vec e_d/4) &=
\frac{1}{4}\begin{pmatrix}-\vec g\\0\end{pmatrix} \\
\vec x_0 - \vec e_d/4 &= \frac{1}{4}
\begin{pmatrix}-\begin{pmatrix}
l\vec I_{d-1}+\vec G
\end{pmatrix}^{-1}\vec g\\0\end{pmatrix}.
\end{align*}
Note that the matrix is almost surely invertible. We abbreviate
$\vec M = l\vec I_{d-1}+\vec G$
and bound $\norm{\vec x_0 - \vec e_d/4} \leq \norm{\vec M^{-1}}\norm{\vec g}/4$.
Using that $\sigma \leq \frac{l}{6\sqrt{d}}$, we apply \eqref{eq:gauss-full-d}
to get $\norm{\vec g} \leq l/2$ with probability at least $0.99$.

The operator norm of the inverse matrix satisfies
$\norm{\vec M^{-1}} = \frac{1}{s_{\min}(l\vec I + \vec G)}$,
and by the triangle inequality we derive
\[ s_{\min}(l\vec I + \vec G) \geq s_{\min}(l\vec I) - s_{\max}(\vec G) = l - s_{\max}(\vec G).\]
By Lemma~\ref{lemma:gaussian-matrix-norm}, we have
$\norm{\vec G} \leq 3\sqrt{d}\sigma \leq l/2$ with
probability at least $0.99$. Putting the pieces together, we conclude that
\[\frac{1}{4}\norm{\vec M^{-1}}\norm{\vec g} \leq
\frac{1}{4}\cdot\frac{1}{l-l/2}\cdot\frac{l}{2} \leq 1/4.\]
We take the union bound over the two bad events and thus conclude that
$\inner{\vec a_i}{\vec x_0} \leq \norm{\vec a_i}\norm{\vec x_0} < 1$
for all $i\in[n]$ with probability at least $0.98$.
\end{proof}

Lastly, we need to prove that the conditionals on lines 10, 12 and 13 of
Algorithm~\ref{alg:upgraded-random-vertex} succeeds with sufficient
probability.

\begin{lemma}[Adapted from \cite{jour/siamjc/Vershynin09}]\label{lem:pr-admissible}
Let $l \leq 1/6\sqrt{\log d}$ and $\sigma \leq 1/8\sqrt{d \log d}$.
For fixed $\vec A$ and $\vec V$ sampled as in lines
4-7 of Algorithm~\ref{alg:upgraded-random-vertex},
let $\vec x^*$ be the optimal solution to (Unit LP') if it exists.
With probability at least $0.24$,
(Unit LP) being unbounded implies that (Unit LP') is unbounded
and (Unit LP) being bounded implies $\vec V\vec x^* < \vec 1$.
\end{lemma}
\begin{proof}
Let $\vec x$ be the maximizer of (Unit LP) if it exists, or otherwise
a generator for an unbounded ray in (Unit LP), and let
$\vec \omega = \vec x/\norm{\vec x}$. We aim to prove that
$\vec V\vec \omega < \vec 0$ with probability at least $0.24$ over the randomness
in $\vec V$, which is sufficient for the lemma to hold.

We fix $\vec A$, and hence $\vec \omega$ as well. We decompose
\begin{align}\label{eq:numbspace}
\inner{\vec v_i^+}{\vec \omega} = 4\inner{\vec d}{\vec \omega} +
\inner{(l\vec R\vec e_i)}{\vec \omega} +
\inner{(\vec R\vec g_i)}{\vec \omega},
\end{align}
for all $i\in[d-1]$ and similarly for $\vec v_i^-$,
and we will bound the different terms separately.

The inner product $\inner{\vec d}{\vec \omega}$ has probability density
proportional to $\sqrt{1-t^2}^{d-3}$, as it is the one-dimensional marginal distribution
over the sphere $\bbS^{d-1}$ (see e.g., \cite{fang1990symmetric}, equation 1.26).
which can differ over the interval $[-\sqrt{\frac{2}{d-1}},\sqrt{\frac{2}{d-1}}]$
by at most a factor $1/e$.
We lower bound the probability that $\inner{\vec d}{\vec \omega}$ is far from being positive:
\begin{align*}
\Pr[\inner{\vec d}{\vec \omega} < -\frac{1}{4}\sqrt{\frac{2}{d-1}}] &= \frac{1}{2}\Pr[\inner{\vec d}{\vec \omega} < -\frac{1}{4}\sqrt{\frac{2}{d-1}} \mid \inner{\vec d}{\vec \omega} \leq 0] \\
&\geq \frac{1}{2}\Pr[\inner{\vec d}{\vec \omega} < -\frac{1}{4}\sqrt{\frac{2}{d-1}} \mid
\inner{\vec d}{\vec \omega} \in [-\sqrt{\frac{2}{d-1}}, 0]] \\
&\geq \frac{1}{2}\cdot \frac{\frac{3}{4e}}{\frac{3}{4e}+\frac{1}{4}} \\
&\geq 0.26.
\end{align*}
Hence, for $\vec d$ a randomly chosen unit vector independent of $\omega$, we have
$4\inner{\vec d}{\vec \omega} < -\sqrt{\frac{2}{d-1}}$ with probability at least $0.26$.
Now we will give an upper bound on the second and third terms in
\eqref{eq:numbspace} with sufficient probability.

By the same measure concentration argument as in the proof of \eqref{eq:laplace-1d}
we know that
$\Pr[\abs{\vec e_i^\T\vec R^\T\vec\omega} > t/\sqrt{d-1}] \leq e^{-t^2/2}$.
We apply the above statement with $t = 3\sqrt{\log d}$ and find that
\[\abs{l\vec e_i^\T\vec R^\T\vec\omega} < tl/\sqrt{d-1} \leq 1/2\sqrt{d-1}\]
with probability at least $1-\frac{0.01}{d}$.

For the last part, fix $\vec R = \vec I$ without loss of generality. The inner product
$\inner{\vec g_i}{\vec \omega}$ is $N(0,\sigma^2)$ distributed, hence
$\Pr[\abs{\inner{\vec g_i}{\vec \omega}} < 4\sigma\sqrt{\log d}] \geq 1-\frac{0.01}{d}$
by standard Gaussian tail bounds. Recall that $4\sigma\sqrt{\log d} \leq 1/2\sqrt{d-1}$.

Putting it all together, we take the union bound over the three terms in \eqref{eq:numbspace}
and all $\vec v_i^+,\vec v_i^-$ with $i \in [d-1]$
and find that $\inner{\vec v_i^+}{\vec \omega} < 0$
and $\inner{\vec v_i^-}{\vec \omega} < 0$ for all $i\in[d-1]$ with probability
at least $0.26-(d-1)\frac{0.01}{d} - (d-1)\frac{0.01}{d} \geq 0.24$.
\end{proof}

\begin{theorem}
\label{thm:success-bnd}
For $\sigma \leq \bar{\sigma}$,
Algorithm~\ref{alg:upgraded-random-vertex} solves (Unit LP) in at most an expected
$6+5\calD_g(d,n+2d-2,\sigma/5)$ number of shadow vertex pivots.
\end{theorem}
\begin{proof}
Let $\vec{a}_1,\dots,\vec{a}_n \in \R^d$ denote the rows of $\vec{A}$, where we
recall that the centers $\vecb{a}_i := \E[\vec{a}_i]$, $i \in [n]$, have norm at
most $1$.

\paragraph{Pivots from line 1} Let $L$ denote the event that the rows of
$\vec{a_1},\dots,\vec{a_n}$ all have norm at most $2$. Noting that each $\vec{a}_i$, $i \in [n]$,
is a variance $\sigma^2$ Gaussian and $1/\sigma \geq 5\sqrt{d \log n}$, by
Lemma~\ref{lem:gaussian-tails} (Gaussian concentration), we have that
\begin{align*}
\Pr[L^c] &= \Pr[\exists i \in [n]: \|\vec{a}_i\| \geq 2]
\leq n \Pr[\|\vec{a}_1-\vecb{a}_1\| \geq 1] \\
&\leq n \Pr[\|\vec{a}_1-\vecb{a}_1\| \geq 5 \sqrt{d \log n} \sigma]
\leq e^{-(d/2)(5\sqrt{\log n}-1)^2} \leq n^{-d} \text{ .}
\end{align*}
Therefore, the simplex run on line 1 is executed with
probability at most $n^{-d}$ incurring at most $n^{-d} \binom{n}{d} \leq 1$
pivots on expectation.

\paragraph{Pivots from the main loop}
Let $\vec V_1,\vec V_2,\dots$ be independent samples of $\vec V$ as described
in lines 3-7 of Algorithm~\ref{alg:upgraded-random-vertex}.
Define the random variable $N = N(A, \vec V_i : i \in \N) \geq 0$ as the number
of iterations of the main loop if Algorithm~\ref{alg:upgraded-random-vertex}
were run on input $\vec A,\vec c$ and the value of $\vec V$ in iteration $i$ equals
$\vec V_i$. Note that $N = 0$ exactly if $L^c$.
Note that the value of $\vec V_i$ unique specifies the value of $\vec d_i$.
Define the event $F_i$ that the checks on lines 9 and 10 would pass
on data $\vec V_i$.
Lastly, let $P(\vec A,\vec V_i)$ denote the number
of pivot steps that an iteration of the main loop would perform on the data
$\vec A,\vec V_i$. In particular, $P(\vec A,\vec V_i) > 0$ exactly when $L$ and $F_i$.

The total number of pivot steps is given by the expectation
\begin{align*}
\E[\sum_{k=1}^N P(\vec A,\vec V_k)]
&=
\E[\sum_{k=1}^\infty P(\vec A,\vec V_k) \bbone[N \geq k]] \\
&=
\sum_{k=1}^\infty \E[P(\vec A,\vec V_k) \bbone[N \geq k]].
\end{align*}
For any $k$, the event $N \geq k$ depends solely on $\vec V_1,\dots,\vec V_{k-1}$, hence
we get
\begin{align*}
\sum_{k=1}^\infty \E[P(\vec A,\vec V_k) \bbone[N \geq k]]
&=
\sum_{k=1}^\infty \E_{\vec A, \vec V_k}[P(\vec A,\vec V_k) \E_{\vec V_1,\dots,\vec V_{k-1}}[\bbone[N \geq k\mid \vec A]] \\
&=
\sum_{k=1}^\infty \E[P(\vec A,\vec V_k) \Pr[N \geq k \mid \vec A]] \\
&=
\sum_{k=1}^\infty \E[P(\vec A,\vec V_k) \Pr[N > 1 \mid \vec A]^{k-1}],
\end{align*}
where the last line follows from the observation that the
seperate trials are independent when $\vec A$ is fixed.
When $\vec A$ is such that $L^c$ holds, then $\Pr[N > 1 \mid \vec A] = 0$.
Now we appeal to
Lemma~\ref{lem:pr-feasible}, Lemma~\ref{lem:pr-admissible}.
The first shows that the Algorithm~\ref{alg:upgraded-random-vertex}
does not restart on line 9 with probability at least $0.98$
and the second shows that the algorithm
does not restart on lines 10 and 14 with probability at least $0.24$.
By the union bound, this implies that $\Pr[N > 1 | \vec A] \leq 1-0.22$
for any $\vec A$ such that $L$ holds. Hence we get
\begin{align*}
\sum_{k=1}^\infty \E_{\vec A, \vec V_k}[P(\vec A,\vec V_k) \Pr[N > 1 \mid \vec A]^{k-1}]
&\leq
\sum_{k=1}^\infty \E[P(\vec A,\vec V_k) (1-0.22)^{k-1}] \\
&=
\frac{1}{0.22} \E[P(\vec A,\vec V_1)].
\end{align*}
The number of pivot steps $P(\vec A,\vec V_1)$ is nonzero exactly when
$L$ and $F_1$ hold, and is always bounded by the shadow size
according to Theorem~\ref{thm:geometric-characterization-shadow-path-size}.
We bound this quantity using Lemma~\ref{lem:unpackblackbox} and get
\begin{align*}
\frac{1}{0.22} \E[P(\vec A,\vec V_1)]
&\leq
5 \E[\bbone_{F_1\cap L}\abs{\edges(\conv(\vec A,\vec V_1) \cap \linsp(\vec c,\vec d_1))}] \\
&\leq
5 \E[\abs{\edges(\conv(\vec A,\vec V_1) \cap \linsp(\vec c,\vec d_1))}] \\
&\leq
5\calD_g(d,n+2d-2,\min(\sigma,\bar\sigma)/5)+5.
\end{align*}

\paragraph{Final Bound} Combining the results from the above paragraphs, we
get that the total expected number of simplex pivots in
Algorithm~\ref{alg:upgraded-random-vertex} is bounded by:
\[
\Pr[L^c] \binom{n}{d} + \E[\sum_{k=1}^N P(\vec A,\vec V_k)]
\leq 6 + 5 \calD_g(d,n+2d-2,\sigma/5) \text{ ,}
\]
as needed.
\end{proof}